\documentclass[11pt, letterpaper]{article}

\usepackage{arxiv}

\usepackage{times}
\usepackage{amsmath,commath,amsthm}
\usepackage{amssymb}        
\usepackage{amsfonts}       
\usepackage{graphicx}       
\usepackage{subcaption}     
\usepackage{multirow,mathrsfs}
\usepackage{adjustbox,bbm,mathdots,mathtools}
\usepackage{algorithm,algorithmicx,algpseudocode}
\usepackage{listings}
\usepackage{tabularray,pifont,makecell,bbold,cellspace,tabularx}
\usepackage{lmodern,anyfontsize,textcomp}
\usepackage{picture,xcolor,color} 

\usepackage{tikz,tikz-3dplot}
    \usetikzlibrary{arrows,patterns} 
    \usetikzlibrary{arrows.meta,graphs,quotes,shadows,backgrounds} 
    \usetikzlibrary{tikzmark,intersections,positioning,plotmarks}
    \usetikzlibrary{automata,shapes,calc,through,fadings}
    \usetikzlibrary{matrix,shapes.geometric,fit,shapes.symbols,chains}
    \usetikzlibrary{decorations.pathreplacing,decorations.pathmorphing}
    \usetikzlibrary{decorations.text,fadings,decorations.markings}

\tikzstyle{rec} = [rectangle, text width=1cm, minimum   height=1cm,text centered, draw=black, fill=gray!25]
\tikzstyle{smrec} = [rectangle, text width=.5cm, minimum   height=.5cm,text centered, draw=black, fill=gray!10]
\tikzstyle{circ} = [circle, text width=1cm, minimum   height=1cm,text centered, draw=black, fill=gray!0, dashed]
\tikzstyle{arrow} = [thick,->, draw=blue,>=stealth]
\tikzstyle{cross} = [cross/.style={path picture={ 
  \draw[black]
(path picture bounding box.south east) -- (path picture bounding box.north west) (path picture bounding box.south west) -- (path picture bounding box.north east);
}}]

\newtheorem{theorem}{Theorem}
\newtheorem{proposition}{Proposition}
\newtheorem{lemma}{Lemma}
\newtheorem{corollary}{Corollary}
\newtheorem{example}{Example}%

\usepackage{eqparbox}

\newcommand{\cmark}{\ding{51}}%
\newcommand{\xmark}{\ding{55}}

\usepackage{etoolbox}  
\makeatletter
\patchcmd{\algorithmic}{\addtolength{\ALC@tlm}{\leftmargin} }{\addtolength{\ALC@tlm}{\leftmargin}}{}{}
\makeatother

\algnewcommand{\algorithmicforeach}{\textbf{for each}}
\algdef{SE}[FOR]{ForEach}{EndForEach}[1]
  {\algorithmicforeach\ #1\ \algorithmicdo}% \ForEach{#1}
  {\algorithmicend\ \algorithmicforeach}% \EndForEach
  
\theoremstyle{thmstyleone}%
\newtheorem{definition}{Definition}%

\usepackage[colorlinks=true, linkcolor=teal, citecolor=yaleblue]{hyperref}

\newcommand{\Cross}{\mathbin{\tikz [x=1.5ex,y=1.5ex,line width=.3ex] \draw (0,0) -- (1,1) (0,1) -- (1,0);}}%

\DeclareMathOperator*{\argmax}{arg\,max}

\DeclareMathOperator\supp{supp}

\newcommand{\tuple}[1]{\ensuremath{\langle #1 \rangle}}

\newcommand{\mGdefo}{\ensuremath{\tuple{G^{0}, G^{1}, \ldots, G^{\vert N \vert}}}}

\newcommand{\mGt}[1]{\ensuremath{mG^{(#1)}}}

\newcommand{\AD}[2]{\ensuremath{{\mathcal AD}^{#1} \boldsymbol{\left(\right.}#2\boldsymbol{\left.\right)}}}

\newcommand{\LAD}[1]{\ensuremath{{{\mathfrak L}_{\mathcal AD}} \left(#1\right)}}

\newcommand{\stirlingii}{\genfrac{\{}{\}}{0pt}{}}

\newcommand{\infl}      {\blacktriangleleft}    % Inflation

\newcommand{\comp}      {\bowtie}               % Compatible strategy profiles

\definecolor{pastelred}{rgb}{1.0, 0.41, 0.38}
\definecolor{pastelgreen}{rgb}{0.47, 0.87, 0.47}
\definecolor{pastelyellow}{rgb}{0.99, 0.99, 0.59}
\definecolor{pastelblue}{rgb}{0.68, 0.78, 0.81}
\definecolor{forestgreen(web)}{rgb}{0.13, 0.55, 0.13}
\definecolor{upsdellred}{rgb}{0.68, 0.09, 0.13}
\definecolor{yaleblue}{rgb}{0.06, 0.3, 0.57}
\definecolor{cambridgeblue}{rgb}{0.64, 0.76, 0.68}
\definecolor{beige}{rgb}{0.96, 0.96, 0.96}

\newlength{\subcolumnwidth}

\newcommand{\nextsubcolumn}[1][]{%
  \cr\noalign{\hfill}
  \if\relax\detokenize{#1}\relax\else\hsize=#1\setlength{\subcolumnwidth}{\hsize}\fi
}

\newcommand{\problem}[4]{
	\begin{table*}[h!]
	\centering
	\begin{tabular}{ | m{0.12\textwidth} | m{0.8\textwidth} | }
 		\hline
 		\multicolumn{2}{|l|}{\textsc{#1} \label{#2}}\\
 		\hline
 		\textbf{Input:}  & #3\\\hline
 		\textbf{Output:} & #4\\
		\hline
	\end{tabular}
	\end{table*}
}

\begin{document}
%===============================================================
\title{Adaptation Procedure in Misinformation Games}

\author{
    Konstantinos Varsos${}^{\dagger}$ \\
    \texttt{varsosk@ics.forth.gr}
    \And
    Merkouris Papamichail${}^{\dagger,\S}$\\
    \texttt{mercoyris@ics.forth.gr}
    \And
    Giorgos Flouris${}^{\dagger}$ \\
    \texttt{fgeo@ics.forth.gr}
    \And
    Marina Bitsaki${}^{\S}$\\
    \texttt{ecbitsaki@gmail.com}
    \And
    {}\\
    ${}^{\dagger}$Institute of Computer Science, Foundation for Research and Technology-Hellas (FORTH)\\
    ${}^{\S}$Computer Science Department, University of Crete\\
}

\date{}
%===============================================================

\maketitle

\begin{abstract}
We study interactions between agents in multi-agent systems, in which the agents are misinformed with regards to the game that they play, essentially having a subjective and incorrect understanding of the setting, without being aware of it. For that, we introduce a new game-theoretic concept, called misinformation games, that provides the necessary toolkit to study this situation. Subsequently, we enhance this framework by developing a time-discrete procedure (called the Adaptation Procedure) that captures iterative interactions in the above context. During the Adaptation Procedure, the agents update their information and reassess their behaviour in each step. We demonstrate our ideas through an implementation, which is used to study the efficiency and characteristics of the Adaptation Procedure.
\end{abstract}


\section{Introduction}\label{sec:introduction}

The importance of multi-agent systems has been heavily acknowledged by the research community and industry. A multi-agent system is a system composed of multiple interacting autonomous, self-interested, and intelligent agents\footnote{Note that we use the terms ``agent'' and ``player'' interchangeably throughout this paper.}, and their environment. 
In such settings, agents need to be incentivized to choose a plan of action, and game theory~\cite{Algorithmic_Game_Theory_book} provides a suitable framework for analyzing these interactions.

A usual assumption of game theory is that the game specifications (i.e., the rules of interaction, which include the number of players, their strategies and the expected payoff for each strategic choice) are common knowledge among the players. In other words, players have correct (although not necessarily complete) information regarding the game. 

However, in realistic situations, a reasoning agent is often faced with erroneous, misleading and unverifiable information regarding the state of the world or the possible outcomes of her actions (see games with misperception \cite{RaiffanLuce,BENNETT1980489,Feinberg2020}). 
This \emph{misinformation} may affect the interaction's outcomes.

Misinformation in an interaction can occur due to various reasons. For example, the designer may choose to deceive (some of) the agents in order to obtain improved social outcomes or for other reasons (this falls under the general area of mechanism design, see \cite{VFFB}). 
Or an agent may communicate deceptive information in order to lead the other agents to suboptimal choices and obtain improved outcomes (e.g., fake financial reports misleading investors \cite{accounting-frauds}). 
Furthermore, when the agents operate in a remote and/or hostile environment, noise and other random effects may distort communication, causing agents to receive a game specification that is different from the intended one (e.g., in the case of autonomous vehicles operating in Mars \cite{Brown2017AreMS}). 
Another scenario that could lead to misinformation refers to cases where the environment changes without the agents' knowledge, causing them to have outdated information regarding the rules of interaction (e.g., an accident causing a major and unexpected disruption in the flow of different roads in a city).
Last but not least, endogenous reasons (e.g., limited awareness, bounded computational capacity, cognitive restrictions, biases \cite{DBLP:conf/aaai/MeirP15} etc.) may cause agents to misinterpret the situation and assign their own (mistaken) payoffs to different actions.

Thereupon, some aspects of the situation, or the modeling and reasoning regarding the situation, leads players to incorporate a possibly incorrect viewpoint of the real aspects. Thus, they may miss crucial specifications, and interact relying on a restricted and incorrect perception of the situation. 
The key characteristic of the described scenario is that the players do not question the rules of interaction given to them; this differentiates this scenario from standard settings of games with other forms of uncertainty (such as Bayesian games \cite{Harsanyi:1967:GII:3218759.3218761,Zamir2009}, uncertainty theory \cite{Gao2013} etc.), in which players are well-aware of the fact that the information given to them is incomplete, uncertain or flawed in various ways, and this knowledge is incorporated in their reasoning (see also Figure \ref{fig: mG vs Bayesian}).
In other words, in misinformation games \emph{agents don't know that they don't know}, as opposed to incomplete or imperfect games where the \emph{agents know that they don't know}.

The study of settings where any participant has (possibly) wrong knowledge with regards to the real situation, has been recently considered in defining the concept of \emph{misinformation games} \cite{VFBF,VFFB}; this work falls under the hood of games with misperception \cite{RaiffanLuce,BENNETT1980489,Feinberg2020}.
In misinformation games, each player has a subjective view of the abstract game's specifications, that may not coincide with the specifications of the real interaction, modelling the fact that agents may operate under an erroneous specification.
For that, the proposed method agglomerates both the real situation (\emph{actual game}), and the subjective (misinformed) views of the players.
A key characteristic of misinformation games is that the equilibrium (i.e., the set of strategic choices where no player wants to deviate from) is determined by the subjective views of the players, rather than the actual game \cite{VFBF}. This is called the \emph{natural misinformed equilibrium} in \cite{VFBF}. On the contrary, the actual payoffs received by the players is determined by the actual game, and may be different than the expected ones.

Given the discrepancy between the actual and the perceived payoffs, a natural question is how the players will react upon their realisation that the received payoffs are different than expected.
However, the approach presented in \cite{VFBF,VFFB} focuses on one-shot interactions, and thus misses this crucial and interesting part of the problem.
To address this limitation, we \emph{incorporate in misinformation games an iterative time-discrete methodology, called the Adaptation Procedure, which models the evolution of the strategic behaviour of rational players in a misinformation game, as they obtain new information and update their (erroneous) game specifications}. The first steps towards this direction were held in \cite{PVF}; here we extend this approach.

More specifically, at each time point, players choose a strategic action (using the methodology of \cite{VFBF,VFFB}), receive the corresponding payoffs, and engulf them into their game specifications. The procedure is then repeated. 
Importantly, the information received in each time point may lead players to a different choice in the next time point (because they now operate under a different payoff matrix), so the procedure is iterative and stabilises when the players have no incentives to deviate from their current choices, based on what they know so far (which may or may not coincide with the actual game). 
When the procedure stabilizes we reach a refinement of the natural misinformed equilibrium, that we call the \emph{stable misinformed equilibrium}.

\begin{figure*}[t!]
\centering
  \begin{subfigure}[t]{0.475\textwidth}
  \centering
  \includegraphics[height=0.42\textwidth]{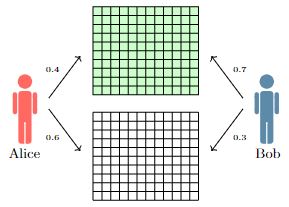}
    \caption{Bayesian game.}\label{fig:Bayesian games}
        \end{subfigure}
        \medskip
        \begin{subfigure}[t]{0.475\textwidth}  
\centering
\includegraphics[height=0.42\textwidth]{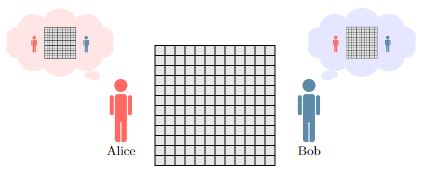}
    \caption{misinformation game.}\label{fig:mG graph}
        \end{subfigure}
        \caption[ Example of games without common and correct information. ]
        {Example of games without common and correct information.} 
        \label{fig: mG vs Bayesian}
\end{figure*}

The structure of this paper follows. In Section~\ref{sec:nfg} we present some preliminaries, including the concept of misinformation games. Then, in Section~\ref{sec:adaptation procedure} we introduce the Adaptation Procedure process, and in Section~\ref{sec:computing-adaptation-procedure} we present an algorithmic analysis and tools to compute the Adaptation Procedure and provide some experimental evaluation.
In Section~\ref{sec:related work} we review the literature related to this paper and conclude in Section~\ref{sec:conclusions}. 
Finally, we provide Appendix~\ref{appendix:omitted proofs} where we present all the proofs of our results.

This paper is an extended and revised version, combining and extending previous works of ours \cite{PVF,VFBF,VFFB}. More specifically, the main new contributions of the paper, compared to the previous ones, are the following:
\begin{itemize}
    \item we elaborate on and describe in detail the inflation process described superficially in \cite{VFBF} (see Propositions~\ref{prop:inflate-cons-hi}-\ref{prop:infl-algo-full}, and Algorithms \ref{algo:inflate_game}, \ref{algo:addplayer_inflate_game}, \ref{algo:addstrategy_inflate_game} in Subsection~\ref{subsec:inflation process-nfg}, as well as Algorithms \ref{algo:addgame}, \ref{algo:inflate_mG} in Subsection~\ref{sec:misinformation games});

    \item we provide new theoretical and computational results regarding the adaptation procedure (see Subsection~\ref{sec:adaptation-procedure: properties});

    \item we provide novel implementations for the Adaptation Procedure (including a parallel one), and new experimental results (see Section \ref{sec:computing-adaptation-procedure}).

\end{itemize}

\section{Normal-form misinformation games}\label{sec:nfg}

In this section we describe misinformation games; we first present a quick review of normal-form games.

\subsection{Normal-form games}\label{subsec:preliminaries-nfg}

In this paper, we consider \emph{normal-form games} $G$. A normal-form game includes a set of players $N$. Each player $i \in N$ has a set of pure strategies that is denoted by $S_{i}$, and the Cartesian product of all $S_{i}$s is the set of joint decisions of the players, which is denoted by $S$. Finally, each player $i$ has a payoff matrix $P_i \in \mathbb{R}^{ \vert S_1 \vert \times \ldots \times \vert S_{\vert N \vert} \vert }$, assigning a payoff to that player depending on the joint strategy decisions of all players. The payoff matrix of the game is $P = (P_1, \dots, P_{\vert N \vert})$. We follow the usual notation and we denote $G$ as a tuple $\langle N, S, P\rangle$.

If player $i$ randomly selects a pure strategy, then she plays a mixed strategy $\sigma_{i} = (\sigma_{i,1}, \ldots, \sigma_{i,\vert S_i \vert})$ which is a discrete probability distribution over $S_{ i}$. 
Let $\Sigma_{ i}$ be the set of all possible mixed strategies $\sigma_{ i}$. 
A strategy profile $\sigma = (\sigma_{ 1}, \ldots, \sigma_{ \vert N \vert})$ is an $\vert N \vert$-tuple in $\Sigma = \Sigma_{ 1} \times \ldots \times \Sigma_{ \vert N \vert}$. We denote by $\sigma_{ -i}$ the joint decision of all other players except for player $i$ in $\sigma$. The payoff function of player $i$ is defined as: $h_{ i}: \Sigma \to \mathbb{R}$, such that: 
\begin{equation}\label{eq:payoff_function}
h_{i}(\sigma) = \sum\limits_{k \in S_1}\dots \sum\limits_{j\in S_{\vert N \vert}}P_i(k, \dots, j) \cdot \sigma_{1, k} \cdot \ldots \cdot \sigma_{\vert N \vert,j},
\end{equation}
where $P_i(k,\dots,j)$ is the payoff of player $i$ in the pure strategy profile $(k, \dots, j)$. Further,
$h_{i}(\sigma_i,\sigma_{-i})$ represents the expected payoff of player $i$ as a function of $\sigma$. The social welfare function is derived from $h$ as $SW(\sigma) := \sum\nolimits_i h_i(\sigma)$.

In the following, we consider the concept of Nash equilibrium, that is a strategy profile in which no player
possesses an improving deviation. Formally, a strategy profile $\sigma^{ *} = (\sigma^*_{ 1}, \ldots, \sigma^{ *}_{ \vert N \vert } )$ is a Nash equilibrium, iff, for any $i$ and for any $\hat{\sigma}_i \in \Sigma_i$,
$h_{ i}(\sigma_{ i}^{ *}, \sigma_{ -i}^{ *}) \geq h_{ i}(\hat{\sigma}_{ i}, \sigma_{ -i}^{ *})$. Further, we denote by $NE(G)$ the set of Nash equilibria of $G$ (or simply $NE$, when $G$ is clear from the context). 

A useful notion in equilibrium computation is the support of a strategy $\sigma_i$, denoted by $\supp(\sigma_i)$, which is the set of pure strategies that are played with positive probability under $\sigma_i$. Formally, $\supp(\sigma_i) := \{j \in [\vert S_i \vert] : \sigma_{i,j} > 0\}$. 

Having at hand the definition of Nash equilibrium, authors in \cite{koutsoupias-papadimitriou:kts-pap} quantify the efficiency of Nash equilibrium by providing the metric of \emph{Price of Anarchy}. 
For that, they compare the behavior of the players in the worst case against the social optimum, based on a social welfare function. We denote by $opt$ the socially optimal strategy profile, i.e., $opt = \argmax_{\sigma} SW(\sigma)$. In terms of normal-form games, this becomes: \\

\noindent \textbf{Price of Anarchy}. Given a normal-form game $G$, the \index{metric!Price of Anarchy}\emph{Price of Anarchy (PoA)} is defined as
\begin{equation}\label{eq:PoA}
	PoA = \frac{SW(opt)}{min_{\sigma \in NE} SW(\sigma)}
\end{equation}

\subsection{Inflation process for normal form games}
\label{subsec:inflation process-nfg}

We consider the process of \emph{inflation}, which can be used to ``enlarge'' a normal form game, adding new players and strategies, in ways that do not alter the game's important properties. In particular, the \emph{inflated version} of a game $G$, say $G'$, is a ``bigger'' game (in terms of number of players and available strategies), such that the common players will have the same strategic behaviour and will receive the same payoffs in both games. 
The process of inflation will become relevant later, when we define ``canonical misinformation games'' (see Definition \ref{def:mis_games_special}). 

Algorithmically, to construct an inflated version of a game, we add strategies and players as appropriate, while respecting the above properties. To achieve this, when we add a strategy, we make it dominated (and thus irrelevant from a strategic perspective), and when we add a player, this player is dummy, in the sense that her actions do not affect the strategic choices of other players. 

Before defining inflation, we start with an auxiliary concept, namely \emph{compatible strategy profiles}. Given two games $G$ and $G'$, and two strategy profiles $\sigma$, $\sigma'$ of $G$, $G'$ respectively, we say that $\sigma$, $\sigma'$ are compatible if they agree (i.e., contain the same strategic choice) on all players and strategies that are common in the two games. Formally:

\begin{definition}
\label{def:compatible-str-prof}
Consider two normal-form games $G=\langle N, S, P \rangle$, $G' = \langle N', S', P' \rangle$, and let $S_i$ (respectively $S_i'$) be the set of strategies available for player $i\in N$ in $G$ (respectively $i \in N'$ in $G'$).
Consider also a strategy profile $\sigma = (\sigma_{ij})$ of $G$ and a strategy profile $\sigma' =(\sigma_{ij}')$ of $G'$.
Then, we say that $\sigma$, $\sigma'$ are \emph{compatible}, denoted by $\sigma \comp \sigma'$, if and only if
$\sigma_{ij} = \sigma_{ij}'$ for all $i \in N \cap N'$, $j\in S_i\cap S_i'$.
\end{definition}

It is easy to see that $\comp$ is an equivalence relation.
Now we can define the notion of inflation:

\begin{definition}
\label{def:inflated-nfgame}
Consider two normal-form games $G=\langle N, S, P \rangle$, $G' = \langle N', S', P' \rangle$, and let $S_i$ (respectively $S_i'$) be the set of strategies available for player $i\in N$ in $G$ (respectively $i \in N'$ in $G'$).
Then, we say that $G'$ is an \emph{inflated version} of $G$, denoted by $G \infl G'$, if and only if:
\begin{itemize}
    \item $N \subseteq N'$
    \item For all $i \in N$, $S_i \subseteq S_i'$
    \item $P_i(s) = P_i'(s')$ whenever $i \in N$ and $s=(s_j)_{j\in N} \in S$, $s'=(s_j')_{j\in N'} \in S'$ such that $s_j = s_j'$ for all $j\in N$.
    \item If $\sigma \in NE(G)$, then there exists some $\sigma' \in NE(G')$ such that $\sigma \comp \sigma'$.
    \item If $\sigma' \in NE(G')$, then there exists some $\sigma \in NE(G)$ such that $\sigma \comp \sigma'$.
\end{itemize}
\end{definition}

Some analysis on Definition \ref{def:inflated-nfgame} and its consequences is in order.
The first two bullets of the definition are straightforward: in order for $G'$ to be an inflated version of $G$, it must have at least the same players, and at least the same strategies for the (common) players.

The third bullet guarantees that the two games provide the same payoffs ``whenever possible''. In other words, $P$ is, in a sense, a submatrix of $P'$. The effect of this is that players that are common in $G$ and $G'$, while playing common (pure) strategies, will receive the same payoffs in the two games. An immediate consequence of this requirement is that the payoffs will also be the same when common players play any common strategy (pure or mixed):

\begin{proposition}
\label{prop:inflate-cons-hi}
Consider two normal-form games $G=\langle N, S, P \rangle$, $G' = \langle N', S', P' \rangle$, such that $G \infl G'$, and let $S_i$ (respectively $S_i'$) be the set of strategies available for player $i\in N$ in $G$ (respectively $i \in N'$ in $G'$). Take also some $i \in N$, and $\sigma$, $\sigma'$ strategy profiles in $G$, $G'$ respectively, such that $\sigma \comp \sigma'$. Then, $h_i(\sigma) = h_i'(\sigma')$, where $h_i$, $h_i'$ are the payoff functions of player $i$ in $G$, $G'$ respectively.
\end{proposition}

A consequence of Proposition \ref{prop:inflate-cons-hi} is that, in $G'$, the actions of players that exist only in $G'$ do not affect the payoffs of players that are common in $G$ and $G'$.
Indeed, for any two strategy profiles $\hat \sigma$, $\tilde \sigma$ in $G'$, if these profiles agree on the actions of the common players, then there exists some strategy profile $\sigma$ in $G$ such that $\sigma \comp \hat \sigma$ and $\sigma \comp \tilde \sigma$, and thus, by Proposition \ref{prop:inflate-cons-hi}, $h_i'(\hat \sigma) = h_i(\sigma) = h_i'(\tilde \sigma)$.

As a result, $P'$ can be seen as consisting of multiple ``copies'' of $P$ along the extra dimensions imposed by the extra players.
Note that this bullet does not impose any immediate requirements for the payoffs of players that exist only in $G'$, or for the payoffs of common players when some (common) player employs a strategy that does not exist in $G$. 

A further immediate consequence of the third bullet is that, if two games have the same number of players, and the same number of strategies per player, then each one is an inflated version of the other if and only if they are the same game.

\begin{proposition}
\label{prop:infl-same-size-same-game}
Consider two normal-form games $G=\langle N, S, P\rangle$, $G'=\langle N', S', P' \rangle$, such that $N = N'$ and $S = S'$. Then, $G \infl G'$ if and only if $G=G'$.
\end{proposition}

The last two bullets on Definition \ref{def:inflated-nfgame} ensure that the strategic behaviour of players in both games is the ``same''. In particular, the fourth bullet guarantees that any Nash equilibrium $\sigma$ of $G$ will have a ``counterpart'' in $G'$, i.e., there will be a Nash equilibrium in $G'$, say $\sigma'$, in which all common players employ the same strategies as in $\sigma$ (as usual, we don't care about the actions of the non-common players).
Analogously, the fifth bullet guarantees that any Nash equilibrium $\sigma'$ in $G'$ will have a ``counterpart'' in $G$, i.e., a Nash equilibrium, say $\sigma$ in which the common players employ the same strategies as in $\sigma'$. Note that this requirement disallows the existence of Nash equilibria in $G'$ for which any common player employs a non-common strategy with a probability higher than $0$, because, in such a strategy profile, the probabilities assigned to the common strategies do not sum to $1$, and thus there cannot exist a compatible strategy profile in $G$.
Note also the use of the compatibility relation between strategy profiles to formalise the notion of ``counterpart''.

An interesting consequence of Definition \ref{def:inflated-nfgame}, which combines the third and fourth bullet, is the following:

\begin{proposition}
\label{prop:infl-all-NE}
Consider two normal-form games $G=\langle N, S, P \rangle$, $G' = \langle N', S', P' \rangle$, such that $G \infl G'$.
Then, $\sigma \in NE(G)$ implies that $\sigma' \in NE(G')$ for all $\sigma \comp \sigma'$.
\end{proposition}

\noindent Clearly, the relation $\infl$ is a partial order, as shown in the proposition below:

\begin{proposition}
\label{prop:infl-partial order}
The relation $\infl$ is a partial order, i.e., for any normal-form games $G_1, G_2, G_3$ the following hold:
\begin{enumerate}
    \item $G_1 \infl G_1$ (reflexivity)
    \item If $G_1 \infl G_2$ and $G_2 \infl G_1$ then $G_1 = G_2$ (antisymmetry)
    \item If $G_1 \infl G_2$ and $G_2 \infl G_3$ then $G_1 \infl G_3$ (transitivity)
\end{enumerate}
\end{proposition}

The next question we will consider is how we create an inflated version of a game. For this, we developed the $\texttt{InflateGame}$ algorithm (Algorithm \ref{algo:inflate_game}), along with two supplementary routines, $\texttt{AddPlayer}$ (Algorithm~\ref{algo:addplayer_inflate_game}) and $\texttt{AddStrategy}$ (Algorithm~\ref{algo:addstrategy_inflate_game}). In Algorithm \ref{algo:inflate_game}, the input is a finite normal-form game $G$, a set of players $N'$ and a strategy space $S'$, and the goal is to enlarge (i.e., inflate) $G$ so as to have $N'$ players and strategy space $S'$. In lines \ref{algo:addplayer_inflate_game N'}-\ref{algo:addplayer_inflate_game P'} we add players to the game using Algorithm~\ref{algo:addplayer_inflate_game}, and in lines \ref{algo: addstrategy_inflate_game start for}-\ref{algo: addstrategy_inflate_game end for} we embed the strategy space using Algorithm~\ref{algo:addstrategy_inflate_game}.

In supplementary Algorithm~\ref{algo:addplayer_inflate_game}, we add a new player $i$ in $G$. Besides the enlargement of the set of players (line \ref{algo:addplayer_inflate_game N'}), we make sure that the strategy space is properly modified in order to support a (single) choice for the new player (line \ref{algo:addplayer_inflate_game S'}). We modify the payoff matrices so as to provide a payoff value for the new player (line \ref{algo:addplayer_inflate_game P'}). More specifically, as each entry of $P$ is an array, we enlarge each such array by one dimension. This is accomplished in line \ref{algo:addplayer_inflate_game P'} of the algorithm as follows: given an array $a \in \mathbb{R}^n$ we multiply it from the right with the matrix $[I_{n \times n} \mathbb{0}_{n \times 1}] \in \mathbb{R^{n \times n + 1}}$. This provides a new array with $n+1$ dimensions. We perform this operation for every entry of $P$. Observe that the new player will receive a zero payoff for any strategy profile.

Finally, in supplementary Algorithm~\ref{algo:addstrategy_inflate_game}, we add a new strategy $j$ for player $i$ in $G$. Again, we enhance the set of strategies $S_i$ of player $i$ in order to support the new strategic choice $j$, resulting to a new set of strategies $S'_i$ (line \ref{algo: addstrategy_inflate_game S'}). Afterwards, we modify the payoff matrices so as to provide a payoff value for the new strategy. For that, first we initialize the matrix $P'$ with dimensions $\vert S_{-i} \times S'_i \vert$ (line \ref{algo: addstrategy_inflate_game P'}). Then, in lines \ref{algo: addstrategy_inflate_game start for}-\ref{algo: addstrategy_inflate_game end for} we exhaustively run through all the pure strategy profiles $s$ of the strategic space $S_{-i} \times S'_i$. If $s \in S$ then $P'$ has the same values as in $P$, otherwise $P'$ has values less than the minimum value, $m$ (line \ref{algo: addstrategy_inflate_game m}), of $P$ for all players (line \ref{algo: addstrategy_inflate_game P' inside else}), in order to make this strategy dominated and thus irrelevant with respect to the Nash equilibrium, as required by the definition; this is ensured by the operation $m \mathbb{1}_{\vert N \vert}$.
Observe that this addition affects all axes in the payoff matrices. Also, the payoff values for the pure strategy profiles that do not include $j$ (aka indexes in $P$) are unchanged, but the remaining entries of the payoff matrices have to be such that they do not affect the strategic behavior of the players.

\begin{algorithm}[H]
\caption{$\texttt{InflateGame}(G,N',S')$}\label{algo:inflate_game}
{
\begin{algorithmic}[1]
\Require 
 \Statex $G = \langle N, S, P\rangle$
 \Statex $N'$ such that $N \subseteq N'$, and $\vert N' \vert < +\infty$
 \Statex $S'$ such that, $\forall i \in N'$, $S_i \subseteq S'_i$ and $\vert S'_i \vert < +\infty$ 

\Statex
 \For{$i \in N'\backslash N$}
 \State Set $G \leftarrow \texttt{AddPlayer}(G, i)$
 \EndFor

 \For{$i \in N'$}
 \For{$j \in S_i'\backslash S_i$}
 \State Set $G \leftarrow \texttt{AddStrategy}(G,j,i)$
 \EndFor
 \EndFor

\State \textbf{return} $\langle N', S', P'\rangle$
\end{algorithmic}
}
\end{algorithm}
\vspace{-2.1em}
\noindent \begin{minipage}{0.465\textwidth}
\vspace{-6.4em}
\begin{algorithm}[H]
\caption{$\texttt{AddPlayer}(G,i)$}\label{algo:addplayer_inflate_game}
{
\begin{algorithmic}[1]
\Require 
\Statex A game $G = \langle N, S, P\rangle$ 
\Statex A player $i \not\in N$
\Statex
\State $N' \leftarrow N \cup \{i\}$ \label{algo:addplayer_inflate_game N'}
\State $S' \leftarrow S \times \{1\}$ \label{algo:addplayer_inflate_game S'}
\State $P'_j \leftarrow P_j [ I_{\vert P_j \vert \times \vert P_j \vert} \enspace \mathbb{0}_{\vert P_j \vert \times 1} ]$, \label{algo:addplayer_inflate_game P'}
\Statex \hspace{9em} $\forall j $ entry of $P$
\State \textbf{return} $\langle N', S', P'\rangle$ \label{algo:addplayer_inflate_game return}
\end{algorithmic}
}
\end{algorithm}
\end{minipage}
\hfill
\begin{minipage}{0.465\textwidth}
\begin{algorithm}[H]
\caption{$\texttt{AddStrategy}(G,j,i)$}\label{algo:addstrategy_inflate_game}
{
\begin{algorithmic}[1]
\Require 
\Statex A game $G = \langle N, S, P\rangle$
\Statex A player $i\in N$, and a strategy $j \notin S_i$ 
\Statex
\State $m \leftarrow \min_{s \in S, k \in N} \{P_k(s)\} - 1$ \label{algo: addstrategy_inflate_game m}
\State $S'_i \leftarrow S_i \cup \{j\}$ \label{algo: addstrategy_inflate_game S'}
\State $P' \leftarrow \mathbb{0}_{S_{-i} \times S'_i}$ \label{algo: addstrategy_inflate_game P'}
\For{$s \in S_{-i} \times S'_i $} \label{algo: addstrategy_inflate_game start for}
\If{$s \in S$} \label{algo: addstrategy_inflate_game start if}
\State $P'(s) \leftarrow P(s)$ \label{algo: addstrategy_inflate_game P' inside if}
\Else
\State $P'(s) \leftarrow m \mathbb{1}_{\vert N \vert}$ \label{algo: addstrategy_inflate_game P' inside else}
\EndIf \label{algo: addstrategy_inflate_game end if}
\EndFor \label{algo: addstrategy_inflate_game end for}
\State \textbf{return} $\langle N, S', P'\rangle$ \label{algo: addstrategy_inflate_game return}
\end{algorithmic}
}
\end{algorithm}
\end{minipage}

We can show that these algorithms correctly produce an inflated version of a game. Given the transitivity of $\infl$, the results below hold also for repetitive applications of the above algorithms.

\begin{proposition}
\label{prop:infl-algo-add-player}
Let $G$ be a normal form game, and let $G'$ be the normal form game that is the output of $\texttt{AddPlayer}(G,i)$ (Algorithm~\ref{algo:addplayer_inflate_game}). Then $G \infl G'$.
\end{proposition}

\begin{proposition}
\label{prop:infl-algo-add-strategy}
Let $G$ be a normal form game, and let $G'$ be the normal form game that is the output of $\texttt{AddStrategy}(G,j,i)$ (Algorithm~\ref{algo:addstrategy_inflate_game}). Then $G \infl G'$.
\end{proposition}

\begin{proposition}
\label{prop:infl-algo-full}
Let $G$ be a normal form game, and let $G'$ be the normal form game that is the output of $\texttt{InflateGame}(G,N',S')$ (Algorithm~\ref{algo:inflate_game}). Then $G \infl G'$.
\end{proposition}

\subsection{Misinformation games}\label{sec:misinformation games}

Misinformation captures the concept that different players may have a specific, subjective, and thus different view of the game that they play. 

\begin{definition}
\label{def:mis_normal_basic}
A \emph{misinformation normal-form game} (or simply \emph{misinformation game}) is a tuple $mG = \tuple{G^{0}$ $, G^{1},$ $\ldots, G^{\vert N \vert}}$, where all $G^i$ are normal-form games and $G^0$ contains $\vert N \vert$ players.
\end{definition}

$G^0$ is called the \emph{actual game} and represents the game that is actually being played, whereas $G^i$ (for $i \in \{ 1, \dots, \vert N \vert$\}) represents the 
subjective view of player $i$, called the \emph{game of player $i$}.
We make no assumptions as to the relation among $G^{0}$ and $G^{i}$, and allow all types of misinformation to occur. An interesting special class of misinformation games is the following:

\begin{definition} \label{def:mis_games_special}
A misinformation game $mG= \mGdefo$ is called \emph{canonical} iff:
\begin{itemize}
    \item For any $i$, $G^0, G^i$ differ only in their payoffs.
    \item In any $G^i$, all players have an equal number of pure strategies.
\end{itemize}
\end{definition}

The first bullet of Definition \ref{def:mis_games_special} is usually true in practice. In order for a player to miss a strategy or player, or to have additional, imaginary, strategies or players, the communication failures should concern entire sections of a payoff matrix, and thus be quite substantial. The second bullet is more commonly violated, because the actual game, $G^0$, might violate it.

In any case, canonical misinformation games are simpler to grasp, and also simplify the formalisation. We can restrict our focus to canonical misinformation games without loss of generality, because we can transform any non-canonical misinformation game into a canonical one by repetitive applications of the inflation process for normal form games that was described in Subsection \ref{subsec:inflation process-nfg}. As explained there, the process of inflation does not alter the strategic properties of the normal form games (which in turn determine the strategic properties of the misinformation game).

To be more precise, we use the $\texttt{Inflation\text{ }process}$ algorithm (Algorithm \ref{algo:inflate_mG}) to transform a non-canonical $mG$ into a canonical $mG$ without altering its strategic properties. In a nutshell, Algorithm \ref{algo:inflate_mG} first computes the full list of players and strategies that each game should have in order to ensure the properties of Definition \ref{def:mis_games_special} (lines \ref{line:algo5 add game players}-\ref{line:algo5 add game strategies}).
Then, it utilizes the $\texttt{AddGame}$ algorithm (Algorithm \ref{algo:addgame}) to plug in a new game in the misinformation game, in order to ensure that any player, either real or fictional, has her own subjective view of the interaction (lines \ref{line:algo5 add game start for}-\ref{line:algo5 add game end for}). 
In $\texttt{AddGame}$, we create such a game for any player that is not in $N^0$ (and thus does not already exist in $mG$), ensuring that it has the correct players and strategies ($N_{\cup}, S_{\cup}$).
Afterwards, Algorithm \ref{algo:inflate_mG} inflates each of the normal-form games in the non-canonical misinformation game, to ensure that they, too, have the correct players and strategies (lines \ref{line:algo5 inflate game start for}-\ref{line:algo5 inflate game end for}). 
The final resulting $mG$ is returned in line \ref{line:algo5 return}.

\begin{algorithm}[H]
{
\caption{$\texttt{AddGame}(mG, N', S')$}\label{algo:addgame}
\begin{algorithmic}[1]
\Require 
\Statex A misinformation game $mG$
\Statex A set of players $N'$ such that $\vert N' \vert < +\infty$
\Statex A set of strategies $S'$
\Statex
\State Set $G' \leftarrow \langle \emptyset, \emptyset, 0 \rangle$
\State $\texttt{InflateGame}(G', N', S')$
\State \textbf{return} $\langle G^0, G^1, \ldots, G^{N}, G'\rangle$
\end{algorithmic}
}
\end{algorithm}

\begin{algorithm}[H]
\caption{$\texttt{Inflation\text{ }process}(mG)$}\label{algo:inflate_mG}
{
\begin{algorithmic}[1]
\Require 
 \Statex $mG = \langle G^0, G^1, \ldots, G^{\vert N \vert}\rangle$, with $\vert N \vert < +\infty$
\Statex

\State Set $N_{\cup} \leftarrow \bigcup_{j=0,\dots, \vert N \vert} N^j$ \label{line:algo5 add game players}
\State Set $S_{\cup} \leftarrow \times_{j \in N_{\cup}} \left( \bigcup_{j\in N_{\cup}}\bigcup_{k\in N_{\cup} \cup \{0\}} S_j^k \right)$ \label{line:algo5 add game strategies}

\For{$i \in N_{\cup} \setminus N^0$} \label{line:algo5 add game start for}
   \State $mG \leftarrow \texttt{AddGame}(mG, N_{\cup}, S_{\cup})$ 
\EndFor \label{line:algo5 add game end for}

\For{$i \in N_{\cup} \cup \{0\}$} \label{line:algo5 inflate game start for}
    \State Set $G^i \leftarrow \texttt{InflateGame}(G^i, N_{\cup}, S_{\cup})$ 
\EndFor \label{line:algo5 inflate game end for}

\State \textbf{return} $mG$ \label{line:algo5 return}
\end{algorithmic}
}
\end{algorithm}

Since the inflation process does not alter the strategic properties of a game, and any non-canonical misinformation game can be turned into a canonical one using Algorithm \ref{algo:inflate_mG}, we will establish strategy profiles and equilibrium concepts in Subsection~\ref{subsec:Strategy profiles in misinformation games} below, focusing on canonical misinformation games only.
As we will see, introducing these concepts for the non-canonical case, brings about all sorts of unrelated technical difficulties that are distracting and complicate the formalism, without contributing to the main intuition.

\subsection{Strategy profiles in misinformation games}
\label{subsec:Strategy profiles in misinformation games}

The definition of misinformed strategies and strategy profiles is straightforward, once noticing that they refer to each player's own game:

\begin{definition}\label{def:misinformed_strategy}
A {\em misinformed strategy}, $m\sigma_i$ is a strategy of player $i$ in game $G^i$. We denote the set of all possible misinformed strategies of player $i$ as $\Sigma_i^i$. A {\em misinformed strategy profile} of $mG$ is an $\vert N \vert$-tuple of misinformed strategies $m\sigma = (m\sigma_{1}, \ldots, m\sigma_{\vert N \vert})$, where $m\sigma_i \in \Sigma_i^i$.
\end{definition}

As usual, we denote by $m\sigma_{-i}$ the $(\vert N\vert - 1)$-tuple strategy profile of all other players except for player $i$ in a misinformed strategy $m\sigma$. The payoff function $h_i$ of player $i$ under a given profile $m\sigma$ is determined by the payoff matrix of $G^0$, and is defined as $h_{ i}: \Sigma_1^1 \times \dots \times \Sigma_{\vert N \vert}^{\vert N \vert} \to \mathbb{R}$, such that:
\begin{equation*}
    h_{i}(m\sigma_i, m\sigma_{-i}) = \sum\nolimits_{k\in S^1_1} \dots \sum\nolimits_{j\in S^{\vert N \vert}_{\vert N \vert}} P^0_i(k,\dots,j) \cdot m\sigma_{1,k} \cdot \ldots \cdot m\sigma_{ \vert N \vert ,j},
\end{equation*}
where $P^0_i(k,\dots,j)$ is the payoff of player $i$ in the pure strategy profile $(k,\dots,j)$ under the actual game $G^0$. Also, $S^j_i$ denotes the set of pure strategies of player $i$ in game $G^j$. 

Observe that, although each player's strategic decisions are driven by the information in her own game ($G^i$), the received payoffs are that of the actual game $G^0$, that may differ than $G^i$. Further, the payoff function would be ill-defined in case of non-canonical misinformation games.

Next, we define the solution concept of a misinformation game, where each player chooses a Nash strategy, neglecting what other players know or play:

\begin{definition}\label{def:misinformed_strategy_equilibrium_Gi}
A misinformed strategy, $m\sigma_{i}$, of player $i$, is a \emph{misinformed equilibrium strategy}, iff, it is a Nash equilibrium strategy for game $G^{i}$. A misinformed strategy profile $m\sigma$ is called a \emph{natural misinformed equilibrium} iff it consists of misinformed equilibrium strategies.
\end{definition}

In the following, we denote by $NME(mG)$ (or simply NME, when $mG$ is obvious from the context) the set of natural misinformed equilibria of $mG$, and by $nme$ the elements of $NME(mG)$.

\begin{example}[Running Example]\label{example:running example-I}
Consider the canonical misinformation game $mG = \langle G^0, G^{1}, G^{2} \rangle$, where $G^{i} = \langle \{1, 2\}, S = \{s_1, s_2\} \times \{s_1, s_2\}, P^{i} \rangle$, with $i \in \{0, 1, 2\}$ and payoff matrices as depicted in Table~\ref{tbl:misninformation-game}.
\begin{table}[h]
    \centering
    \caption[Payoff matrices of the misinformation game of Example~\ref{example:running example-I}.]{\small Payoff matrices in the misinformation game of Example~\ref{example:running example-I}.}\label{tbl:misninformation-game}
    \begin{subtable}{.32\linewidth}
       \centering
       \begin{tabular}{ |c|c|c| }
              \hline
              & $s_1$ & $s_2$   \\
              \hline
              $s_1$ &  (6,6) & (2,7)   \\
             \hline
              $s_2$ &  (7,2) & (1,1) \\
             \hline
            \end{tabular}
            \caption{$P^0$}\label{tbl:example-actual-game}
    \end{subtable}
    \enspace
    \begin{subtable}{.32\linewidth}
       \centering
        \begin{tabular}{ |c|c|c| }
         \hline
              & $s_1$ & $s_2$   \\
              \hline
              $s_1$ &  (2,2) & (0,3)   \\
             \hline
              $s_2$ &  (3,0) & (1,1) \\
             \hline
        \end{tabular}
        \caption{$P^1$}\label{tbl:example-1-view}
    \end{subtable}
    \enspace
    \begin{subtable}{.32\linewidth}
       \centering
        \begin{tabular}{ |c|c|c| }
         \hline
              & $s_1$ & $s_2$   \\
              \hline
              $s_1$ &  (-1,1) & (2,-2)   \\
             \hline
              $s_2$ &  (1,-1) & (0,0) \\
             \hline
        \end{tabular}
        \caption{$P^2$}\label{tbl:example-2-view}
    \end{subtable}
\end{table}

Player's $1$ equilibrium strategy is $s_2$ in $G^{1}$ (i.e., player $1$ will play $(0,1)$), while player's $2$ mixed equilibrium strategy in $G^{2}$ is $(1/2,1/2)$. Thus, we have $NME = \{ (s_2, s_1), (s_2, s_1) \}$ with $nme$s $(0, 1), (1, 0))$ and $(0, 1), (0, 1))$.
\end{example}

To understand the computational properties of calculating an $nme$, we observe that any natural misinformed equilibrium consists of the agglomeration of the different Nash equilibrium strategy profiles, one for each $G^i$. The computation of a Nash equilibrium for each $G^i$ is \textbf{PPAD}-complete \cite{DBLP:journals/cacm/DaskalakisGP09}, and computing the $nme$ amounts to repeating this computation for each $G^i$. Thus, the computation of the natural misinformed equilibrium is also \textbf{PPAD}-complete.

Inspired by the {\em Price of Anarchy} metric we define a metric, called the \emph{Price of Misinformation (PoM)} to measure the effect of misinformation compared to the social optimum. 
$PoM$ is defined as follows: 
\begin{definition}\label{def:PoM}
Given a misinformation game $mG$, the \emph{Price of Misinformation} is defined as:
\begin{equation}\label{eq:PoM_max}
\displaystyle PoM = \frac{SW(opt)}{\min_{\scriptscriptstyle \sigma \in NME} SW(\sigma)}
\end{equation}
\end{definition}
\noindent Using the definition of {\em PoA}~\cite{koutsoupias-papadimitriou:kts-pap} and $\eqref{eq:PoM_max}$ we derive the following formula: 
\begin{equation}\label{eq:PoMandPoA}
 \frac{PoM}{PoA} = \frac{\min_{\scriptscriptstyle \sigma \in NE} SW(\sigma)}{\min_{\scriptscriptstyle \sigma \in NME} SW(\sigma)}
\end{equation}

Observe that, if $PoM < PoA$, then misinformation has a beneficial effect on social welfare, as the players are inclined (due to their misinformation) to choose socially better strategies. On the other hand, if $PoM > PoA$, then misinformation leads to a worse outcome from the perspective of social welfare.

\addtocounter{example}{-1}
\begin{example}[continued]\label{example:running example-II}
For the $mG$ in Example~\ref{example:running example-I} game $G^0$ has three Nash equilibria, with strategy profiles $((1,0), (0, 1))$, $((0,1), (1,0))$, and $((2/3, 1/3), (2/3, 1/3)$. The worst Nash equilibrium provides a value of $26/3$, whereas the optimal is in strategy profile $((1, 0), (1, 0))$ with a value of $12$, both in terms of social welfare. Thus, $PoA = 18/13$.

On the other hand, the $nme$ strategy profile $((0,1),(1/2, 1/2))$ has a social welfare value of $9/2$, resulting to $PoM = 8/3$. Hence, $PoA < PoM$.
\end{example}

Another interesting result shown in \cite{VFFB} is that misinformation imposes any desirable behaviour on players. Therefore, the use of misinformation games provides a powerful technique that trivializes the problem of mechanism design, when misinformation can be used. A side-effect of this, related to the Price of Misinformation introduced above, is that anything is possible in terms of improving (or deteriorating) social welfare while using misinformation as a means for mechanism design. We will not further consider this property in this paper; for more details, see \cite{VFFB}.

As a side note, let us consider how the definitions of this subsection would apply (or not) for the case of non-canonical misinformation games. In a non-canonical misinformation game, a player $i$ may have a subjective view $G^i$ that diverts from $G^0$ in every aspect, i.e., number of players, number of pure strategies, payoffs. We observe that her strategic behavior will still be guided by the contents of her subjective game, $G^i$, so Definition \ref{def:misinformed_strategy} applies normally. Consequently, Definition \ref{def:misinformed_strategy_equilibrium_Gi} for the $nme$ is also well-defined. 
Contrarily, we observe that the formula for computing the payoff values $h_i$ for some misinformed strategy profile $m\sigma$ is no longer well-defined, because there may exist misinformed strategy profiles (either pure or mixed) that are unknown to $G^0$, and thus correspond to no entry in the payoff matrix $P^0$.
Moreover, from the perspective of individual players, it is unclear what a player should do when being informed that another player used a strategy that is unknown to her (or that an unknown player participated in the game), and how the received payoff should be interpreted (since it corresponds to no entry in her own payoff matrix).
Addressing this interesting case will be considered in future work.

\section{Adaptation Procedure}\label{sec:adaptation procedure}

In previous work (and in our analysis on misinformation games above), only one-shot interactions were considered; this is a major limitation, as we expect players to realise (partly) the shortcomings of their knowledge following the reception of the (unexpected) payoffs and to revise accordingly. To address this shortcoming, we introduce the Adaptation Procedure that enables us to model iterative interactions. Abusing notation, we write $mG^{(t)}$ instead of $mG$ in order to plug in the step of interaction (this change in the notation is transfused in every concept).

We start with an informal description of the Adaptation Procedure, and continue by presenting the respective formalisation and showing some properties.

\subsection{Informal description of the Adaptation Procedure}

The Adaptation Procedure is initialized with a misinformation game, say $mG^{(0)}$. As explained in Section \ref{subsec:Strategy profiles in misinformation games}, this will cause each player to employ one of the equilibrium strategies in her own game, resulting to a natural misinformed equilibrium. The payoff received from the joint strategic choices will be provided by the actual game, $G^0$, and this may be different from what each player knows (and expects) from her subjective game. 
Notably, we assume that the payoffs received by each player for their strategic choices are publicly announced, thus are common knowledge.
As a result, players will update their payoff matrices by replacing the erroneous payoffs with the correct ones just received, leading to a new misinformation game.

Interestingly, the above process is not, in general, linear. When the misinformation game has more than one natural misinformed equilibria, and/or when there exist mixed strategic choices in them, each of these choices will be considered in a separate branch of the process, because the players may choose any of those strategies to play. 
Any such branch leads to a different new misinformation game, and, thus, our process should take into account all those potential scenarios. To do this, $mG^{(0)}$ will spawn several new misinformation games, one for each element of the support of the natural misinformed equilibria, eventually forming a tree (more precisely, a directed graph, as there may be self-loops or mergers, as we will show later). 
Note that this spawning should not be interpreted as leading to multiple parallel processes; instead, it represents different potential ways that the process can unfold.

The process continues recursively for each branch, creating new misinformation games. When no new misinformation games are spawned (i.e., when the players learn nothing new from the environment and no longer change their payoff matrices), the Adaptation Procedure \emph{terminates}.

Observe that the Adaptation Procedure produces new games (and thus new $nme$s) in each time step. In general an $nme$ is not a stable equilibrium concept when the Adaptation Procedure is considered, because, after the players learn something new about the reality as explained above, they might choose to change their choices. Nevertheless, the $nme$s of the games appearing in the leaves of the recursive tree at the time when the Adaptation Procedure terminates are stable (nothing new is learnt). Given that, we define a new equilibrium concept, that is, a strategy profile that is an $nme$ of one of the leaves, and we call it \emph{stable misinformed equilibrium} ($sme$); players do not have incentives to deviate from an $sme$, even in the presence of the payoff information received by the environment (game). 

\subsection{Formal definition}

Consider a multidimensional matrix $A$, and a vector $\vec{v}$. We denote by $A_{\vec{v}}$ the element of $A$ in position $\vec{v}$. For example, $A_{(1,2)}$ is the top right element of the $2 \times 2$ matrix $A$.
More general, for a $n_1 \times n_2 \times \dots \times n_m$ matrix, a \emph{position} is an $m$-dimensional vector $\vec v \in [n_1] \times [n_2] \times \dots \times [n_m]$, where $[n_i] = \{1,2,\dots, n_i\}$. 
Furthermore, the set of all position vectors in a multidimensional payoff matrix is the set of joint strategies $S$. In what follows, we will use the term position vectors when referring in a specific point in the payoff matrix. Thereupon, we define the operation of replacement of element $A_{\vec{v}}$ with $b$ as follows:

\begin{definition}
\label{def:update_operation}
Consider set $F$, matrix $A \in F^{n_1 \times n_2 \times \ldots \times n_m}$, vector $\vec{v}$ indicating a position in $A$ and some $b \in F$. We denote by $A \oplus_{\vec{v}} b$ the matrix $B \in F^{n_1 \times n_2 \times \ldots \times n_m}$, such that $B_{\vec{v}} = b$ and $B_{\vec{u}} = A_{\vec{u}}$ for all $\vec{u} \neq \vec{v}$.
\end{definition}

In practice, the operator $\oplus_{\vec v}$ will be used to replace a vector of payoffs (that appears in the position $\vec v$) with some other vector. So the set $F$ in our case will consist of all appropriately-sized vectors of real numbers, and $b$ will be one such vector. This replacement reflects the effect of the Adaptation Procedure in the subjective payoff matrices of players, as it leads to the updating of an erroneous payoff vector with the actual one that is communicated to them by the game (and appears in the actual payoff matrix). These ideas are formalised below:

\begin{definition}\label{def:vec_def}
Consider a canonical misinformation game $mG = \langle G^0, G^1, \dots, G^{\vert N \vert} \rangle$, where $G^i = \langle N, S, P^i \rangle$ (for $0 \leq i \leq \vert N \vert$), and some vector $\vec{v}$. We define the $\vec{v}$-update of $mG$, denoted by $mG_{\vec{v}}$, to be the misinformation game 
$\langle G^0, {G^1}', \ldots, {G^{\vert N \vert}}'\rangle$, where ${G^i}' = \langle N, S, P^i\oplus_{\vec{v}} P^{0}_{\vec{v}}\rangle$, for $1 \leq i \leq \vert N \vert$.
\end{definition}

Definition \ref{def:vec_def} tells us how to perform the update process that the Adaptation Procedure requires, in particular by replacing the payoffs of the newly learnt position in each of the subjective payoff matrices ($P^i$) by the actual payoff given by $P^0$ (i.e., by $P^0_{\vec v}$). Interestingly, the $\vec{v}$-update of $mG$ is idempotent, $(mG_{\vec{u}_1})_{\vec{u}_1} = mG_{\vec{u}_1}$, and commutative, $(mG_{\vec{u}_1})_{\vec{u}_2} = (mG_{\vec{u}_2})_{\vec{u}_1}$. 

Abusing notation, for a set of positions $X = \{\vec{u}_1,\dots,\vec{u}_k\}$, we denote by $mG_X$ the game $mG_X = (\dots(mG_{\vec{u}_1})_{\vec{u}_2}\dots)_{\vec{u}_k}$. Given the properties above, the notation $mG_X$ is well-defined.

The position where the update takes place (denoted by $\vec{v}$ in Definition \ref{def:vec_def}) is determined by the strategic choices of players, and can be ``extracted'' from a strategy profile that is an $nme$, using the following definition:

\begin{definition}\label{def:characteristic}
Consider a strategy profile $\sigma = (\sigma_1, \ldots, \sigma_N)$ with $\sigma_i \in [0, 1]^{\vert S_i \vert}$ and $S_i = \{s_{i1}, \ldots, s_{i\vert S_i \vert}\}$. The characteristic strategy set of vectors of $\sigma$ is $\chi(\sigma) = \chi(\supp(\sigma_1)) \times \dots \times \chi(\supp(\sigma_N))$, with $\chi(\supp(\sigma_j)) = \{i \vert s_{ji} \in \supp(\sigma_j)\}$.
\end{definition}

In practice, $\chi$ identifies the indices of the strategies that each player plays with non-zero probability, and then returns a set containing all possible vectors that can be formed with these indices (where each point in the vector corresponds to a player). This will be further clarified with the following example:

\begin{example}
Assume a $4 \times 3$ bimatrix game. Then, the characteristic strategy set of vectors of $\sigma = ((1/2,0,1/3,1/6),$ $(0,0,1))$ is $\chi(\sigma) = \{(1,3), (3,3), (4,3)\}$. \qed 
\end{example}

As explained above, the Adaptation Procedure occurs in discrete time steps $t \in \mathbb{N}_0 = \mathbb{N}\cup\{0\}$. It starts from $t = 0$ where player $i$ has the view $G^{i, (0)}$, $\forall i \in [\vert N \vert]$, and in each time step $t$ we implement the update operation described in Definition \ref{def:vec_def} for the payoff vector(s) that correspond to the strategic choices of the players. In other words, we update the strategy values of the strategy profiles of the players. The following example illustrates this procedure using the above notions, and is also visualised in Figure~\ref{fig:example adaptation}: 

\addtocounter{example}{-2}
\begin{example}[continued]\label{example:running example-III}
We write $mG$ as $\mGt{t} = \langle G^0, G^{1, (t)}, G^{2, (t)} \rangle$. Using the support of the $nme$ and the characteristic strategy vector we take $\chi\{((0, 1), (1/2, 1/2))\} = \{(2,1), (2,2)\}$. So, we have the strategy value vector $(2,1)$ for the strategy profile $((0, 1), (1, 0))$, and the strategy value vector $(2,2)$ for the strategy profile $((0, 1), (0, 1))$.

Notice that, as one player is randomized, $\chi(nme)$ has more than one elements, and the Adaptation Procedure branches result to two new misinformation games, say $\mGt{1a}$, $\mGt{1b}$. The first one becomes, $\mGt{1a} = \langle G^{0}, G^{1, (1a)}, G^{2, (1a)} \rangle$ with payoff matrices as shown in Table~\ref{tbl:misninformation-game-II} (note how the bottom-left payoff has been updated).

\begin{table}[h]
    \centering
    \caption[Payoff matrices in the misinformation game of the Example~\ref{example:running example-I} using the $nme$ $((0, 1), (1, 0))$ for the update.]{\small Payoff matrices in the misinformation game of the Example~\ref{example:running example-I} using the $nme$ $((0, 1), (1, 0))$ for the update.}\label{tbl:misninformation-game-II}
    \begin{subtable}{.32\linewidth}
       \centering
       \begin{tabular}{ |c|c|c| }
              \hline
              & $s_1$ & $s_2$   \\
              \hline
              $s_1$ &  (6,6) & (2,7)   \\
             \hline
              $s_2$ &  (7,2) & (1, 1) \\
             \hline
            \end{tabular}
            \caption{$P^0$}\label{tbl:example-actual-game-II}
    \end{subtable}
    \enspace
    \begin{subtable}{.32\linewidth}
       \centering
        \begin{tabular}{ |c|c|c| }
         \hline
              & $s_1$ & $s_2$   \\
              \hline
              $s_1$ &  (2,2) & (0,3)   \\
             \hline
              $s_2$ &  (\textcolor{upsdellred}{7}, \textcolor{upsdellred}{2}) & (1,1) \\
             \hline
        \end{tabular}
        \caption{$P^{1,(1a)}$}\label{tbl:example-1-view-II}
    \end{subtable}
    \enspace
    \begin{subtable}{.32\linewidth}
       \centering
        \begin{tabular}{ |c|c|c| }
         \hline
              & $s_1$ & $s_2$   \\
              \hline
              $s_1$ &  (-1,1) & (2,-2)   \\
             \hline
              $s_2$ &  (\textcolor{upsdellred}{7}, \textcolor{upsdellred}{2}) & (0,0) \\
             \hline
        \end{tabular}
        \caption{$P^{1,(1a)}$}\label{tbl:example-2-view-II}
    \end{subtable}
\end{table}

Similarly, the second one becomes, $\mGt{1b} = \langle G^{0}, G^{1, (1b)}, G^{2, (1b)} \rangle$ with the payoff matrices (note how the bottom-right payoff has been updated) as are shown in Table~\ref{tbl:misninformation-game-III}.

\begin{table}[h]
    \centering
    \caption[Payoff matrices in the misinformation game of the Example~\ref{example:running example-I} using the $nme$ $((0, 1), (0, 1))$ for the update.]{\small Payoff matrices in the misinformation game of the Example~\ref{example:running example-I} using the $nme$ $((0, 1), (0, 1))$ for the update.}\label{tbl:misninformation-game-III}
    \begin{subtable}{.32\linewidth}
       \centering
       \begin{tabular}{ |c|c|c| }
              \hline
              & $s_1$ & $s_2$   \\
              \hline
              $s_1$ &  (6,6) & (2,7)   \\
             \hline
              $s_2$ &  (7,2) & (1,1) \\
             \hline
            \end{tabular}
            \caption{$P^0$}\label{tbl:example-actual-game-III}
    \end{subtable}
    \enspace
    \begin{subtable}{.32\linewidth}
       \centering
        \begin{tabular}{ |c|c|c| }
         \hline
              & $s_1$ & $s_2$   \\
              \hline
              $s_1$ &  (2,2) & (0,3)   \\
             \hline
              $s_2$ &  (3,0) & (\textcolor{upsdellred}{1}, \textcolor{upsdellred}{1}) \\
             \hline
        \end{tabular}
        \caption{$P^{1,(1b)}$}\label{tbl:example-1-view-III}
    \end{subtable}
    \enspace
    \begin{subtable}{.32\linewidth}
       \centering
        \begin{tabular}{ |c|c|c| }
         \hline
              & $s_1$ & $s_2$   \\
              \hline
              $s_1$ &  (-1,1) & (2,-2)   \\
             \hline
              $s_2$ &  (1,-1) & (\textcolor{upsdellred}{1}, \textcolor{upsdellred}{1}) \\
             \hline
        \end{tabular}
        \caption{$P^{2,(1b)}$}\label{tbl:example-2-view-III}
    \end{subtable}
\end{table}
\end{example}

The procedure shown in Example \ref{example:running example-I} is formalised as follows, taking into account the fact that the process may branch when $\chi(\sigma)$ is not a singleton set:

\begin{definition}\label{def:AD(M)-adaptation procedure iterative process}
For a set $M$ of misinformation games, we set:
\begin{equation*}
\AD{}{M} = \{ mG_{\vec{u}} \mid mG \in M, \vec{u} \in \chi(\sigma), \sigma \in NME(mG) \}
\end{equation*}
We define the \emph{Adaptation Procedure} as the following iterative process:
\begin{equation}\label{eq:naive rec}
    \left\{\begin{array}{l}
     \AD{(0)}{M} = M \\
     \AD{(t+1)}{M} = \AD{(t)}{\AD{}{M}}
   \end{array}
\right.
\end{equation}
for $t \in \mathbb{N}_{0}$.
\end{definition}

The functionality of the Adaptation Procedure between two consecutive time steps $t$ and $t+1$, as provided by Definition~\ref{def:AD(M)-adaptation procedure iterative process}, is depicted in Figure~\ref{fig:schematic_representation_AD}.

\pgfdeclarelayer{background}
\pgfdeclarelayer{foreground}
\pgfsetlayers{background,main,foreground}

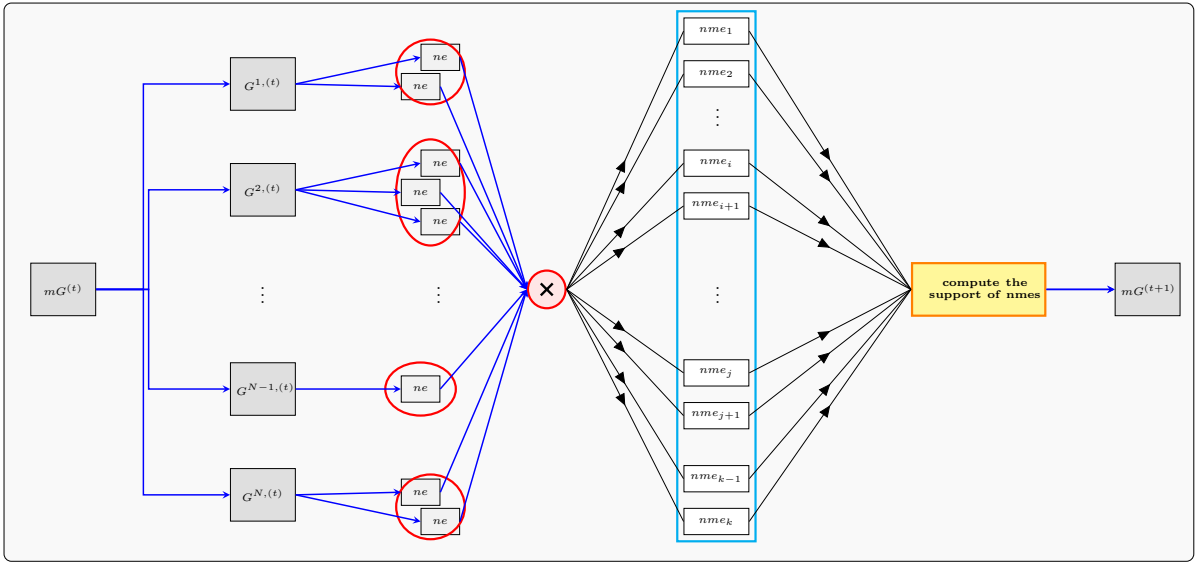
\begin{figure}[t]
\centering
\begin{adjustbox}{max width=\linewidth}
\begin{tikzpicture}[
    middlearrow/.style 2 args={
        decoration={             
            markings, 
            mark=at position 0.5 with {\arrow{triangle 45}, \node[#1] {#2};}
        },
        postaction={decorate}
    },
    my mark/.style={
        decoration={
            markings,
            mark=at position 0.5 with{\color{red}\pgfuseplotmark{x}},
        },
        postaction=decorate,
    }
]		
		\node (start) [rec] {\tiny $mG^{(t)}$};

		\node (gdts) [draw=none, right= 3cm of start] {$\mathbf{\vdots}$};
		\node (g2) [rec, above= of gdts] {\tiny $G^{2,(t)}$};
		\node (g1) [rec, above= of g2] {\tiny $G^{1,(t)}$};
		\node (gN_1) [rec, below= of gdts] {\tiny $G^{N-1,(t)}$};
		\node (gN) [rec, below= of gN_1] {\tiny $G^{N,(t)}$};

		\node (ngdts) [draw=none, right= 3cm of gdts] {$\mathbf{\vdots}$};
  
        \node (ne21) [smrec, right= 3cm of g2.north] {\tiny $ne$};
        \node (ne22) [smrec, below = .3cm of ne21.west] {\tiny $ne$};
        \node (ne23) [smrec, below = .3cm of ne22.east] {\tiny $ne$};

        \node[ellipse, very thick, draw=red, fit=(ne21) (ne22) (ne23), inner sep=-1mm] (all1) {};
        
        \node (ne11) [smrec, right= 3cm of g1.north] {\tiny $ne$};
        \node (ne12) [smrec, below = .3cm of ne11.west] {\tiny $ne$};

        \node[ellipse, very thick, draw=red, fit=(ne11) (ne12), inner sep=-1mm] (all2) {};

        \node (neN11) [smrec, right= 2cm of gN_1] {\tiny $ne$};

        \node[ellipse, very thick, draw=red, fit=(neN11), inner sep=1mm] (all3) {};

        \node (neN1) [smrec, right= 3cm of gN.south] {\tiny $ne$};
        \node (neN2) [smrec, above = .3cm of neN1.west] {\tiny $ne$};

        \node[ellipse, very thick, draw=red, fit=(neN1) (neN2), inner sep=-1mm] (all4) {};

       \node (cross) [circle, very thick, draw=red, fill=red!10, right= 1.5cm of ngdts] {$\Cross$};

  		\node (nmegdts) [draw=none, right= 2.7cm of cross] {$\vdots$};
  
        \node (nme21) [smrec, text width=1cm, right= 8cm of g2.north, fill=white] {\tiny $nme_{i}$};
        \node (nme22) [smrec, text width=1cm, below = .3cm of nme21, fill=white] {\tiny $nme_{i+1}$};
        
        \node (nmedgts1) [draw=none, above= 2.6cm of nmegdts] {$\vdots$};
        
        \node (nme11) [smrec, text width=1cm, above= 2cm of nme21.north, fill=white] {\tiny $nme_1$};
        \node (nme12) [smrec, text width=1cm, below = .3cm of nme11, fill=white] {\tiny $nme_2$};

        \node (nmeN11) [smrec, text width=1cm, right= 8cm of gN_1.south, fill=white] {\tiny $nme_{j+1}$};
        \node (nmeN12) [smrec, text width=1cm, above = .3cm of nmeN11, fill=white] {\tiny $nme_{j}$};
        
        \node (nmeN1) [smrec, text width=1cm, right= 8cm of gN.south, fill=white] {\tiny $nme_{k}$};
        \node (nmeN2) [smrec, text width=1cm, above = .3cm of nmeN1, fill=white] {\tiny $nme_{k-1}$};

        \node[rectangle, very thick, draw=cyan, fit=(nme11) (nme12) (nme21) (nme22) (nmegdts) (nmeN11) (nmeN12) (nmeN1) (nmeN2), inner sep=1.2mm] (all5) {};

  		\node (support) [rec, very thick, text width=2.3cm, right= 3.5cm of nmegdts, draw = orange, fill=yellow!50] {\tiny \begin{tabular}{c} \textbf{compute the} \\ \textbf{support of nmes}\end{tabular}};

		\node (end) [rec, right= 1.3cm of support]{\tiny $mG^{(t+1)}$};

  \draw [arrow,->] (start.east) -- ++(+.9,0) node[midway,above]{} |- (g1.west);
  \draw [arrow,->] (start.east) -- ++(+1,0) node[midway,above]{} |- (g2.west); 
  \draw [arrow,->] (start.east) -- ++(+1,0) node[midway,above]{} |- (gN_1.west);
  \draw [arrow,->] (start.east) -- ++(+.9,0) node[midway,above]{} |- (gN.west);

  \draw [arrow,->] (ne11.east) -- (cross.west);
  \draw [arrow,->] (ne12.east) -- (cross.west);

  \draw [arrow,->] (ne21.east) -- (cross.west);
  \draw [arrow,->] (ne22.east) -- (cross.west);
  \draw [arrow,->] (ne23.east) -- (cross.west);

  \draw [arrow,->] (neN11.east) -- (cross.west);

  \draw [arrow,->] (neN1.east) -- (cross.west);
  \draw [arrow,->] (neN2.east) -- (cross.west);

  \draw [arrow,->] (support.east) -- (end.west);

  \draw [arrow,->] (g1.east) -- (ne11.west);
  \draw [arrow,->] (g1.east) -- (ne12.west);

  \draw [arrow,->] (g2.east) -- (ne21.west);
  \draw [arrow,->] (g2.east) -- (ne22.west);
  \draw [arrow,->] (g2.east) -- (ne23.west);

  \draw [arrow,->] (gN_1.east) -- (neN11.west);

  \draw [arrow,->] (gN.east) -- (neN1.west);
  \draw [arrow,->] (gN.east) -- (neN2.west);

   \draw [middlearrow,-] (cross.east) -- (nme11.west);
  \draw [middlearrow,-] (cross.east) -- (nme12.west);
  \draw [middlearrow,-] (cross.east) -- (nme21.west);
  \draw [middlearrow,-] (cross.east) -- (nme22.west);
  \draw [middlearrow,-] (cross.east) -- (nmeN11.west);
  \draw [middlearrow,-] (cross.east) -- (nmeN12.west);
  \draw [middlearrow,-] (cross.east) -- (nmeN1.west);
  \draw [middlearrow,-] (cross.east) -- (nmeN2.west);
  
  \draw [middlearrow,-] (nme11.east) -- (support.west);
  \draw [middlearrow,-] (nme12.east) -- (support.west);
  \draw [middlearrow,-] (nme21.east) -- (support.west);
  \draw [middlearrow,-] (nme22.east) -- (support.west);
  \draw [middlearrow,-] (nmeN11.east) -- (support.west);
  \draw [middlearrow,-] (nmeN12.east) -- (support.west);
  \draw [middlearrow,-] (nmeN1.east) -- (support.west);
  \draw [middlearrow,-] (nmeN2.east) -- (support.west);
  \draw [arrow,->] (support.east) -- (end.west);  
    \coordinate (a) at (-2,.75);
    \coordinate (b) at (0,1.75);
    \coordinate (c) at (18,.75);
    \coordinate (d) at (18,1.75);
		
    \begin{pgfonlayer}{background}
          \path (a.east |- a.south)+(.25,.5) node (f) {};
          \path (c.south -| c.east)+(+2,-4) node (g) {};
           \path (d.west |- d.north)+(+3.5,3.7) node (h) {};
      
        \path[fill=gray!5,rounded corners, draw=black]
             (f.south east)+(0.5,-6.3) rectangle (h);           
        
    \end{pgfonlayer}
		
	\end{tikzpicture}
\end{adjustbox}
	\caption[Schematic representation of the functionality of Adaptation Procedure from time step $t$ to time step $t+1$, that is $\AD{t+1}{mG}$.]{\small Schematic representation of the functionality of Adaptation Procedure from time step $t$ to time step $t+1$, that is $\AD{t+1}{mG}$.}
    \label{fig:schematic_representation_AD}
\end{figure}

Note that the Adaptation Procedure is defined over a set of misinformation games. Although our intent is basically to apply it over a single misinformation game, the branching process, along with the recursive nature of the definition, forces us to consider the more general case right from the start. Note also that we will often abuse notation and write \AD{}{mG} (or \AD{(t)}{mG}) instead of \AD{}{ \{ mG \} } (or \AD{(t)}{\{mG\}}).

The following example shows how the Adaptation Procedure of Example \ref{example:running example-I} continues in its second step. Interestingly, $\mGt{1}$ includes a hybrid natural misinformed equilibrium (i.e., one in which one player plays a pure strategy and the other a mixed one), thus illustrating the branching process mentioned above.

\addtocounter{example}{-1}
\begin{example}[continued]
\label{example:running example-IV}
At $t = 1$, we note that the Adaptation Procedure is branched. Now, we demonstrate each branch separately.

For the first branch, that results from updating profile $(s_2, s_1)$ in $\mGt{0}$, player $1$ has equilibrium strategy $s_2$ in $G^{1, (1a)}$. Similarly, player $2$ has equilibrium strategy $s_1$ in $G^{2, (1a)}$. Thus, the Adaptation Procedure provides a new misinformation game, say $\mGt{2a}$. Observe that the payoff matrices of $mG^{(1a)}$ are already updated with the correct value with respect to the bottom-left element, therefore $\mGt{2a} = \mGt{1a}$.

For the second branch, that results from updating profile $(s_2, s_2)$ in $\mGt{0}$, we have one pure Nash equilibrium in $G^{1,(1b)}$ for player $1$, that is $((0, 1), (0, 1))$. Furthermore, for player $2$ we have one mixed Nash equilibrium in $G^{2,(1b)}$ with strategy profile $((2/5,3/5), (1/3, 2/3))$. Thus, the $nme$ is $((0,1), (1/3, 2/3))$. Applying the characteristic strategy vector $\chi((0,1), (1/3, 2/3))$ yields ${(2,1), (2,2)}$, leading to further branching into new misinformation games $\mGt{2b}$ and $\mGt{2c}$.

Let us consider the element $(2, 2)$ of $\chi((0,1), (1/3, 2/3))$ (which leads to $\mGt{2c}$). We note that the payoff matrices of $mG^{(1b)}$ are already updated, specifically the bottom-left element, making $\mGt{2c} = \mGt{1b}$.

Similarly, for the element $(2,1)$ of $\chi((0,1), (1/3, 2/3))$, we update the bottom-right element of $P^{1, (1b)}$ and $P^{2, (1b)}$, so $\mGt{1b}$ leads to $\mGt{2b} = \langle G^{0},  G^{1, (2b)}, G^{2, (2b)} \rangle$ with payoff matrices detailed in Table~\ref{tbl:misninformation-game-IV}.

\begin{table}[h]
    \centering
    \caption[Payoff matrices in the misinformation game of the Example.]{\small Payoff matrices in the misinformation game of the Example~\ref{example:running example-I} at time $t=1$, using the $nme$ $((0,1), (1/3, 2/3))$.}
    \label{tbl:misninformation-game-IV}
    \begin{subtable}{.42\linewidth}
       \centering
        \begin{tabular}{ |c|c|c| }
         \hline
              & $s_1$ & $s_2$   \\
              \hline
              $s_1$ &  (2,2) & (0,3)   \\
             \hline
              $s_2$ &  (\textcolor{upsdellred}{7}, \textcolor{upsdellred}{2}) & (\textcolor{upsdellred}{1}, \textcolor{upsdellred}{1}) \\
             \hline
        \end{tabular}
        \caption{$P^{1,(2b)}$}\label{tbl:example-2b-view-IV-row}
    \end{subtable}
    \enspace
    \begin{subtable}{.42\linewidth}
       \centering
        \begin{tabular}{ |c|c|c| }
         \hline
              & $s_1$ & $s_2$   \\
              \hline
              $s_1$ &  (-1,1) & (2,-2)   \\
             \hline
              $s_2$ &  (\textcolor{upsdellred}{7}, \textcolor{upsdellred}{2}) & (\textcolor{upsdellred}{1}, \textcolor{upsdellred}{1}) \\
             \hline
        \end{tabular}
        \caption{$P^{2,(2b)}$}\label{tbl:example-2b-view-IV-col}
    \end{subtable}
\end{table}

\noindent In summary, we deduce that $\AD{(2)}{\{\mGt{0}\}} = \{ \mGt{2a}, \mGt{1b},  \mGt{2c} \}$.
\end{example}

\subsection{Stabilisation of the Adaptation Procedure}
\label{subsec:stable}

The following definition determines when the procedure is assumed to have ``terminated''; this corresponds to the time point where any further iterations do not provide new information to the players:

\begin{definition}\label{def:endingcriterion}
We say that the Adaptation Procedure \emph{terminates} (or \emph{stops}) at step $t$, iff $t$ is the smallest non-negative integer for which
\begin{equation*}
    \AD{(t+1)}{M} = \AD{(t)}{M}
\end{equation*}
for $t \in \mathbb{N}_0$. We call $t$ the length of the Adaptation Procedure and we denote it as $\LAD{M}$.
\end{definition}

In essence, the Adaptation Procedure concludes when all spawned misinformation games appear already in the input. We will demonstrate (see Theorem~\ref{thm:bounded mG}) that termination is a well-defined concept for all finite misinformation games, as they invariably reach a termination point.

To simplify presentation in the following, let $\AD{*}{M}$ represent the set of all misinformation games created by $M$ through successive applications of the \AD{}{} function, i.e., $\AD{*}{M} = \bigcup^{\infty}_{t=0} \AD{(t)}{M}$. Moreover, we will denote by $\AD{\infty}{M}$ the misinformation games produced by $\AD{}{}$ post-termination, i.e., $\AD{\infty}{M} = \AD{(t)}{M}$ for $t = \LAD{M}$. $\AD{\infty}{M}$ will be called the \emph{Stable Set}. 

\begin{definition}\label{def:sme}
Consider a misinformation game $mG$. Then, $\sigma$ is a \emph{stable misinformed equilibrium} (or \emph{sme} for short) of $mG$, iff 
there exists some $\widehat{mG} \in \AD{\infty}{\{ mG \}}$ such that
$\sigma \in NME(\widehat{mG})$ and, for all $\vec{v} \in \chi(\sigma)$, $\widehat{mG}_{\vec{v}} = \widehat{mG}$.
\end{definition}

We denote by $SME(mG)$ the $sme$s of $mG$. The careful reader might wonder why we didn't define $sme$s as the $nme$s of the stable set. The extra condition (for all $\vec{v} \in \chi(\sigma)$, $\widehat{mG}_{\vec{v}} = \widehat{mG}$) ensures that the $nme$ satisfies some condition of ``stability'' in the sense that the players already know the actual payoffs of the positions being played. Note that this is not necessarily true for all $nme$s, even when restricting ourselves to the $nme$s of $mG$s that appear in the stable set: an $nme$ could lead to learning new positions, creating different $mG$s, which happen to be in the stable set. Such an $nme$ is not an $sme$, according to Definition \ref{def:sme}; instead, an $nme$ is an $sme$ if and only if the players learn nothing because of it. The following example illustrates this scenario, among other things.

\newcolumntype{M}[1]{>{\centering\arraybackslash}m{#1}}
\newcolumntype{N}{@{}m{0pt}@{}}

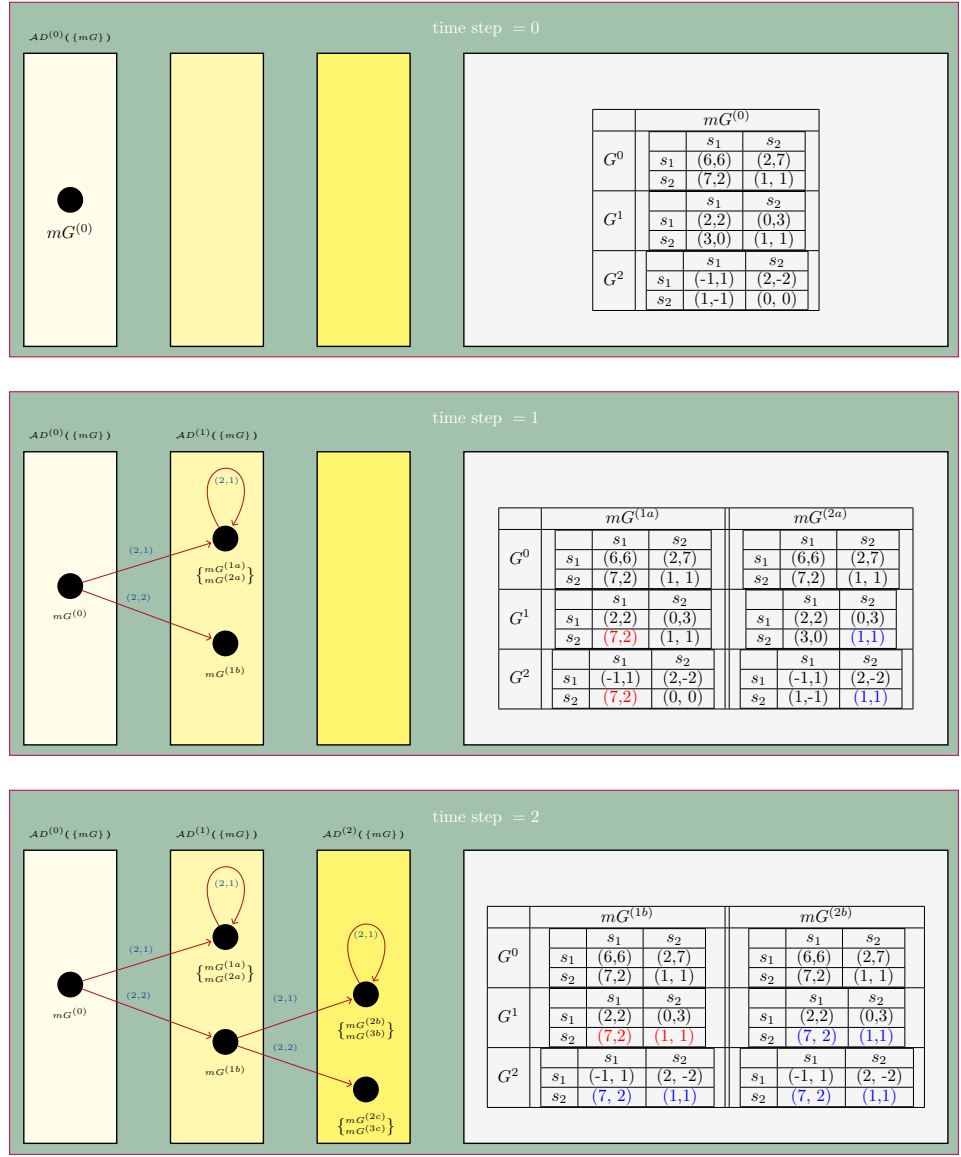
\begin{figure}
\centering
\begin{adjustbox}{scale=.65}
    \begin{tikzpicture}[
node distance = 7mm and 7mm,
     N/.style = {draw, draw=purple, fill=cambridgeblue,
                 minimum size=5mm, inner sep =2mm}, 
  FIT/.style = {draw, draw=yellow, fill=yellow!30,
                inner sep=2mm, fit=#1},           
every edge/.append style = {-{Straight Barb[scale=0.8]}, semithick}
                        ]
\node[N]   (w11) {

    \begin{tikzpicture}[->,>=stealth',shorten >=1pt,node distance=2.6cm,
        thick,main node/.style={circle,inner sep=0pt, fill = black,align=center,draw,font=\tiny},
        every loop/.style={min distance=8mm,in=60,out=100,looseness=18},
        set/.style={draw,rectangle, minimum width=1.7cm, minimum height=1.5cm,inner sep=0pt,align=center}]
    \centering
    
    \node[set,fill=yellow!40,text width=1.2cm,minimum width=1.9cm, minimum height=6cm, shift={(0, 0)},label={}] (nat) at (3,0)  (rea) {};
    
    \node[set,fill=yellow!70,text width=1.2cm,minimum width=1.9cm, minimum height=6cm, shift={(0, 0)}, label={}] (int) at (6,0)  {};
    
    \node[set,fill=yellow!10,text width=1.2cm,minimum width=1.9cm, minimum height=6cm, shift={(0, 0)}, label={[xshift=0cm, yshift=0cm]$\scriptscriptstyle\AD{(0)}{\{mG\}}$}] (nat) at (0,0) {};

    \node[set,fill=beige,text width=1.2cm, minimum width=9.9cm, minimum height=6cm, shift={(2, 0)}, label={[xshift=0cm, yshift=-5.5cm] \begin{tabular}{|c|c|}
       \hline
         &  $mG^{(0)}$ \\
         \hline
        $G^0$ & \begin{tabular}{|c|c|c|}
                          \hline
              & $s_1$ & $s_2$   \\
              \hline
              $s_1$ &  (6,6) & (2,7)   \\
             \hline
              $s_2$ &  (7,2) & (1, 1) \\
             \hline 
       \end{tabular} \\ 
       \hline
       $G^1$ & \begin{tabular}{|c|c|c|}
                          \hline
              & $s_1$ & $s_2$   \\
              \hline
              $s_1$ &  (2,2) & (0,3)   \\
             \hline
              $s_2$ &  (3,0) & (1, 1) \\
             \hline
       \end{tabular} \\
       \hline
       $G^2$ & \begin{tabular}{|c|c|c|}
                          \hline
              & $s_1$ & $s_2$   \\
              \hline
              $s_1$ &  (-1,1) & (2,-2)   \\
             \hline
              $s_2$ &  (1,-1) & (0, 0) \\
             \hline
       \end{tabular} \\
       \hline
    \end{tabular}}] (table1) at (11,0) {};

    \draw[-] (8.5,3.5) node {\textcolor{white}{$\text{time step } = 0$}};
    \node[main node,text width=.25cm, label=below:{ $mG^{(0)}$}] (11) {};
 	
\end{tikzpicture}

};
\begin{scope}[node distance = 7mm and 0mm]
\node[N, below=of w11] (w21) {

\begin{adjustbox}{scale=1}
\begin{tikzpicture}[->,>=stealth',shorten >=1pt,node distance=2.6cm,
        thick,main node/.style={circle,inner sep=0pt, fill = black,align=center,draw,font=\tiny},
        every loop/.style={min distance=8mm,in=60,out=120,looseness=18},
        set/.style={draw,rectangle, minimum width=1.7cm, minimum height=1.5cm,inner sep=0pt,align=center}]
    \centering
    
    \node[set,fill=yellow!40,text width=1.2cm,minimum width=1.9cm, minimum height=6cm, shift={(0, 0)},label={[xshift=0cm, yshift=0cm]$\scriptscriptstyle\AD{(1)}{\{mG\}}$}] (nat) at (3,2.5)  (rea) {};
    
    \node[set,fill=yellow!70,text width=1.2cm,minimum width=1.9cm, minimum height=6cm, shift={(0, 0)}, label={}] (int) at (6,2.5)  {};
    
     \node[set,fill=yellow!10,text width=1.2cm,minimum width=1.9cm, minimum height=6cm, shift={(0, 0)}, label={[xshift=0cm, yshift=0cm]$\scriptscriptstyle\AD{(0)}{\{mG\}}$}] (nat) at (0,2.5) {};

    \node[set,fill=beige,text width=1.2cm,minimum width=9.9cm, minimum height=6cm, shift={(2, 0)}, label={[xshift=0cm, yshift=-5.5cm] \begin{tabular}{|c|c||c|}
        \hline
         &  $mG^{(1a)}$ & $mG^{(2a)}$\\
         \hline
        $G^0$ & \begin{tabular}{|c|c|c|} 
                          \hline
              & $s_1$ & $s_2$   \\
              \hline
              $s_1$ &  (6,6) & (2,7)   \\
             \hline
              $s_2$ &  (7,2) & (1, 1) \\
             \hline 
       \end{tabular} &  \begin{tabular}{|c|c|c|} 
                          \hline
              & $s_1$ & $s_2$   \\
              \hline
              $s_1$ &  (6,6) & (2,7)   \\
             \hline
              $s_2$ &  (7,2) & (1, 1) \\
             \hline 
       \end{tabular} \\ 
       \hline
       $G^1$ & \begin{tabular}{|c|c|c|}
                          \hline
              & $s_1$ & $s_2$   \\
              \hline
              $s_1$ &  (2,2) & (0,3)   \\
             \hline
              $s_2$ &  \textcolor{red}{(7,2)} & (1, 1) \\
             \hline
       \end{tabular} &  \begin{tabular}{|c|c|c|} 
                          \hline
              & $s_1$ & $s_2$   \\
              \hline
              $s_1$ &  (2,2) & (0,3)   \\
             \hline
              $s_2$ &  (3,0) & \textcolor{blue}{(1,1)} \\
             \hline 
       \end{tabular} \\
       \hline
       $G^2$ & \begin{tabular}{|c|c|c|}
                          \hline
              & $s_1$ & $s_2$   \\
              \hline
              $s_1$ &  (-1,1) & (2,-2)   \\
             \hline
              $s_2$ & \textcolor{red}{(7,2)} & (0, 0) \\
             \hline
       \end{tabular} &  \begin{tabular}{|c|c|c|} 
                          \hline
              & $s_1$ & $s_2$   \\
              \hline
              $s_1$ &  (-1,1) & (2,-2)   \\
             \hline
              $s_2$ &  (1,-1) & \textcolor{blue}{(1,1)} \\
             \hline 
       \end{tabular}\\
       \hline
    \end{tabular}}] (table2) at (11,2.5) {};

    \node[main node,text width=.25cm, label=below:{$\scriptscriptstyle mG^{(0)}$}] (21) {};
    
    \node[main node,text width=.25cm, label=below:{$ \scriptscriptstyle\stirlingii{mG^{(1a)}}{mG^{(2a)}}$}] (22) [above right= .6cm and 2.8cm of 21] {};
    
    \node[main node,text width=.25cm, label=below:{$ \scriptscriptstyle mG^{(1b)}$}] (23) [below right= .8cm and 2.8cm of 21] {};

    \node[] at (1.45, 0.75)   (a3) {\textcolor{yaleblue}{$\scriptscriptstyle (2,1)$}};
    \node[] at (3.2, 2.2)   (b3) {\textcolor{yaleblue}{$\scriptscriptstyle (2,1)$}};
    \node[] at (1.4, -0.2)   (c3) {\textcolor{yaleblue}{$\scriptscriptstyle (2,2)$}};
    \path[upsdellred, -{Latex[length=2mm]}]
 	    (22) edge [loop above] node {} (22)
		  (21) edge node {} (22)
		  (21) edge node {} (23);

    \draw[-] (8.5,3.5) node {\textcolor{white}{$\text{time step } = 1$}};

\end{tikzpicture}
\end{adjustbox}

};

\node[N, below=of w21] (w31) {

\begin{adjustbox}{scale=1}
\begin{tikzpicture}[->,>=stealth',shorten >=1pt,node distance=2.6cm,
        thick,main node/.style={circle,inner sep=0pt, fill = black,align=center,draw,font=\tiny},
        every loop/.style={min distance=8mm,in=60,out=120,looseness=18},
        set/.style={draw,rectangle, minimum width=1.7cm, minimum height=1.5cm,inner sep=0pt,align=center}]
    \centering
    
    \node[set,fill=yellow!40,text width=1.2cm,minimum width=1.9cm, minimum height=6cm, shift={(0, 0)},label={[xshift=0cm, yshift=0cm]$\scriptscriptstyle\AD{(1)}{\{mG\}}$}] 
            (nat) at (3,2.5)  (rea) {};
    \node[set,fill=yellow!70,text width=1.2cm,minimum width=1.9cm, minimum height=6cm, shift={(0, 0)}, label={[xshift=0cm, yshift=0cm]$\scriptscriptstyle \AD{(2)}{\{mG\}}$}] (int) at (6,2.5)  {};
     \node[set,fill=yellow!10,text width=1.2cm,minimum width=1.9cm, minimum height=6cm, shift={(0, 0)}, label={[xshift=0cm, yshift=0cm]$\scriptscriptstyle\AD{(0)}{\{mG\}}$}] (nat) at (0,2.5) {};

    \node[set,fill=beige,text width=1.2cm,minimum width=9.9cm, minimum height=6cm, shift={(2, 0)}, label={[xshift=0cm, yshift=-5.5cm]    \begin{tabular}{|c|c||c|}
        \hline
         &  $mG^{(1b)}$ & $mG^{(2b)}$\\
         \hline
        $G^0$ & \begin{tabular}{|c|c|c|} 
                          \hline
              & $s_1$ & $s_2$   \\
              \hline
              $s_1$ &  (6,6) & (2,7)   \\
             \hline
              $s_2$ &  (7,2) & (1, 1) \\
             \hline 
       \end{tabular} & \begin{tabular}{|c|c|c|} 
                          \hline
              & $s_1$ & $s_2$   \\
              \hline
              $s_1$ &  (6,6) & (2,7)   \\
             \hline
              $s_2$ &  (7,2) & (1, 1) \\
             \hline 
       \end{tabular} \\
       \hline
       $G^1$ & \begin{tabular}{|c|c|c|}
                          \hline
              & $s_1$ & $s_2$   \\
              \hline
              $s_1$ &  (2,2) & (0,3)   \\
             \hline
              $s_2$ &  \textcolor{red}{(7,2)} &\textcolor{red} {(1, 1)} \\
             \hline
       \end{tabular} &  \begin{tabular}{|c|c|c|} 
                          \hline
              & $s_1$ & $s_2$   \\
              \hline
              $s_1$ &  (2,2) & (0,3)   \\
             \hline
              $s_2$ &  \textcolor{blue}{(7, 2)} & \textcolor{blue}{(1,1)}  \\
             \hline 
       \end{tabular} \\
       \hline
       $G^2$ & \begin{tabular}{|c|c|c|} 
                          \hline
              & $s_1$ & $s_2$   \\
              \hline
              $s_1$ &  (-1, 1) & (2, -2)   \\
             \hline
              $s_2$ &  \textcolor{blue}{(7, 2)} & \textcolor{blue}{(1,1)} \\
             \hline 
       \end{tabular} &  \begin{tabular}{|c|c|c|} 
                          \hline
              & $s_1$ & $s_2$   \\
              \hline
              $s_1$ &  (-1, 1) & (2, -2)   \\
             \hline
              $s_2$ &  \textcolor{blue}{(7, 2)} & \textcolor{blue}{(1,1)} \\
             \hline 
       \end{tabular}\\
       \hline
    \end{tabular}}] (table3) at (11,2.5) {};
    
    \node[main node,text width=.25cm, label=below:{$ \scriptscriptstyle mG^{(0)}$}] (31) {};
 	\node[main node,text width=.25cm, label=below:{$ \scriptscriptstyle\stirlingii{mG^{(1a)}}{mG^{(2a)}}$}] (32) [above right= .6cm and 2.8cm of 31] {};
 	\node[main node,text width=.25cm, label=below:{$ \scriptscriptstyle mG^{(1b)}$}] (33) [below right= .8cm and 2.8cm of 31] {};
 	\node[main node,text width=.25cm, label=below:{$ \scriptscriptstyle\stirlingii{mG^{(2b)}}{mG^{(3b)}}$}] (34) [above right= .6cm and 2.5cm of 33] {};
 	\node[main node,text width=.25cm, label=below:{$ \scriptscriptstyle\stirlingii{mG^{(2c)}}{mG^{(3c)}}$}] (35) [below right= .6cm and 2.5cm of 33] {};
    \node[] at (1.45, 0.75)   (a3) {\textcolor{yaleblue}{$\scriptscriptstyle (2,1)$}};
    \node[] at (3.2, 2.1)   (b3) {\textcolor{yaleblue}{$\scriptscriptstyle (2,1)$}};
     \node[] at (1.4, -0.2)   (c3) {\textcolor{yaleblue}{$\scriptscriptstyle (2,2)$}};
    \node[] at (4.4, -.25)   (c3) {\textcolor{yaleblue}{$\scriptscriptstyle (2,1)$}};
    \node[] at (4.4, -1.25)   (d3) {\textcolor{yaleblue}{$\scriptscriptstyle (2,2)$}};
    \node[] at (6.05, 1.05)   (e3) {\textcolor{yaleblue}{$\scriptscriptstyle (2,1)$}};
 	 \path[upsdellred, -{Latex[length=2mm]}]
 	    (32) edge [loop above] node {} (22)
 	    (34) edge [loop above] node {} (24)
		(31) edge node {} (32)
		(31) edge node {} (33)
		(33) edge node {} (35)
		(33) edge node {} (34);
\draw[-] (8.5,3.5) node {\textcolor{white}{$\text{time step } = 2$}};

\end{tikzpicture}
\end{adjustbox}

};
    \end{scope}

    \end{tikzpicture}

\end{adjustbox}
\caption[ Schematic representation of Example. ]{\small Schematic representation of Example~\ref{example:running example-I}.}
    \label{fig:example adaptation}
\end{figure}

\addtocounter{example}{-1}
\begin{example}[continued]
\label{example:running example-V}
For $t = 2$, let us first consider $\mGt{2a}$. As mentioned, $\mGt{2a} = \mGt{1a}$, leading to $\mGt{3a} = \mGt{1a}$, as previously discussed in Example \ref{example:running example-IV}.

Analogously for $\mGt{2c}$, we note that $\mGt{2c} = \mGt{1b}$ so, again, from the analysis in the previous steps of Example \ref{example:running example-IV}, it branches into $\mGt{3a} = \mGt{1b}$, $\mGt{3b} = \mGt{2b}$.

Finally, for $\mGt{2b}$, we have that both players' subjective games, $G^{1,(2b)}$ and $G^{2, (2b)}$, have a pure Nash equilibrium with strategy profile $((0, 1)), (1, 0))$. Thus, the $nme$ of $\mGt{2b}$ is $\sigma = ((0, 1), (1, 0))$, for which $\chi(\sigma) = \{(2, 1)\}$. 
Observe that this position in $\chi(\sigma)$ is known to the players, i.e., it holds that $P^{i,(2b)}_{\vec{v}} = P^0_{\vec{v}}$, for $i\in \{1,2\}$, $\vec{v} = \{(2,1)\}$. Thus, $\mGt{3c} = \mGt{2b}$.

Combining the above, we observe that $\AD{(3)}{\{\mGt{0}\}} = \{ \mGt{1}, \mGt{2b} \} = \AD{(2)}{\{\mGt{0}\}}$, so the Adaptation Procedure terminates at step $2$, i.e., $\mathfrak{L}_{\AD{}{\{\mGt{0}\}}} = 2$.

Now let us identify the $sme$s of $\mGt{0}$. As explained above, and in the previous steps of Example \ref{example:running example-IV}, $NME(\mGt{1a}) = \{ ((0,1), (1,0)) \}$, and $NME(\mGt{2a}) = \{ ((0,1), (1,0)) \}$. Thus, the profile $\sigma = ((0,1), (1,0))$ is an $sme$.
However, the sole $nme$ of $\mGt{1b}$, $((0, 1), (1/3, 2/3))$, does not meet the criteria for an $sme$ since $mG_{\vec v} \notin \AD{\infty}{mG}$ for $\vec v = (2,1)$, and thus $\mGt{1b}$ fails the second condition of Definition \ref{def:sme}.
\end{example}

\subsection{Adaptation Procedure: Properties}
\label{sec:adaptation-procedure: properties}

\subsubsection{General properties}
\label{subsec:general-properties}

We will start by providing two useful properties of the Adaptation Procedure. The first shows that the Adaptation Procedure is ``local'', i.e., the different misinformation games in the input of $\AD{}{.}$ do not interact with each other:

\begin{proposition}\label{prop:additive}
For any set of misinformation games $M$, $\AD{}{M} = \bigcup_{mG\in M} \AD{}{\{mG\}}$.
\end{proposition}

The next two results essentially show that an $mG$ that appears in any step of the Adaptation Procedure cannot reappear in any subsequent step, unless it stays there permanently:

\begin{proposition}\label{prop:mG eq mG'}
Take some finite sequence $mG_1, \dots, mG_n$ such that $mG_{i+1} \in \AD{}{\{ mG_i \}}$, for $i \in [n-1]$ and $mG_1 \in \AD{}{\{ mG_n \}}$. Then, $mG_i = mG_j$ for all $i, j$.
\end{proposition}

\begin{proposition}\label{prop:mG-in-terminal-set}
For any two misinformation games $mG$, $mG'$ and $t \geq 0$, if $mG' \in \AD{(t)}{mG}$ and $mG' \in \AD{}{mG'}$ then $mG' \in \AD{\infty}{mG}$.
\end{proposition}

From Definition \ref{def:sme}, not all misinformation games in the Stable Set contribute an $sme$. In other words, some of the games in $\AD{\infty}{mG}$ are irrelevant when it comes to computing the $sme$s. Moreover, the termination criterion outlined in Definition \ref{def:endingcriterion} indicates a ``global'' situation, where the sets of misinformation games are identical across two successive steps. The set-theoretic comparison of two sets of misinformation games would be computationally expensive, so a more ``local'' termination criterion is preferable. Therefore, an interesting question is the following: can we know, while examining a single misinformation game, whether it is safe to discard it from future iterations? This problem is addressed with the notion of the terminal set.

\begin{definition}
    \label{def:terminal-set}
    Let $\mGt{0}$ be a misinformation game. 
    We define the \emph{terminal set} of the Adaptation Procedure on $\mGt{0}$ as follows:
    \begin{equation*}
        \displaystyle \mathcal{T} = \lbrace mG \in \AD{*}{\mGt{0}} \mid mG \in \AD{}{mG} \rbrace.
    \end{equation*}
\end{definition}

\addtocounter{example}{-1}
\begin{example}[continued]
\label{example:running example-VI}
The terminal set in Example~\ref{example:running example-II} is $\mathcal{T} = \{\mGt{1a}, \mGt{2b}\}$.
\end{example}

The next proposition shows that the games in the terminal set are also in the Stable Set; in addition, they are the only members of the Stable Set that matter when it comes to computing $sme$s:

\begin{proposition}
    \label{prop:algo-terminal-set-sme}
    Let $\mathcal{T}$ be the terminal set of the Adaptation Procedure on $\mGt{0}$.
    Then:
    \begin{itemize}
        \item $\mathcal{T} \subseteq \AD{\infty}{\mGt{0}}$
        \item For any $\sigma \in SME(\mGt{0})$, there exists $mG\in \mathcal{T}$ such that $\sigma \in NME(mG)$.    
    \end{itemize}
\end{proposition}

Note that the opposite of the second bullet of Proposition \ref{prop:algo-terminal-set-sme} does not hold: a misinformation game in the terminal set may contain $NME$s that are not themselves $SME$s, because they ``lead'' to another misinformation game, thereby violating the condition ``for all $\vec{v} \in \chi(\sigma)$, $\widehat{mG}_{\vec{v}} = \widehat{mG}$'' of Definition \ref{def:sme}.

Consider the implications of Proposition \ref{prop:algo-terminal-set-sme} on the computational aspects of the Adaptation Procedure. When a game $mG$ is not a member of the terminal set $\mathcal{T}$, our interest shifts to its outcomes, specifically $\AD{}{mG}$, rather than the game itself, since it will not yield an $sme$. In contrast, if $mG$ is part of the terminal set, we retain it without further processing, as it will perpetually recur (as well as its outcomes), rendering any additional processing superfluous. On the other hand, the games in $\AD{}{mG} \setminus \{mG\}$ may be relevant, if any of them (or any of their descendants) is in the terminal set.

\subsubsection{Termination, and existence of smes}
\label{subsec:termination-existence}

In this subsection we show some results related to the termination of the Adaptation Procedure.

We start by showing that the Adaptation Procedure will always terminate when $mG$ is finite. In particular, when $mG$ is finite, then at any given point in the Adaptation Procedure a finite number of positions can be learnt, and the positions that remain to be learnt are also finite. Therefore the length of the procedure is bounded, as well as the total number of $mG$s that can be generated at any step, so the set generated by the Adaptation Procedure is finite and the Adaptation Procedure terminates. This result is formally phrased in the next propositions.

\begin{proposition}
    \label{thm:bounded mG}
    Consider a finite $\vert N \vert$-player canonical misinformation game $mG$ on $\vert S \vert$ strategies. For the length $\LAD{mG}$ of the Adaptation Procedure on $mG$, we have,
    \begin{equation}
        \label{eq:bounded mG}
        \LAD{mG} \leq \vert S \vert
    \end{equation}
\end{proposition}

\begin{proposition}
    \label{prop:finite-mgs}
    Let $mG$ be a $\vert N \vert$-player canonical misinformartion game on $\vert S \vert$ strategy profiles. Then for the number $\vert \AD{*}{\{mG\}} \vert$ of all the misinformation games produced during the Adaptation Procedure on $mG$, it holds that,
    \begin{equation}
        \label{eq:finite-mgs}
        \vert \AD{*}{\{mG\}} \vert = O\left( \min\left[\vert S \vert^{\LAD{mG} + 1}, 2^{\vert S \vert}\right] \right)\footnote{
            When $\LAD{mG} \ll \vert S \vert$, then $\vert S \vert^{\LAD{mG} + 1} < 2^{\vert S \vert}$. On the other hand, when $\vert S \vert \approx \LAD{mG}$ or $\vert S \vert \gg \LAD{mG}$, we have $2^{\vert S \vert} < \vert S \vert^{\LAD{mG} + 1}$.
        }.
    \end{equation}
\end{proposition}

\begin{corollary}
    \label{cor:ad-infty-bound}
    Assume a $\vert N \vert$-player canonical misinformation game $mG$, on $\vert S \vert$ strategies. Then the stable set $\AD{\infty}{mG}$ is also finite.
\end{corollary}

Unfortunately, a similar result cannot be shown for infinite games. The following counter-example proves this fact:

\addtocounter{example}{+1}
\begin{example}
\label{ex:inf-games counter-example}
Consider $G^0 = \langle N, S, P^0 \rangle$ such that $N = \{ r, c\}$, $S_r = S_c = \{1,2,\dots\}$, $S = S_r \times S_c$, and the payoff for a position $(x,y) \in S_r \times S_c$ is computed as follows: 
\[
P^0_{(x,y)} = \left(\frac 1 x, \frac 1 y \right)
\]
As a result of this definition, the only (pure) Nash equilibrium for $G^0$ is in position $(1,1)$, where the payoff is $(1,1)$.

Now consider the canonical misinformation game $mG = \langle G^0, G^r, G^c \rangle$, where $G^r = \langle N, S, P^r \rangle$ such that $P^r_{(x,y)} = (\frac 1 x + 1, \frac 1 y)$, and $G^c = G^0$. By the definition of $mG$ we note that player $c$ knows the correct payoffs and will always play strategy $1$ (and will never learn anything from the Adaptation Procedure). On the other hand, player $r$ knows the correct payoffs as far as player $c$ is concerned, but her own payoffs are distorted, and she believes that their actual value is 1 point more than they really are. The key observation here is that, for player $r$, any of her subjective payoffs is better than any of her actual ones. Therefore, when she learns any position, this position becomes highly unattractive and cannot be selected again.

More formally, we note that $NME(mG) = \{ \sigma_0 \}$, where $\chi(\sigma_0) = \{ (1,1) \}$. Thus, $\AD{(1)}{\{ mG \}} = \{ \mGt{1} \}$, where $\mGt{1} = mG_{(1,1)}$. It is easy to see that $NME(\mGt{1}) = \{ \sigma_1 \}$, where $\chi(\sigma_1) = \{ 2,1 \}$. Continuing this process, we observe that $\AD{(i)}{ \{ mG \}} = \{\mGt{i}\}$, where $\mGt{i}$ is such that the positions $(1,1), (2,1), \dots, (i,1)$ have been learnt. But the only (pure) $nme$ of $\mGt{i}$ corresponds to the position $(i+1,1)$. As a result, the Adaptation Procedure will continuously lead to the learning of new positions (and, thus, to new misinformation games), which shows that the Adaptation Procedure will not terminate.\qed
\end{example}

Due to this negative result, all subsequent analysis focuses on finite misinformation games.

The following theorem characterises a Stable Set as the closure of the terminal set with respect to the adaptation operator $\AD{}{.}$.
    
\begin{theorem}[Stable Set Characterisation]
\label{theo:stable-set-characterisation}
Let $\AD{\infty}{mG}$ be the Stable Set and $\mathcal{T}$ the terminal set of the Adaptation Procedure on $\mGt{0}$. Then,
\begin{equation*}
    \label{eq:stable-set-characterisation}
    \AD{\infty}{mG} = \AD{*}{\mathcal{T}}.
\end{equation*}
\end{theorem}

Intuitively the above theorem shows that a misinformation game $mG$ belongs the Stable Set of the Adaptation Procedure if it is either in the terminal set itself or has an ancestor that belongs to the terminal set. Moreover, Theorem \ref{theo:stable-set-characterisation} shows that the terminal set contains all the information of the Adaptation Procedure on a misinformation game $mG^0$. Not only all the $sme$s of the Adaptation Procedure are located somewhere between the elements of $\mathcal{T}$ (see Proposition \ref{prop:algo-terminal-set-sme}), but also we can \emph{reconstruct} the Stable Set $\AD{\infty}{mG^0}$, from the terminal set $\mathcal{T}$. As we shall see in Section \ref{sec:computing-adaptation-procedure}, keeping track of the terminal set suffices to compute the Adaptation Procedure.

We will now show that all finite misinformation games have an $sme$. To start with, we show the following result, which describes a condition sufficient for the existence of an $sme$. In particular, Proposition \ref{prop:mG leq mG'} states that if the players in some misinformation game $mG'$ learn nothing new from the respective $nme$s (i.e., if $\AD{}{mG'} = \{mG'\}$, then the $nme$s of $mG'$ are also $sme$s. Formally:

\begin{proposition}\label{prop:mG leq mG'}
If $mG' \in \AD{*}{\{mG\}}$ and $\AD{}{\{ mG' \}} = \{mG'\}$, then $NME(mG') \subseteq SME(mG)$.
\end{proposition}

\begin{proposition}
\label{prop:existence}
If $mG$ is finite, then $SME(mG) \neq \emptyset$.
\end{proposition}

\noindent An immediate observation from Proposition~\ref{prop:existence} is the following,
\begin{corollary}\label{cor:terminal set non emptyness}
If $mG$ is finite, then $\mathcal{T} \neq \emptyset$.
\end{corollary}

Proposition \ref{theo:sme-complexity} below shows that the upper bound on $\LAD{mG}$ provided by equation \ref{eq:bounded mG} is tight.

\begin{proposition}
\label{theo:sme-complexity}
Let $N$ be a set of players and $S$ a set of strategies. Then there is a $\vert N \vert$-player canonical misinformation game $mG$, on $\vert S \vert$ strategies, such that,
\begin{equation*}
    \label{eq:comp-1-sme}
    \LAD{mG} = \vert S \vert.
    \end{equation*}
\end{proposition}

\section{Computing the Adaptation Procedure}
\label{sec:computing-adaptation-procedure}

To test the properties of the Adaptation Procedure in practice, we have developed an algorithm to see how the Adaptation Procedure behaves across a spectrum of randomly generated misinformation games, as well as its computational properties. In this section, we describe the algorithm and its implementation, as well as the corresponding experiments and their results.

\subsection{The adaptation graph}
\label{subsec:adaptation-graph}

We encode the computation of the Adaptation Procedure as a graph search problem. We will see that the Adaptation Procedure operator induces a structure on the set of all misinformation games $\AD{*}{mG}$ derived from the Adaptation Procedure on $mG$. We call this structure \emph{adaptation graph}.

\begin{definition}
    \label{def:adaptation-graph}
    Let $mG$ be a misinformation game, and the set of all misinformation games $\AD{*}{mG}$ derived from the Adaptation Procedure on $mG$. We denote with $\Gamma = (\AD{*}{mG}, E)$ the directed graph where,
    \begin{equation}
        \label{eq:adaptation-graph}
        E = \left\{(mG^{1}, mG^{2}) \mid \text{ there is } \vec{v} \in \chi(\text{NME}(mG^{1})), \text{ s.t. } mG^{2} = (mG^{1})_{\vec{v}}\right\}.
    \end{equation}
    Further, we denote by $\Gamma^\prime = (\AD{*}{mG}, E^\prime)$ the loopless version of $\Gamma$ where,
    \begin{equation}
        \label{eq:adapt-graph-dag}
        E^\prime = E \setminus \left\{(mG^{1}, mG^{2}) \in E \mid mG^{1} = mG^{2}\right\}.
    \end{equation}
\end{definition}

Some notes on Definition \ref{def:adaptation-graph} are in order. First, by Corollary \ref{cor:ad-infty-bound}, $\Gamma$ (and $\Gamma^\prime$) are finite. 
Moreover, comparing Definition \ref{def:adaptation-graph} with Definition \ref{def:AD(M)-adaptation procedure iterative process}, we note that the edge $(mG^1,mG^2)$ belongs to the adaptation graph $\Gamma$ if and only if $mG^2 \in \AD{}{mG^1}$.
Further, the definition of the terminal set $\mathcal{T}$ (Definition \ref{def:terminal-set}) implies that each misinformation game $mG$ that belongs to $\mathcal{T}$ has a self-loop in $\Gamma$. Moreover, as the following proposition shows, $\Gamma'$ is a \emph{directed acyclic graph (DAG)}, i.e., $\Gamma$ is ``almost'' a DAG.

\begin{proposition}
    \label{prop:adapt-graph-dag}
    Let $\Gamma = (\AD{*}{mG}, E)$ be the adaptation graph of the Adaptation Procedure on a misinformation game $mG$ and $\Gamma^\prime$ its loopless version. Then $\Gamma^\prime$ is a directed acyclic graph.
\end{proposition}

Focusing on $\Gamma^\prime$, we can explore more effectively the adaptation graph $\Gamma$. In the following analysis we will use the concept of the degree of a vertex, where $d^-_{\Gamma}(v)$ is the number of the incoming edges and $d^+_{\Gamma}(v)$ is the number of the outgoing edges in graph $\Gamma$.

\begin{proposition}
    \label{prop:adapt-graph-source}
    Let $\Gamma = (\AD{*}{mG}, E)$ be the adaptation graph, of the Adaptation Procedure, on a misinformation game $mG$ and $\Gamma'$ its loopless version. Then, $mG$ is the \emph{only source node} of $\Gamma^\prime$. That is, $mG$ is the only node in $\Gamma^\prime$ such that
    \begin{equation}
        \label{eq:adapt-graph-source}
        d^{-}_{\Gamma^\prime}(mG) = 0.
    \end{equation}
\end{proposition}

Moreover, we are able to characterize the \emph{sink} nodes of $\Gamma^\prime$ as members of the terminal set $\mathcal{T}$, and we argue that \emph{every} $nme$ of such a node is an $sme$.

\begin{proposition}
    \label{prop:adapt-graph-sinks}
    Let $\Gamma = (\AD{*}{mG}, E)$ be the adaptation graph of the Adaptation Procedure on a misinformation game $mG$ and $\Gamma'$ its loopless version. Let $K$ be the set of sink nodes in $\Gamma^\prime$ that is
    \begin{equation}
        \label{eq:adapt-graph-sinks}
        K = \{mG^\prime \in \AD{*}{mG} \mid d^{+}_{\Gamma^\prime}(mG) = 0\}.
    \end{equation}
    Then, $K \subseteq \mathcal{T}$. Moreover, every $nme$ of some $mG^\prime \in K$ is an $sme$ of the Adaptation Procedure on $mG$.
\end{proposition}

Note that the reverse direction of Proposition \ref{prop:adapt-graph-sinks} is not true. That is, in general, $K \subsetneq \mathcal{T}$, and there are $sme$s which are not $nme$s of a sink node.
Continuing our exploration of the structure of $\Gamma^\prime$, we remark that the \emph{longest} path in $\Gamma^\prime$ is $\LAD{mG}$.

\begin{proposition}
    \label{prop:adapt-graph-longest-path}
    Let $\Gamma = (\AD{*}{mG}, E)$ be the adaptation graph of the Adaptation Procedure on a misinformation game $mG$ and $\Gamma'$ its loopless version. Then, the length of the longest path in $\Gamma^\prime$ is $\LAD{mG}$.
\end{proposition}

Let us summarize what the above results say about the form of the adaptation graph $\Gamma$. First, the adaptation graph is ``almost'' a DAG; when we omit the self-loops from $\Gamma$, the resulting graph $\Gamma^\prime$ is a DAG. The graph $\Gamma^\prime$ begins with a single source node, namely the root of the Adaptation Procedure. Additionally, $\Gamma^\prime$ extends for $\LAD{mG}$ length units, concluding to a set of sink nodes, each of which belongs to the terminal set $\mathcal{T}$. Moreover, every $nme$ of a sink node is an $sme$. The ideas highlighted here will be exploited in the next subsection to develop algorithms for computing either one, or all $sme$s; as we shall see, there is an intriguing asymmetry between these two problems (computing a single $sme$, and computing all $sme$s), that we explore next.

\subsection{Set-based representation: efficiently detecting self-loops}
\label{subsec:syntactic-semantic-implementation}

In the previous subsection we captured the graph structure induced on the space of misinformation games by the Adaptation Procedure. We named this structure the \emph{adaptation graph} $\Gamma$. We also showed how we can characterise the terminal set of the Adaptation Procedure as the nodes in $\Gamma$ that contain self-loops. Moreover, we observed that we can omit these self-loops during graph exploration, resulting in the graph $\Gamma^\prime$. Since the self-loop detection constitutes one of the main operations that will be performed at each step of the graph exploration, we want it to be as efficient as possible. In this direction, we define the \emph{set-based representation} $\widetilde{\Gamma}$ of the adaptation graph $\Gamma$, showing that this reduces the computational cost of self-loop detection. Lastly, we will explore $\widetilde{\Gamma}$, the same way we did with $\Gamma$, omitting the self-loops and thus resulting in the graph $\widetilde{\Gamma}^\prime$.

Assume a misinformation game $mG^\prime$ in the adaptation graph $\Gamma$. In order to detect whether $mG$ contains a self-loop, we need to determine whether there exists some $\sigma \in NME(mG)$ and some $\vec v \in \chi(\sigma)$ such that $mG = mG_{\vec v}$. To determine whether $mG = mG_{\vec v}$, we need to verify that in both misinformation games, the players have equal subjective and objective payoffs. The computational complexity of this operation is given by the following result.

\begin{proposition}
\label{prop:semantic-equivalence-complexity}
Let $mG^1, mG^2$ be two canonical $\vert N \vert$-player misinformation games with $\vert S \vert$ joint strategies. Then the computational complexity to determine if $mG^1 = mG^2$ is
\begin{equation}
    \label{eq:semantic-equivalence-complexity}
    \Theta\left( \vert N \vert^2 \cdot \vert S \vert \right).
\end{equation}
\end{proposition}

To reduce the computational cost of self-loop detection, we will adopt a representation of a misinformation game that allows us to omit equality checks, when possible. In this direction, we will consider the \emph{set-based representation} $\widetilde{\Gamma}$ of the adaptation graph $\Gamma$. Recall that $S$ is the set of all position vectors in a misinformation game $mG$. Consider a mapping $\theta \colon 2^S \to \AD{*}{mG}$ where,
\begin{equation}
    \label{eq:sets-to-mgs-mapping}
    \theta[X] = mG_X.
\end{equation}

The mapping $\theta[\cdot]$ is surjective by construction but not injective, as different sets of position vectors may result in the same misinformation game, e.g., when the subjective payoffs of the players coincide with the objective payoffs at this position vector. Using $\theta[\cdot]$ we can define the \emph{set-based representation} $\widetilde{\Gamma}$.

\begin{definition}
\label{def:adapt-graph-syntactic-approximation}
Let $\Gamma = (\AD{*}{mG}, E)$ be the adaptation graph, of the Adaptation Procedure, on a misinformation game $mG$. We denote with $\widetilde{\Gamma} = (2^S, \widetilde{E})$ the set-based representation of $\Gamma$, where
\begin{equation*}
    \widetilde{E} = \left\{ (X,Y) \mid Y = X \cup \{\vec v\}, X \neq Y \text{ and the edge } (\theta[X], \theta[Y]) \text{ exists in } \Gamma \right\}.
\end{equation*}
Further, we denote by $\widetilde{\Gamma}' = (2^S, \widetilde{E}')$ the loopless version of $\widetilde{\Gamma}$ where,
\begin{equation*}
    \widetilde{E}' = \widetilde{E} \setminus \left\{ (X,Y) \in \widetilde{E} \mid \theta[X] = \theta[Y] \right\}.
\end{equation*}
\end{definition}

Note that, with respect to Definition \ref{def:adapt-graph-syntactic-approximation}, the vertex set of $\widetilde{\Gamma}$ contains all the position vectors sets, i.e. the whole powerset $2^S$. From Proposition \ref{prop:finite-mgs}, in the worst case, all the position vector sets will be connected to the root of the Adaptation Procedure. On the other hand, in general, in $\widetilde{\Gamma}$ there will be isolated nodes, i.e. nodes with zero in and out degree. Subsequently, we implicitly discard the isolated nodes, and only consider the connected component of $\widetilde{\Gamma}$ to which the root of the Adaptation Procedure belongs. The existence of isolated nodes, does not affect the subsequent worst case analysis of our algorithms.

In Figure \ref{fig:Set-based adaptation graph} we depict an adaptation graph $\Gamma$ and its corresponding set-based representation $\widetilde{\Gamma}$. In $\widetilde{\Gamma}$, the misinformation games are represented as position vector sets, which are implemented as sorted lists. Additionally, we keep a dictionary that implements $\theta[\cdot]$, by mapping position vector sets to misinformation games. Before performing an equivalence check between two misinformation games, we check whether their set-based representations are equal. Note that two position vector sets $X$, $Y$ that are connected with an edge in $\widetilde{\Gamma}$, differ by \emph{at most} one element $\vec{v}$. Given the position vector $\vec{v}$, we can perform a set equivalence check between $X$, $Y$, by only checking whether $\vec{v}$ belongs to $X$. This can be done in sub-linear time, with a simple binary search (given that position vector sets are implemented as sorted lists).
    
\begin{proposition}
\label{prop:child-syntactic-equivalence-complexity}
Let $\Gamma = (\AD{*}{mG}, E)$ be the adaptation graph of the Adaptation Procedure on a canonical misinformation game $mG$, and $\widetilde{\Gamma} = (2^S, \widetilde{E})$ be the set-based representation of $\Gamma$. Assume two nodes $X, Y$ of $\widetilde{\Gamma}$, and a position vector $\vec{v} \in \chi(NME(\theta[X]))$, such that $\theta[Y] = \theta[X]_{\vec{v}}$. The set equivalence $X = Y$ can be performed in
\begin{equation}
    O(\log \vert S \vert)
\end{equation}
time, where $\vert S \vert$ is the number of joint strategies.
\end{proposition}

As we see in the sequel, the position vector $\vec{v}$, required by Proposition \ref{prop:child-syntactic-equivalence-complexity}, is always known, as it is the position vector used to apply the update operation to the parent to obtain the child. Our previous observations regarding $\Gamma$ and $\Gamma'$ apply also for $\widetilde{\Gamma}$ and $\widetilde{\Gamma}'$. In Table \ref{tab:adaptation graphs} we summarize the characteristics of the graphs $\Gamma$, $\Gamma'$, $\tilde{\Gamma}$, and $\tilde{\Gamma}'$ (visualised in Figure \ref{fig:diagram adaptation graph}).

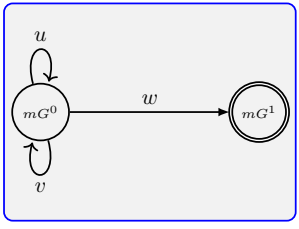
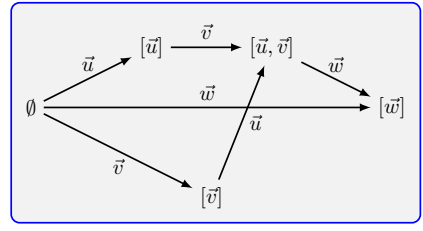
\begin{figure*}[t!]
    \centering
    \begin{subfigure}[t]{0.42\textwidth}
        \centering
\begin{adjustbox}{scale=.83}
\begin{tikzpicture}[transform shape,remember picture]

    \node[circle,draw,thick] (C) at (10,10) {\tiny $mG^0$};
    \node (C1) at (10,11.5) {};
    \node (C2) at (10,8.5) {};
    \node[circle,draw,thick,double] (N) at (13.5,10) {\tiny $mG^1$};

    \node at (10,11.2) (c3) {};
    \node at (10,8.6) (c4) {};
    \node at (11.75,10) (c5) {};
    
    \path[thick,-latex] (C) edge node[above] {$w$} (N);
    \path[thick,-latex] (C) edge [loop above] node {$u$} (C);
    \path[thick,-latex] (C) edge [loop below] node {$v$} (C);

    \begin{pgfonlayer}{background}
       \node[draw=blue,rounded corners,thick,fill=gray!10,fit=(C)(C1)(N)(C2)] {};
    \end{pgfonlayer}
    \end{tikzpicture}
    \end{adjustbox}
    \caption{Adaptation graph $\Gamma$.}\label{fig:set-based-ad-graph}
    \end{subfigure}%
    \hfil 
    \begin{subfigure}[t]{0.58\textwidth}
        \centering
  \begin{adjustbox}{scale=.8}
  \begin{tikzpicture}[transform shape,remember picture]
    \node at (0,0) (b1) {$\emptyset$};
    \node at (2,1) (b2) {$\left[\vec{u}\right]$};
    \node at (2,1.5) (b6) {};
    \node at (3,-1.5) (b3) {$\left[\vec{v}\right]$};
    \node at (6,0) (b4) {$\left[\vec{w}\right]$};
    \node at (4,1) (b5) {$\left[\vec{u}, \vec{v}\right]$};

    \node at (0.9,.6) (b7) {};
    \node at (1.5,-.75) (b8) {};
    \node at (3,0.3) (b9) {};
    \node at (2.8,1.4) (b10) {};
    \node at (5,.6) (b11) {};
    \node at (3.5,-.25) (b12) {};

    \draw[-latex,thick] (b1) -- (b2) node[midway,above] {$\vec{u}$};
    \draw[-latex,thick] (b1) -- (b3)
    node[midway,below] {$\vec{v}$};
    \draw[-latex,thick] (b1) -- (b4)
    node[midway,above] {$\vec{w}$};
    \draw[-latex,thick] (b2) -- (b5)
    node[midway,above] {$\vec{v}$};
    \draw[-latex,thick] (b3) -- (b5)
    node[midway,right] {$\vec{u}$};
    \draw[-latex,thick] (b5) -- (b4)
    node[midway,above] {$\vec{w}$};

    \begin{pgfonlayer}{background}
        \node[draw=blue,rounded corners,thick,fill=gray!10,fit=(b1)(b2)(b6)(b3)(b4)(b5)] {};
    \end{pgfonlayer}
    \end{tikzpicture}
    \end{adjustbox}
    \caption{a connected component of the set-based representation graph $\tilde{\Gamma}$.}
    \label{fig:BFS-tree}
    \end{subfigure}
    \caption{Correspondence between the adaptation graph $\Gamma$ and $\widetilde{\Gamma}$.}
\label{fig:Set-based adaptation graph}
\end{figure*}

\begin{figure*}[t!]
    \centering
    \begin{tikzpicture}[transform shape,scale=1.2]
        \filldraw[pastelred] (0.02,0.02) rectangle (2,2);
        \filldraw[pastelyellow] (-0.02,-0.02) rectangle (-2,-2);
        \filldraw[pastelblue] (-0.02,0.02) rectangle (-2,2);
        \filldraw[pastelgreen] (0.02,-0.02) rectangle (2,-2);
        
        \draw[latex-latex] (2,0)--(-2,0);
        \draw[latex-latex] (0,-2)--(0,2);
        
        \coordinate (A1) at (-2,0);
        \node at (A1) [below right = .5mm and .1mm of A1] {\tiny $\substack{\textbf{pos} \\ \textbf{vec}}$};
        \node at (A1) [above right = .5mm and .1mm of A1] {\tiny $mG$s};
        \coordinate (A2) at (0,2);
        \node at (A2) [below left = .1mm and .5mm of A2] {\tiny $\substack{\textbf{self} \\ \textbf{loops}}$};
        \node at (A2) [below right = .1mm and .5mm of A2] {\tiny $\substack{\textbf{no self} \\ \textbf{loops}}$};;
        
        \node (B1) at (-1,-1) {$\widetilde{\Gamma}$};
        \node (B2) at (-1,1) {$\Gamma$};
        \node (B3) at (1,1) {$\Gamma'$};
        \node (B4) at (1,-1) {$\widetilde{\Gamma}'$};
        
        \draw[->,thick] (B2)--(B1);
        \draw[->,thick] (B2)--(B3);
        \draw[->,thick] (B1)--(B4);
        
        \end{tikzpicture}
        \caption{Correspondence between all variations of adaptation graph.}
\label{fig:diagram adaptation graph}
\end{figure*}
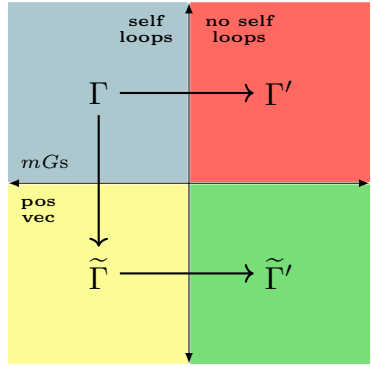

\begin{table}[]
    \centering
\resizebox{\linewidth}{!}{
   {\scriptsize
    \begin{tabular}{>{\centering\arraybackslash}p{3.5cm}>{\centering\arraybackslash}p{8.7cm}}
         \hline
          Graph & Edges \\
         \hline
         \hline
       $\Gamma = (\AD{*}{mG}, E)$ & $E = \left\{(mG^{1}, mG^{2}) \bigm| \begin{array}{l} \text{ there is } \vec{v} \in \chi(\text{NME}(mG^{1})), \\ \text{ s.t. } mG^{2} = (mG^{1})_{\vec{v}} \end{array} \right\}$ \\
       \\
       $\Gamma' = (\AD{*}{mG}, E')$ & $E^\prime = E \setminus \{(mG^{1}, mG^{2}) \in E \mid mG^{1} = mG^{2}\}$ \\ 
       \\
       $\widetilde{\Gamma} = (2^F, \widetilde{E})$ & 
        $\widetilde{E} = \left\{ (X,Y) \bigm|
         \begin{array}{l}  
         Y = X \cup \{\vec v\}, X \neq Y \\ 
         \text{ and the edge } (\theta[X], \theta[Y]) \text{ exists in } \Gamma  
         \end{array} 
         \right\}$\\
         \\
        $\widetilde{\Gamma}' = (2^F, \widetilde{E}')$ &  $\widetilde{E}' = \widetilde{E} \setminus \{ (X,Y) \in \widetilde{E} \mid \theta[X] = \theta[Y] \}$ \\
        \\
    \hline
    \end{tabular}
    }
    }
    \caption{Taxonomy of the adaptation graphs.}
    \label{tab:adaptation graphs}
\end{table}

Before concluding this section, it is important to note that the use of the set equivalence checks does not improve the worst-time analysis. If a set equivalence check fails, a subsequent equality check between the misinformation games is necessary. In the worst-case, each set equivalence check would be followed by a misinformation game equality check. However, in practice, we have observed considerable runtime improvements. This practical benefit justifies our decision to base our subsequent algorithms on this approach.

Additionally, recall that we cannot explore the Adaptation Procedure solely based on the set-based representation graph $\widetilde{\Gamma}$. Since we need the underlying misinformation games both for computing the $nme$s and for performing equality checks, we need to store the $\theta[\cdot]$ mapping. In the algorithms presented below, $\theta[\cdot]$ is dynamically mapped to a dictionary. The dictionary is expanded whenever a new misinformation game is discovered. For the details involved in implementing $\theta[\cdot]$ see \cite{PVF}. Moreover, in the sections below, we insist on the notation we used here: given a set of misinformation games $\mathcal X$, the set of position vectors that correspond to a set of strategy profiles $\mathcal{S}$ in $\mathcal{X}$ is $\widetilde{\mathcal{X}} := \{\chi(\sigma) \mid \text{for } \sigma \in \mathcal{S} \text{ in } \mathcal{X}\}$. For simplicity we will refer to $\widetilde{\mathcal{X}}$ as $\chi(\mathcal{S}(\mathcal{X}))$. Furthermore, in the following we abuse the notation $\theta[\mathcal{X}]$, where $\mathcal{X}$ is a set of position vector sets. Naturally, $\theta[\mathcal{X}] = \{\theta[X] \mid X \in \mathcal{X}\}$.

\subsection{Computing ALL-SMEs}
\label{subsec:computing all sme}

In this subsection, we present and analyse an algorithm for computing all the stable misinformed equilibria of a misinformation game. We state the problem formally below.

\problem{All-SMEs}{prob:all-smes}{
    $mG$, a $N$-player canonical misinformation game on $\vert S \vert$ joint strategies.
}{
    $SME$, the set of all stable misinformed equilibria of $mG$.
}

We organize our algorithm in three routines. The main routine is \texttt{AdaptationProcedure} (Algorithm \ref{algo:algo-1}), which uses two subroutines, $\texttt{TraverseAdaptationGraph}(mG)$ (Algorithm \ref{algo:algo-2}) and $\texttt{ComputeSME}(\mathcal{T})$ (Algorithm \ref{algo:algo-3}).

\paragraph{Main routine: Adaptation Procedure}
\label{par:main algorithm}

Algorithm \ref{algo:algo-1} implements the Adaptation Procedure in two phases. In the first phase (line \ref{line:algo-1-phase-1}), we use Algorithm \ref{algo:algo-2} to traverse the adaptation graph using the set-based representation described  previously. That is we traverse $\widetilde{\Gamma}$.
The output of that routine is the terminal set (see Definition \ref{def:terminal-set}) and the mapping $\theta[\cdot]$. Subsequently, the second phase (line \ref{line:algo-1-phase-2}) scans the terminal set to identify the stable misinformed equilibria by calling Algorithm \ref{algo:algo-3}. In order to do that, the second phase utilizes the mapping $\theta[\cdot]$ computed in the first phase.

\begin{algorithm}[h]
\caption{\texttt{AdaptationProcedure}}
\label{algo:algo-1}
\begin{algorithmic}[1]
    \Require A root misinformation game $mG$.
    \Statex
    \State $\widetilde{\mathcal{T}}, \theta[\cdot] \gets \texttt{TraverseAdaptationGraph}(mG)$ \label{line:algo-1-phase-1}
    \State $SME \gets \texttt{ComputeSME}(\widetilde{\mathcal{T}}, \theta[\cdot])$ \label{line:algo-1-phase-2}
    
    \State \Return $SME$
\end{algorithmic}
\end{algorithm}

\paragraph{Phase 1: Traverse Adaptation Graph}
\label{par:traverse adaptation graph}

Algorithm \ref{algo:algo-2} ($\texttt{TraverseAdaptationGraph}(mG)$) essentially employs a breadth-first traversal of the set-based representation $\widetilde{\Gamma}$ of the adaptation graph to compute the terminal set. To do so, we maintain two sets of sets of position vectors, namely $\widetilde{\mathcal{Q}}$ (a queue of the nodes that remain to be processed) and $\widetilde{\mathcal{V}}$ (the set of visited nodes). Moreover, we incrementally compute the $\theta[\cdot]$ mapping as we traverse the graph.

In the beginning, the set-based representation of the terminal set $\widetilde{\mathcal{T}}$ is empty and the sets $\widetilde{\mathcal{Q}}, \widetilde{\mathcal{V}}$ start with the value $\{ \emptyset \}$, as we initially have to process the root misinformation game $mG$. We initialize the $\theta[\cdot]$ dictionary with the root misinformation game and compute the root's $nme$s. In each iteration of the \textbf{while} loop in lines \ref{algo-Traverse-Adaptation-Graph-start-while}-\ref{algo-Traverse-Adaptation-Graph-end-while}, we select (arbitrarily) a set of position vectors $W \in \widetilde{\mathcal{Q}}$, and remove $W$ from $\widetilde{\mathcal{Q}}$. The set $W$ denotes the node undergoing processing during the current iteration of the \textbf{while} loop. Note that  for each new set $X$, we compute $NME(\theta[X])$ before inserting $X$ to the queue $\widetilde{\mathcal{Q}}$.

To process node $W \in \widetilde{\mathcal{Q}}$, we compute the expansions $X$ of $W$. An expansion $X$ of $W$ corresponds to $W \cup \{ \vec{v} \}$, for some position vector $\vec{v} \in  \chi(NME(\theta[W]))$.
Note that it is possible to have a position vector $\vec{v} \in  \chi(NME(\theta[W])$, which also belongs to $W$.

In this case, $W$ has a self-loop in $\widetilde{\Gamma}$, thus constituting a terminal node. Lines \ref{algo-step:syntactic-equivalence-1}--\ref{algo-step:syntactic-equivalence-2} perform this check, and if positive, add $W$ to the terminal set. Additionally, $X$ is another set-based representation for the same misinformation game. Thus $\theta[\cdot]$ is updated to reflect that fact.

For the remaining expansions $X = W \cup \{\vec{v}\}$, with $\vec{v} \notin W$, we firstly check if $X$ is visited, in lines \ref{algo-step:if visit X}-\ref{algo-step: end if visit X}. If $X$ is visited, we proceed to the next expansion. Otherwise, we append $X$ to the set of visited nodes $\widetilde{\mathcal{V}}$, line \ref{algo-step: append X}. Recall that, we may have $\theta[W] = \theta[X]$, even if $W \neq X$. For instance, the position vector $\vec{v}$ may correspond to a position of the payoff matrices, already known to the players. Hence, an additional equivalence check between misinformation games should be performed in lines \ref{algo-step:semantic-equivalence-1}-\ref{algo-step:semantic-equivalence-2}. If the equivalence check proves negative, then $mG_X$ is a new misinformation game. $mG_X$ is added to the dictionary $\theta[\cdot]$, and its $nme$s are computed. Lastly, $X$ is inserted into the queue $\widetilde{\mathcal{Q}}$.

The algorithm concludes when there are no further unprocessed nodes ($\widetilde{\mathcal{Q}} = \emptyset$), in which case the complete adaptation graph has been explored, and all misinformation games that must be added to the terminal set are already there. The terminal set is returned in line \ref{algo-Traverse-Adaptation-Graph-return-T}.

    \begin{algorithm}[h]
    \caption{$\texttt{TraverseAdaptationGraph}(mG)$}
    \label{algo:algo-2}
    \begin{algorithmic}[1]
        \Require A root misinformation game $mG$.

        \Statex
        \State $\widetilde{\mathcal{T}} \gets \emptyset$,
        $\widetilde{\mathcal{Q}} \gets \{ \emptyset \}$,
        $\widetilde{\mathcal{V}} \gets \{\emptyset\}$

        \State $\theta[\emptyset] \gets mG$
        \State compute $NME(\theta[\emptyset])$
        
        \While {$\widetilde{\mathcal{Q}} \neq \emptyset$}
        \label{algo-Traverse-Adaptation-Graph-start-while}
            
            \State Pick some $W \in \widetilde{\mathcal{Q}}$
            \State $\widetilde{\mathcal{Q}} \gets \widetilde{\mathcal{Q}} \setminus \{W\}$

            \If{$\exists \vec{v} \in \chi(NME(\theta[W])) \cap W$} 
            \label{algo-step:syntactic-equivalence-1}
                \State $\widetilde{\mathcal{T}} \gets \widetilde{\mathcal{T}} \cup \{W\}$ 
                \Comment{Self-loop detected}
            \EndIf \label{algo-step:syntactic-equivalence-2}
               
            \ForEach{position vector $\vec{v} \in \chi(NME(\theta[W])) \setminus W$}
                \State $X \gets W \cup \{\vec{v}\}$
                
                \If{$X \in \widetilde{\mathcal{V}}$} \label{algo-step:if visit X}
                    \Comment{Check if $X$ revisited}
                    \State \textbf{continue}
                \EndIf \label{algo-step: end if visit X}
                
                \State $\widetilde{\mathcal{V}} \gets \widetilde{\mathcal{V}} \cup \{X\}$  \label{algo-step: append X}
                \State $\#$ Compute the update operation resulting in $mG_X$.
                \If{$mG_{X} = mG_W$} 
                    \label{algo-step:semantic-equivalence-1}
                    \State $\theta[X] \gets \theta[W]$
                    \Comment{New set-based representation}
                    \Statex \Comment{for same misinformation games}
                    \State $\widetilde{\mathcal{T}} \gets \widetilde{\mathcal{T}} \cup \{W\}$
                    \Comment{Self-loop detected}%
                    \label{algo-step:semantic-equivalence-2}
                \Else
                    \State $\theta[X] \gets mG_X$
                    \Comment{New misinformation game}
                    \State compute $NME(\theta[X])$
                    \State $\widetilde{\mathcal{Q}} \gets \widetilde{\mathcal{Q}} \cup \{X\}$
                        
                \EndIf
            \EndFor
        \EndWhile
        \label{algo-Traverse-Adaptation-Graph-end-while}
        \State \Return $\widetilde{\mathcal{T}}$, $\theta[\cdot]$
        \label{algo-Traverse-Adaptation-Graph-return-T}
    \end{algorithmic}
    \end{algorithm}

\paragraph{Phase 2: Compute SMEs}
\label{par:compute stable misinformed equilibria}

The role of Algorithm \ref{algo:algo-3} is to determine the $sme$s of the Adaptation Procedure. Given the set-based representation of the terminal set $\widetilde{\mathcal{T}}$, it traverses the terminal set and checks whether each $nme$ in each of the misinformation games in the terminal set satisfies Definition \ref{def:sme}; if so, it is an $sme$ and it is added to the list of $sme$s (maintained in variable $SME$). Proposition \ref{prop:algo-terminal-set-sme} guarantees that this procedure will not miss any of the stable misinformed equilibria.

\begin{algorithm}[h]
\caption{$\texttt{ComputeSME}(\widetilde{\mathcal{T}})$}
\label{algo:algo-3}
\begin{algorithmic}[1]
        
    \Require The terminal set $\widetilde{\mathcal{T}}$ and the mapping $\theta[\cdot]$

    \Statex
    \State $SME \gets \emptyset$
    
    \ForEach {$X \in \widetilde{\mathcal{T}}$} 
        \ForEach {$\sigma \in NME(\theta[X])$}
            \If {for all position vectors $\vec{v} \in \chi(\sigma)$, we have $\theta[X] = \theta[X \cup \vec{v}]$}
                \State $SME \gets SME \cup \{\sigma\}$
                \Comment{Following Definition \ref{def:sme}.}
            \EndIf
        \EndFor
    \EndFor
    \State \Return $SME$
\end{algorithmic}
\end{algorithm}

\paragraph{Correctness}

Now we discuss the complexity of algorithms \ref{algo:algo-1}, \ref{algo:algo-2}, \ref{algo:algo-3}. Observe that the correctness of algorithms \ref{algo:algo-1}, \ref{algo:algo-3} depends solely on the correctness of Algorithm \ref{algo:algo-2}. Indeed, assuming that Algorithm \ref{algo:algo-2} correctly computes the terminal set $\mathcal{T}$, then Algorithm \ref{algo:algo-3} correctly computes the $sme$s, since it simply verifies Definition \ref{def:sme} for each element in the terminal set. Moreover, the correctness of Algorithm \ref{algo:algo-1} is implied by Proposition \ref{prop:algo-terminal-set-sme}. Thus, we only need to validate the correctness of Algorithm \ref{algo:algo-2}. In this direction, we need to show that Algorithm \ref{algo:algo-2} terminates when and only when the end criterion of Definition \ref{def:endingcriterion} is met, and after considering all the misinformation games in $\AD{*}{mG}$.

To simplify the notation of the following proofs we use sets of misinformation games instead of their set-based representations. Consider the sets $\widetilde{\mathcal{T}},\widetilde{\mathcal{Q}}, \widetilde{\mathcal{V}}$ of Algorithm \ref{algo:algo-2}. In the following we consider the corresponding sets of misinformation games, i.e., $\mathcal{T}, \mathcal{Q}, \mathcal{V}$. Recall that we can transition from a set of sets of position vectors to a set of misinformation games, using the mapping $\theta[\cdot]$, e.g., $\mathcal{Q} = \theta[\widetilde{\mathcal{Q}}]$.

We begin by stating a recursive definition of the queue $\mathcal{Q}$, which is respected by Algorithm \ref{algo:algo-2}. Observe that only the misinformation games that do not have a self-loop in the adaptation graph $\Gamma$ are inserted in the queue $\mathcal{Q}$. Therefore, for the $i$-th iteration of the \textbf{while} loop the following holds.

\begin{equation}
    \label{eq:queue-rec}
    \left\{\begin{split}
    & \mathcal{Q}_0       = \{mG^0\},\\
    & \mathcal{Q}_{i+1}   = \AD{}{\mathcal{Q}_i} \setminus \mathcal{T}
    \end{split}\right.
\end{equation}

Let $\mathcal{T}_i = \mathcal{Q}_i \cap \mathcal{T}$. Intuitively, $\mathcal{T}_i$ contains the terminal sets computed in the $i$-th iteration. Using the recursive equations in \eqref{eq:queue-rec}, we can rewrite the recursive equation in Definition \ref{def:AD(M)-adaptation procedure iterative process} as,
\begin{equation}
    \label{eq:stable-rec}
    \left\{\begin{split}
    &  M_0       = \mathcal{Q}_0,\\
    &  M_{i+1}   = \mathcal{Q}_{i+1} \cup \AD{i+1}{\bigcup_{k=1}^{i+1} \mathcal{T}_k}.
\end{split}\right.
\end{equation}
To see that, note that $\mathcal{Q}_{i+1}$ is produced from $\mathcal{Q}_i$, we apply the adaptation operator on $\mathcal{Q}_i$, neglecting the terminal nodes in $\mathcal{Q}_i$. To properly compute $M_{i+1}$, we need to reinstantiate the omitted terminal nodes and apply the adaptation operator on them. Noteworthy, Algorithm \ref{algo:algo-2} terminates when $\mathcal{Q}_{i} = \emptyset$. We argue that $\mathcal{Q}_{i} = \emptyset$ if and only if $M_i = M_{i+1}$, i.e., the end criterion of Definition \ref{def:endingcriterion} holds. In order to prove that, we first need the following lemma.
        
\begin{lemma}
\label{lem:correctness}
If $\mathcal T_i = \mathcal Q_i \cap \mathcal T$, where $\mathcal Q_i$ are defined as in equation \eqref{eq:queue-rec}, then it holds $\mathcal{T} = \bigcup_{i=1}^\infty \mathcal{T}_i$.
\end{lemma}

The above lemma states that every node in the terminal set will be considered by Algorithm \ref{algo:algo-2}, since, at some point, will belong to the queue $\mathcal{Q}$. The following proposition shows that when the queue $\mathcal{Q}$ becomes empty, i.e. the graph exploration terminates, then the end criterion of the Adaptation Procedure is met. Thus proving the correctness of Algorithm \ref{algo:algo-2}.

\begin{proposition}[Correctness of Algorithm \ref{algo:algo-2}]
\label{prop:algo-3-correctness}
Consider the equations \eqref{eq:queue-rec}, \eqref{eq:stable-rec}. If  $\mathcal{Q}_{i} = \emptyset$, then $M_{i+1}$ is a Stable Set, and  $\bigcup_{k=1}^{i+1} \mathcal{T}_k = \mathcal{T}$. 
\end{proposition}

\paragraph{Complexity}

We now proceed to the time analysis of our method. We begin from Algorithm \ref{algo:algo-2} since it does most of the heavy lifting. For the time complexity of Algorithm \ref{algo:algo-2}, it suffices to see that we are essentially doing a breadth-first search in the set-based representation $\widetilde{\Gamma}$ of the adaptation graph. Thus, the computational complexity is bounded by the number of vertices. We state this complexity rigorously in the following proposition.

\begin{proposition}
\label{prop:graph-traversal-complexity}
Assume a canonical $\vert N \vert$-player misinformation game $mG$, on $\vert S \vert$ joint strategies. Then Algorithm \ref{algo:algo-2} will terminate after
\begin{equation}
    O\left( 2^{\vert S \vert} \vert N \vert^3 \mathtt{t}(\textbf{NASH}) \vert S \vert^{2} \right).
\end{equation}
\end{proposition}

The following theorem summarises our analysis of algorithms \ref{algo:algo-1}, \ref{algo:algo-2}, \ref{algo:algo-3}. Note that the total complexity of Algorithm \ref{algo:algo-1} is the sum of the complexities of Algorithms \ref{algo:algo-2}, \ref{algo:algo-3}.

\begin{theorem}[Correctness and Complexity of Algorithm \ref{algo:algo-1}]
\label{theo:complexity-corectness}
Algorithm \ref{algo:algo-1} correctly computes \emph{all} the $sme$s of a canonical $\vert N \vert$-player misinformation game $mG^0$, on $\vert S \vert$ joint strategies in time
\begin{equation}
    \label{eq:complexity-corectness}
    O\left( \vert N \vert^2 \vert S \vert^{2} \left[ \vert \widetilde{\mathcal{T}} \vert + 2^{\vert S \vert} \vert N \vert \mathtt{t}(\textbf{NASH}) \right] \right)
\end{equation}
\end{theorem}

In the above theorem, we provide an upper bound for the computation of the \emph{all} the $sme$s. Since the input to the algorithm is a misinformation game, the size of the input is $\vert N \vert \cdot \vert S \vert$. Therefore, from equation \ref{eq:complexity-corectness} our algorithm is exponential to the size of the input. Moreover, Algorithm \ref{algo:algo-1} performs an exponential number of calls to an oracle computing the Nash equilibria of a $\vert N \vert$-player normal-form game. This computational hardness seems to be unavoidable since the Adaptation Procedure (see Definition \ref{def:AD(M)-adaptation procedure iterative process}) leads, in the worst case, to an exponential explosion of the number of misinformation games in $\AD{*}{mG}$. We will see in the next subsection that computing a single $sme$ appears to be a much easier problem.

\subsection{Computing 1-SME}
\label{subsec:comput-1-sme}

In this subsection, we will introduce a simple algorithm for computing a single $sme$ of the Adaptation Procedure on a misinformation game $mG$. First, the computation of a single $sme$ problem is articulated as follows

\problem{1-SME}{prob:1-smes}{
    $mG$, a canonical $\vert N \vert$-player misinformation game on $\vert S \vert$ joint strategies.
}{
    $sme$, a (possibly mixed) strategy profile that constitutes a stable misinformed equilibrium.
}

\noindent The rationale of our algorithm is stated in the following corollary.

\begin{corollary}
    \label{cor:maximal-path}
    Let $\Gamma = (\AD{*}{mG}, E)$ be the adaptation graph of the Adaptation Procedure on a misinformation game $mG$. Also, let $\Gamma^\prime$ be the resulting graph after omitting the self-loops from $\Gamma$. Let $mG^\prime$ be the last node on a \emph{maximal} path in $\Gamma^\prime$, starting from $mG$. Then, all the $nme$s of $mG^\prime$ are $sme$s.
\end{corollary}

From the above corollary, it suffices to find a maximal path in $\Gamma^\prime$. Intuitively, we can do that by iteratively adapting a single misinformation game, beginning from the root $mG$. We present this simple method in Algorithm \ref{algo:find-1-sme}. In the following theorem, we discuss the computational complexity of Algorithm \ref{algo:find-1-sme}.

\begin{algorithm}[h]
    \caption{$\texttt{Find-1-SME}(mG)$}
    \label{algo:find-1-sme}
    \begin{algorithmic}[1]
        \Require A root misinformation game $mG$.
        \Statex
        
        \While {True}
            \State compute $NME(mG)$
            \State choose some $\sigma \in NME(mG)$
            \If{$\exists \vec{v} \in \chi(\sigma)$ such that $mG \neq (mG)_{\vec{v}}$}
                \State $mG \gets (mG)_{\vec{v}}$
            \Else
                \State \Return $\sigma$
            \EndIf
        \EndWhile
         
    \end{algorithmic}
\end{algorithm}

\begin{theorem}
\label{theo:complexity-1-sme-upper-bound}
Consider the canonical $\vert N \vert$-player misinformation game $mG$, on $\vert S \vert$ joint strategies. Then for the computational complexity of the Algorithm \ref{algo:find-1-sme} it holds
\begin{align*}
    O\left( \vert S \vert^{2} \cdot \vert N \vert^3  \cdot t(\textbf{NASH}) \right). \label{eq:comp-1-sme}
\end{align*}
\end{theorem}

Using Theorems \ref{theo:complexity-corectness}, \ref{theo:complexity-1-sme-upper-bound} we can directly compare the \textbf{ALL-SMEs} problem, which involves computing all the $sme$s of a misinformation game, and the \textbf{1-SME} problem, which involves computing a single $sme$. In both problems, we use an oracle to the \textbf{NASH} problem to compute all the Nash equilibria of a normal-form game. Since \textbf{NASH} is a \textbf{PPAD}-complete problem, both \textbf{ALL-SMEs} and \textbf{1-SME} are computationally intractable. However, there is a qualitative difference between these problems. In \textbf{ALL-SMEs}, we will be making an exponential number of calls to the \textbf{NASH} oracle (with respect to the size of the input $\vert S \vert$), whereas in \textbf{1-SME}, we make only a polynomial number of calls. This difference in the number of $\textbf{NASH}$ oracle calls comes from the fact that Algorithm \ref{algo:algo-2} explores the entire adaptation graph, while Algorithm \ref{algo:find-1-sme} explores only a small portion of it, corresponding to a single path in the adaptation graph. The latter is achieved by \emph{greedily} expanding a single misinformation game. Such a greedy approach cannot be applied to the \textbf{ALL-SMEs} problem.

\subsection{Implementation}
\label{subsec:Implementation}

In this section, we review the implementation details of the algorithms outlined in Subsections \ref{subsec:computing all sme} and \ref{subsec:comput-1-sme}. We will focus on three important parts of Algorithms \ref{algo:algo-2}, \ref{algo:algo-3}. Namely, the implementation of the $\mathcal{Q}, \mathcal{V}$ data structures, the computation of the update operation, and the computation of the Nash equilibria. We close this section by briefly discussing a parallel variant of Algorithms \ref{algo:algo-2}, \ref{algo:algo-3}.

Before addressing the details mentioned above, we make some remarks on the architecture of our program. The Algorithms \ref{algo:algo-1}, \ref{algo:algo-2}, \ref{algo:algo-3} are implemented in \textsc{Python}. On the other hand, the update operation and the computation of Nash equilibria form \textsc{Python} subprocesses that communicate with the main program via pipes. This allows for modularity and better maintainability of our code. The subprocess computing the update operation is a logic program, written in the Answer Set Programming language \textsc{Clingo}\footnote{\url{https://potassco.org/clingo/}}. Further, we use the \textsc{Gambit}~\cite{GAMBIT} package to compute Nash equilibria in $\vert N \vert$-player normal-form games.

\paragraph{The $\mathcal{Q}, \mathcal{V}$ data structures}

The implementation of the $\mathcal{Q}, \mathcal{V}$ data structures significantly influences both the behaviour and the efficiency of Algorithm \ref{algo:algo-2}. The set $\mathcal{Q}$ essentially denotes a dynamic linear structure of the sets of position vectors to be processed by the algorithm. The implementation of this structure impacts the way the algorithm traverses the graph of misinformation games. Implementing $\mathcal{Q}$ as a \emph{first in-first out} queue will induce a breadth-first search strategy, whereas implementing $\mathcal{Q}$ as a \emph{first in-last out} stack will induce a depth-first search strategy. Both choices are adequate since they offer constant-time insertions and extractions.

Perhaps the most significant concern for the implementation of $\mathcal{V}$ will be to support efficient membership queries (see Algorithm \ref{algo:algo-2}, line \ref{algo-step: append X}). Therefore, a suitable choice would be a \emph{dictionary}, with $O(\log \vert \mathcal{V} \vert)$ membership queries, or a hash table, with \emph{randomized} constant membership queries.

Taking into account the above remarks, we implemented $\mathcal{Q}$ as a simple queue, and $\mathcal{V}$ as a dictionary because it supports efficient worst-case membership queries.

\paragraph{Implementing the update operation}
\label{par:answer-set-programming}

We implemented the update operation of Definition \ref{def:update_operation} using the Answer Set Programming language \textsc{Clingo}. As the update operation captures the main logic of the program, it is suitable to encode the axioms of Definition \ref{def:update_operation} using a Logic Programming language. A notable feature of our algorithm is its ability, via the \textsc{Clingo} program, to ascertain whether a child misinformation game diverges from its parent. Therefore, the \emph{equality check} of lines \ref{algo-step:semantic-equivalence-1}--\ref{algo-step:semantic-equivalence-2} in Algorithm \ref{algo:algo-2} is realized within the logic program.

In order to encode Definition \ref{def:update_operation} as a logic program, some technicalities are involved. The payoff function is encoded using the predicate $\mathtt{u(G, P, SP, U)}$ of arity four. It signifies the payoff value ($\mathtt{U}$) awarded to player $\mathtt{P}$ in the context of player $\mathtt{G}$'s normal-form game when the strategy profile $\mathtt{SP}$ is played. Note that \textsc{Clingo} does not support lists natively\footnote{Unlike other Logic Programming languages, e.g., Prolog.}. Therefore, we use symbolic functions to encode finite lists. For example, we encode the strategy profile $(1, 2, 3)$, with $\mathtt{sp}(1, \mathtt{sp}(2, \mathtt{sp}(3, \mathtt{nul})))$, where $\mathtt{sp}/2$ is a symbolic function with two arguments and $\mathtt{nul}/0$ is a constant.

\begin{table}
\begin{lstlisting}[
                    backgroundcolor = \color{teal!10},
                    breaklines=true,
                    showstringspaces=false,
                    basicstyle=\footnotesize\ttfamily,
                    language=Prolog,
                    caption=The Update Operation's Logic program (Pseudocode).,
                    label=list:impl-1
                ]
  % the new payoff function if pos(SP).
v(G, P, SP, V) :- pos(SP), u(0, P, SP, U), V = U.

 % the new payoff function if not pos(SP).

v(G, P, SP, V) :- not pos(SP), u(G, P, SP, U), V = U.

 % check if child-mG != parent-mG

changed :- v(G, P, SP, V), u(G, P, SP, U), V != U.
unchanged :- not changed.
\end{lstlisting}
\end{table}

In Listing \ref{list:impl-1} we present the logic program for applying the update operation of Definition \ref{def:update_operation}. The logic program takes as input a list of payoff predicates $\mathtt{u}/4$ and a single predicate $\mathtt{pos(SP)}$ describing a position vector as a finite list. Then, it outputs (proves) a list of updated payoff predicates $\mathtt{v}/4$, which follow the syntax of $\mathtt{u}/4$. For each strategy profile $\mathtt{SP}$ we check whether the predicate $\mathtt{pos(SP)}$ holds. If it does, we force the new payoff entry $\mathtt{v}/4$ of this strategy profile to equal the \emph{actual} game $G^0$. If there is an instance where $\mathtt{v} \neq \mathtt{u}$, then the logic program proves the predicate $\mathtt{changed}/0$. Otherwise, the predicate $\mathtt{unchanged}/0$ holds. This way the \texttt{Python} program can tell whether $mG \neq mG^\prime$, where $mG^\prime$ is a child of $mG$.

\paragraph{Computation of Nash equilibria}
\label{par:nash-equilibria}

To compute the Nash equilibria of a normal-form game we utilized the \textsc{Gambit} package. In particular, we used the \texttt{gambit-pol} command which implements the Support Enumeration method for computing the Nash equilibria in ${\vert N \vert}$-player normal-form games \cite{Algorithmic_Game_Theory_book}.

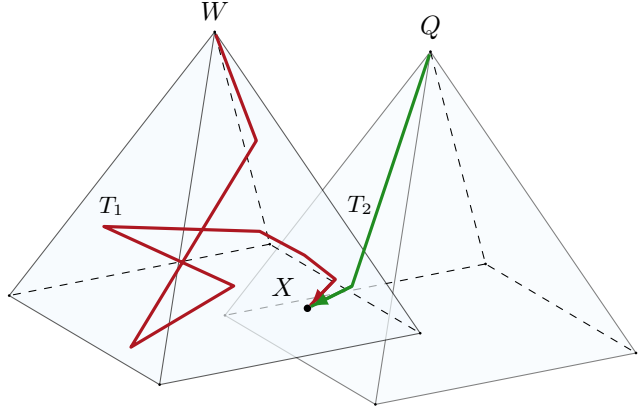
\begin{figure}[t]
    \centering
    \tdplotsetmaincoords{70}{150}
    \begin{tikzpicture}[tdplot_main_coords,line cap=butt,line join=round,c/.style={circle,fill,inner sep=0pt},
        declare function={a=4;h=4;},mycirc/.style={circle,fill=black, minimum size=0.1pt}]
        \path
        (0,0,0) coordinate (A)
        (a,0,0) coordinate (B)
        (a,a,0) coordinate (C)
        (0,a,0) coordinate (D)      
    (a/2,a/2,h)  coordinate (S);
\draw[fill=cyan!4,opacity=.5] (S) -- (D) -- (C) -- (B) -- cycle (S) -- (C);
\draw[dashed] (S) -- (A) --(D) (A) -- (B);
        \path foreach \p/\g in {A/-90,B/-90,C/-90,D/-90,S/90}{(\p)node[c]{}+(\g:2.5mm) node{}};
    \node[scale=0.3,label=above:{$Q$}] (int) at (S) {};

    \node[draw=none, fill=none,scale=0.3,label=below:{}] (sol1) at (.65,2.25,.4) {};

    \node[draw=none,fill=none] (s1) at (.3,2.35,.6) {};

   \node[draw=none,fill=none] (q1) at (1,2.85,1.4) {};

   \node[draw=none,fill=none] (w1) at (.4,2.75,.4) {};
   
      \path
        (2.1,-2.1,0) coordinate (A1)
        (a+2.1,-2.1,0) coordinate (B1)
        (a+2.1,a-2.1,0) coordinate (C1)
        (2.1,a-2.1,0) coordinate (D1)      
        (a/2+2.1,a/2-2.1,h)  coordinate (S1);
     \draw[fill=cyan!4,opacity=.65] (S1) -- (D1) -- (C1) -- (B1) -- cycle (S1) -- (C1);
    \draw[dashed] (S1) -- (A1) --(D1) (A1) -- (B1);
        \path foreach \p/\g in {A1/-90,B1/-90,C1/-90,D1/-90,S1/90}{(\p)node[c]{}+(\g:2.5mm) node{}};
    \node[scale=0.3,label=above:{$W$}] (int1) at (S1) {};    

    \node[] (inter) at (1,2.075,1.2) {};
    \node[mycirc,scale=0.3,label=above left:{$X$}] (sol) at (3.6,1.5,.5) {};
    
    \draw[upsdellred,very thick] (int1) -- (4.5,1.7,3.1) -- (6.8,2.35,.8) -- (5.8,3.35,1.8) -- (7.1,2.15,2.5) -- (5.4,3.35,2.5) -- (4.8,3.5,2.1) -- (4.5,3.8,1.8) -- (sol)  ;
    \draw[-latex,upsdellred,very thick] (4.5,3.8,1.8) to (sol);
    
    \node[draw=none,fill=none] (w1) at (6.75,1.75,2.6) {\small $T_1$};
    
    \draw[forestgreen(web),very thick] (int) -- (4.1,3.55,1.55) -- (sol)  ;
    \draw[-latex,forestgreen(web),very thick] (4.1,3.55,1.55) to (sol);
    
    \node[draw=none,fill=none] (w2) at (3.8,3.25,2.55) {\small $T_2$};
    
    \end{tikzpicture}
    \caption{Two threads $T_1, T_2$ simultaneously discovering the same node $X$. $T_1$ begins from $W$, while $T_2$ starts from $Q$.}
    \label{fig:parallelism}
\end{figure}

\paragraph{Parallelism}

In order to fully utilize the computational resources of a modern CPU, a parallel variant of our ALL-SMEs algorithm was developed. In details, Algorithm \ref{algo:algo-1} does not need any parallelization, whereas the parallelisation of Algorithm \ref{algo:algo-3} is straightforward, since Algorithm \ref{algo:algo-3} simply iterates through the terminal set $\mathcal{T}$. On the other hand, developing a parallel version of the traversal algorithm, i.e., Algorithm \ref{algo:algo-2}, involves some technicalities that need to be addressed (see Algorithm \ref{algo:algo-4}); these are described below.

Our main objective is to parallelise the expensive computations of Algorithm \ref{algo:algo-2}, i.e., the \emph{computation of the nmes}, which presupposes the computation of Nash equilibria, and the update operation. Recall that these computations are performed by the two subsystems, namely the package \textsc{Gambit} and a logic program in \textsc{Clingo}, respectively. In this direction, we assume $k$ threads. Each thread will execute a variation of Algorithm \ref{algo:algo-2}, thus \emph{independently} traversing the adaptation graph. Each thread will communicate with the others through the \emph{shared} data structures $\widetilde{\mathcal{Q}}, \widetilde{\mathcal{V}}, \widetilde{\mathcal{T}}$, $\theta[\cdot]$. Observe that for each new misinformation game $mG$ that appears in the adaptation graph we need to compute an $nme$, and compute the result of the update operation. Since each thread independently traverses the adaptation graph, it is possible to rediscover a node that was previously discovered by another thread (see Figure \ref{fig:parallelism}). This fact entails two possible inefficiencies of parallel traversal, namely redundant computations and long waiting times. In other words, if a new misinformation game is discovered (almost) simultaneously by two threads, i.e., before the one notifies the other about its discovery, it is possible that the same computations will be performed in both threads, thus nullifying the benefits of parallelization. On the other hand, if one thread waits the other to complete its computations before continuing its traversal, the benefit of parallelization is also nullified.

Our solution to the above problems is described in Algorithm \ref{algo:algo-4}. Each of the shared data structures $\widetilde{\mathcal{Q}}, \widetilde{\mathcal{V}}, \widetilde{\mathcal{T}}$, $\theta[\cdot]$ is protected by a lock. In order to gain access to the data structure the thread needs to successfully acquire the lock. Note that the independent traversal begins after line \ref{step:algo-4-begin-threads} and ends in line \ref{step:algo-4-join-threads}, where the threads are joined. We focus our attention in lines \ref{step:algo-4-lock-V}--\ref{step:algo-4-unlock-V}. Observe that the expensive computations we want to be performed in parallel are carried out in lines \ref{step:parallel-algo-4-semantic-loop}, \ref{step:parallel-algo-4-compute-nme}. Assume that two threads $T_1, T_2$ (almost) simultaneously discover $X$, with $T_1$ discovering it slightly earlier. Since $T_1$ arrives slightly earlier at $X$, it will achieve to acquire the lock for $\widetilde{\mathcal{V}}$. If $X \notin \widetilde{\mathcal{V}}$, $T_1$ will mark $X$ as visited (see line \ref{step:algo-4-mark-visited}) and release the lock. We claim that the expensive computations in lines \ref{step:parallel-algo-4-semantic-loop} and \ref{step:parallel-algo-4-compute-nme} will be executed only by $T_1$. Additionally, $T_2$ will only wait for the lines \ref{step:algo-4-lock-V}--\ref{step:algo-4-unlock-V}. These lines only involve a membership query, and (possibly) an insertion to the $\widetilde{\mathcal{V}}$ data structure. In order to see that, we follow the execution of $T_2$. Since $T_2$ acquires the lock for $\widetilde{\mathcal{V}}$ after $T_1$, the membership query in line \ref{step:algo-4-membership-V} will be positive. Thus, $T_1$ will omit all computations in lines \ref{step:algo-4-mark-visited} to \ref{step:algo-4-end-for}, and \emph{continue} to process the next position vector.

\begin{algorithm}[!htbp]
\caption{$\texttt{Parallel-TraverseAdaptationGraph}(mG, k)$}
\label{algo:algo-4}
\small{
\begin{algorithmic}[1]
        \Require A root misinformation game $mG$ and the number of threads $k$.
        \Statex
        \State $\widetilde{\mathcal{T}} \gets \emptyset$,
        $\widetilde{\mathcal{Q}} \gets \{ \emptyset \}$,
        $\widetilde{\mathcal{V}} \gets \{\emptyset\}$
        \State $\theta[\emptyset] \gets mG$
        \State compute $NME(mG)$
        \State \textbf{threads--begin}($k$) \label{step:algo-4-begin-threads}
        \While{$\widetilde{\mathcal{Q}} \neq \emptyset$}
            \State \textbf{lock}($\widetilde{\mathcal{Q}}$)
            \State Pick some $W \in \widetilde{\mathcal{Q}}$
            \State $\widetilde{\mathcal{Q}} \gets \widetilde{\mathcal{Q}} \setminus \{W\}$
            \State \textbf{unlock}($\widetilde{\mathcal{Q}}$)
            \If{$\exists \vec{v} \in \chi(NME(\theta[W])) \cap W$}
                \State \textbf{lock}($\widetilde{\mathcal{T}}$)
                \State $\widetilde{\mathcal{T}} \gets \widetilde{\mathcal{T}} \cup \{W\}$ \Comment{Self-loop detected}
                \State \textbf{unlock}($\widetilde{\mathcal{T}}$)
            \EndIf
            \ForEach{position vector $\vec{v} \in \chi(NME(\theta[W])) \setminus W$}
                \State $X \gets W \cup \{\vec{v}\}$
                \State \textbf{lock}($\widetilde{\mathcal{V}}$) \label{step:algo-4-lock-V}
                \If{$X \in \widetilde{\mathcal{V}}$} \Comment{Check if $X$ unvisited} \label{step:algo-4-membership-V}
                    \State \textbf{unlock}($\widetilde{\mathcal{V}}$)
                    \State \textbf{continue}
                \EndIf
                \State $\widetilde{\mathcal{V}} \gets \widetilde{\mathcal{V}} \cup \{X\}$
                \label{step:algo-4-mark-visited}
                \State \textbf{unlock}($\mathcal{V}$) \label{step:algo-4-unlock-V}
                \State $\#$ Compute the update operation resulting in $mG_X$.
                \If{$mG_{X} = mG_W$} \label{step:parallel-algo-4-semantic-loop}
                    \State \textbf{lock}($\theta$)
                    \State $\theta[X] \gets \theta[W]$
                    \Comment{New set-based representation for}
                    \Statex \Comment{same misinformation games}
                    \State \textbf{unlock}($\theta$)
                    \State \textbf{lock}($\widetilde{\mathcal{T}}$)
                    \State $\widetilde{\mathcal{T}} \gets \widetilde{\mathcal{T}} \cup \{mG_X\}$ \Comment{Self-loop detected}
                    \State \textbf{unlock}($\widetilde{\mathcal{T}}$)
                \Else
                    \State \textbf{lock}($\theta$)
                    \State $\theta[X] \gets mG_X$
                    \Comment{New misinformation game}
                    \State \textbf{unlock}($\theta$)
                    \State compute $NME(mG_X)$ \label{step:parallel-algo-4-compute-nme}
                    \State \textbf{lock}($\widetilde{\mathcal{Q}}$)
                    \State $\widetilde{\mathcal{Q}} \gets \widetilde{\mathcal{Q}} \cup \{X\}$
                    \State \textbf{unlock}($\widetilde{\mathcal{Q}}$)
                \EndIf
            \EndFor \label{step:algo-4-end-for}
        \EndWhile
        \State \textbf{threads--join}() \label{step:algo-4-join-threads}
        \State \Return $\widetilde{\mathcal{T}}, \theta[\cdot]$
\end{algorithmic}
}
\end{algorithm}

\subsection{Experiments}
\label{subsec:experiments}

In this section, we detail our experimental findings on the Adaptation Procedure in misinformation games, exploring its performance across a spectrum of multiplayer misinformation games. The methodology employed involved Monte Carlo simulations, following a structured approach. 
In particular, we first considered various configurations regarding the number of players and pure strategies per player; these are shown in the first column of Tables \ref{tbl:efficiency monte-carlo} and \ref{tbl:characteristics monte-carlo} (under ``Setting''). For instance, the notation $3 \times 2 \times 2$ denotes a misinformation game with $3$ players, where the first player has $3$ pure strategies and the other two have $2$ pure strategies, across all games in the misinformation game. Then, we constructed the payoff matrices of the misinformation game by using a random process to generate the payoff matrices of the actual game and of each subjective game. In particular, we populated the payoff matrices with integers drawn randomly from the interval $[-10, 10]$. 
We repeated this process $100$ times and executed the Adaptation Procedure for each generated misinformation game.

For each randomly generated misinformation game (and subsequent run of the Adaptation Procedure), we collected various data regarding the efficiency and characteristics of the process.
We present those results below in two separate tables, one concerning the time efficiency of the implementation (Table~\ref{tbl:efficiency monte-carlo}), and one concerning the characteristics of the adaptation graph (Table~\ref{tbl:characteristics monte-carlo}). 

All experiments were conducted on a macOS Ventura version 13.5.2 (22G91) machine, utilizing \textsc{Python} 3.11.4, \textsc{Gambit} 15.1.1, and \textsc{Clingo} 5.4.0. Notably, all experiments were executed using a multi-threaded implementation with eight hardware threads\footnote{The implementation of the Adaptation Procedure is in \url{https://github.com/merkouris148/adaptation-procedure-misinformation-games}.}.

\begin{table}[htp] 
\centering
 \caption[Efficiency of the Adaptation Procedure for various realizations.]{Efficiency of the Adaptation Procedure for various realizations.
 }\label{tbl:efficiency monte-carlo} 
    \begin{tblr}{
      colspec={ccc},
      row{even}={bg=gray!25},  
      row{1}={bg=white},
    }
      \hline
      Setting & Total time (sec) & CPU time (sec)   \\
      \hline\hline
      $2 \times 2$ & [0.11, 0.18] & [0.04, 0.1]  \\
      \hline
      $3 \times 2$ & [0.23, 0.29] & [0.17, 0.22]  \\
      \hline
      $3 \times 3$ & [0.17, 0.2] & [0.11, 0.15]  \\
      \hline
      $4 \times 3$ & [1.1, 1.29] & [1.06, 1.28]   \\
      \hline
      $4 \times 4$ & [1424.42, 1476.48] & [163.59, 188.12] \\
      \hline
      $2 \times 2 \times 2$ & [1.83, 2.19] & [1.66, 1.91] \\
      \hline
      $3 \times 2 \times 2$ & [22.76, 30.33] & [18.94, 22.45] \\
      \hline
      $2 \times 2 \times 2 \times 2$ & [3793.66, 4092.59] & [2172.43, 2386.21] \\
      \hline
    \end{tblr} 
\end{table}

For assessing the efficiency of the process, we used the metrics of Total run time (sec) and CPU run time (sec), as shown in Table~\ref{tbl:efficiency monte-carlo}. CPU time refers to the time consumed only by the \textsc{Python} program, whereas Total time refers to the total elapsed time, including the time consumed by the subprocesses, \textsc{Clingo} and \textsc{Gambit}. For both attributes we provide the interval between the minimum and maximum time as provided by the simulation, in the form $[min,max]$. 

We observe that when the misinformation game has two players with three pure strategies each ($3 \times 3$) the calls to \textsc{Python} constitute the main bottleneck of our pipeline, as shown by the fact that CPU time is close to Total time. On the other hand, in $2 \times 2$ and $4 \times 4$ settings the elapsed time is dominated by the subprocesses (i.e., \textsc{Clingo} and \textsc{Gambit}), and the CPU time consumes only $\sim \! \! 20\%$ and $\sim \! \! 12\%$ respectively of the elapsed time on average. An explanation here is that for the $2 \times 2$ case the possible misinformation games spawned from the Adaptation Procedure are few. On the contrary, in the $4 \times 4$ case, the adaptation graph may have more extensive branches, so \textsc{Clingo} and \textsc{Gambit} take their toll.

Moreover, as we increase the number of players ($2 \times 2 \times 2 \times 2$) we witness that both Total time and CPU time increase rapidly. To comprehend the reason behind this performance degradation, it's crucial to consider the exponential growth of the search space with the increase in available players and/or strategies.
Recall that the computational complexity for computing all Nash equilibria growths exponentially as proved in Theorem \ref{theo:complexity-corectness}.

\begin{table}[htp] 
\centering
 \caption[Characteristics of the adaptation graph for various realizations.]{Characteristics of the adaptation graph for various realizations.}\label{tbl:characteristics monte-carlo} 
    \begin{tblr}{
      colspec={cccccc},
      row{even}={bg=gray!25},
      row{1}={bg=white},
    }
      \hline
      Setting & $\#$nodes & unique $mG$s & $\#$leaves & $\# sme$s  & SP \\
      \hline\hline
      $2 \times 2$ & 13 & 7 & 6 &  3  & $2^4$\\
      \hline
      $3 \times 2$ & 50 & 17 & 18 &  3 & $2^6$ \\
      \hline
      $3 \times 3$ & 49 & 16 & 20 &  10 & $2^9$ \\
      \hline
      $4 \times 3$ & 451 & 132 & 205 &  9 & $2^{12}$ \\
      \hline
      $4 \times 4$ & 54122 & 10653 & 32779 & 377 & $2^{16}$ \\
      \hline
      $2 \times 2 \times 2$ & 832 & 148 & 489 & 373 & $2^8$ \\
      \hline
      $3 \times 2 \times 2$ & 8226 & 1287 & 4655 &  634 & $2^{12}$ \\
      \hline
      $2 \times 2 \times 2 \times 2$ & 632302 & 58617 & 24346 &  102338 & $2^{16}$ \\
      \hline
    \end{tblr} 
\end{table}

Regarding the characteristics of the Adaptation Procedure, we analyzed various aspects, as shown in Table \ref{tbl:characteristics monte-carlo}. Below, we explain the meaning of each of these quantities, and analyse the respective results.
Note that the number shown for each quantity in Table \ref{tbl:characteristics monte-carlo} is the average over all $100$ runs (rounded to the closest integer).

To study the characteristics of the adaptation graph generated by our algorithm, we also implemented the recursive method \eqref{eq:naive rec}. Using this implementation, we compared the number of nodes (i.e., misinformation games) generated by the implementation (as shown in column ``\#nodes'') against the number of nodes generated by the implementation of the set-based representation \eqref{eq:queue-rec} (shown in column ``unique $mG$s''). From this comparison it is obvious that the number of duplicate computations generated by \eqref{eq:naive rec} is huge. For example, in the $2 \times 2$ setting, the \eqref{eq:naive rec} computed $13$ misinformation games, while the \eqref{eq:queue-rec} one only $7$, meaning that $6$ (almost half) of the games computed by the \eqref{eq:naive rec} were actually duplicates, and thus caused redundant computations. Analogous observations can be made in the other settings (rows of the table).

Further, it is interesting to compare the numbers appearing in columns ``\#leaves'' and ``\#$sme$s'' in the table. 
The former (``\#leaves'') corresponds to the number of distinct sets of position vectors that lead to an $mG$ that belongs to the terminal set (Definition \ref{def:terminal-set}), or, in other words, it denotes the size of the terminal set (i.e., the number of distinct $mG$s that belong to the terminal set). 
The latter (``\#$sme$s'') corresponds to the actual number of $sme$s that this terminal set generates (according to Definition \ref{def:sme}).
Note that all $sme$s are generated by $mG$s in the terminal set. However, not every $mG$ in the terminal set is guaranteed to yield an $sme$; some may lead to none, while others may lead to one or more $sme$s. The fact that the numbers in the ``\#leaves'' column are significantly larger than the numbers in the ``\#$sme$s'' column, signifies that the majority of the $mG$s in the terminal set produced no $sme$. Upon reviewing Definitions \ref{def:sme} and \ref{def:terminal-set}, we conclude that there are a lot of $mG$s that ``self-replicate'' via adaptation (in the sense of Definition \ref{def:terminal-set} on the terminal set), yet they do not achieve ``stability'' (for their $nme$s to produce an $sme$, as described in Definition \ref{def:sme}).

Another interesting observation stems by comparing the numbers appearing in Table~\ref{tbl:characteristics monte-carlo} across different rows.
For example, one might anticipate that the settings $4 \times 3$ and $4 \times 4$ would yield comparable outcomes, but, in fact, the adaptation graph for $4\times 3$ generates only $451$ nodes ($132$ distinct), whereas $4 \times 4$ results in $54,122$ nodes ($10,653$ distinct). This counter-intuitive result is explained by looking at the last column of the table (``SP''), which denotes the size of the strategy space of the respective setting, i.e., the total number of sets of position vectors that the setting can generate; this in turn corresponds to the maximum possible number of distinct nodes (misinformation games) that the algorithm implementing the Adaptation Procedure could generate (i.e., the maximum possible value of column ``unique $mG$s''). 

Considering the strategy space (``SP'') column values, the $4 \times 3$ scenario indicates that the Adaptation Procedure could generate up to $2^{12}$ unique misinformed strategy profiles. In contrast, the $4 \times 4$ setting could yield $2^{16}$ distinct misinformed strategy profiles. It's noteworthy that the ``SP'' value by itself does not suffice to fully explain the size of the adaptation graph and the related quantities in all settings. It appears that incorporating more players contributes more to this aspect, compared to the addition of strategies, even when such additions (players or strategies) have a comparable effect on the ``SP'' value.

\section{Related work}\label{sec:related work}

Here we present the works that are close to the setting where agents with subjective views of the interaction interact in a turn-based procedure. As stated in \cite{VFBF} the misinformation games are close to the Hypergames (HG), the games with (un)awareness (GwU), and the Selten games (see \cite{Copic2006AwarenessAA,DBLP:journals/mss/Schipper14,DBLP:journals/anor/Sasaki17,Schipper2017SelfconfirmingGU,Feinberg2020,Harsanyi:1967:GII:3218759.3218761}). The first two of these classes are under the hood of games with misperception, often abbreviated as MP \cite{RaiffanLuce}, while the third is the a posteriori lottery model and defined in Harsanyi's seminal work \cite{Harsanyi:1967:GII:3218759.3218761}. However, since the setting in \cite{VFBF} focused only on one-shot interactions, the comparison with MPs (HGs and GwU) was also affected by this limitation.

However, MPs do not address only one-shot interactions but have also been proposed as a novel framework to analyze iterative interactions. Towards this direction, there is a significant stream of works that aims in the analysis of players' beliefs in MP (e.g., \cite{ArrowGreen,Esponda2016EquilibriumIM,Esponda2014BerkNashEA,Esponda2019AsymptoticBO,Spiegler2016,DBLP:conf/nips/RayKMD08,Heidhues2018ConvergenceIM,RePEc:ecm:emetrp:v:56:y:1988:i:5:p:1045-64}), where players update their beliefs using probabilistic formulas and techniques. Moreover, the agents have common knowledge about their views, as they update their knowledge in order to converge as close as possible to the outcome that they should have if they did not experience subjectivity. Contrary to these works, in our work the agents update their beliefs in a one-way manner, meaning that the agents update their subjective views plugging into the feedback that they receive from the environment.

In \cite{ArrowGreen}, a learning framework was plugged in with an MP, where the players are completely ignorant regarding opponents' decisions and any new information is automatically integrated. In \cite{lerer2019improving} the authors describe approaches for MP, where the objective views emerged from partially observable cooperative games, but they assume perfect knowledge of other players’ policies. Here, we drop the constraint of the limited observability and the cooperative setting.

Two other classes of games that are closely related to the misinformation games are the Selten games \cite{Harsanyi:1967:GII:3218759.3218761} and the Global games \cite{Globalgames}. In the first case, the agents learn only their own payoff matrix. So the agents have uncorrelated observations, but this is common knowledge. In the second case, the agents receive private and fuzzy observations of the game, where the latter is influenced by small random fluctuations. To cope with incompleteness, the agents take into account a whole family of a priori possible games; this is common knowledge. In other words, the agents know that they do not know, as opposed to the misinformation games where the agents do not have common knowledge and take into account only the agglomeration of their subjective views and the information that is broadcasted from the environment. 

Further, in \cite{JORDAN199160} a class of Bayesian processes for iterated normal form games is studied, where each player knows her own payoff function, but is uncertain about the opponents' payoffs. Another learning mechanism was introduced in  \cite{Esponda2014BerkNashEA} for games with misperception. In the same manner, in \cite{DBLP:conf/nips/RayKMD08}, a behavioral-based model is provided in order to model the ignorance of each player about adversaries. As opposed to the stream of works that rely on probabilistic techniques, in this work, we take no probabilistic considerations regarding the beliefs of the players. Here we focus on the effect of new information in the misinformed views of the players, i.e., the players ``do not know that they do not know'', as opposed to probabilistic approaches, where the players ``know that they do not know''. This leads to a conceptual and methodological difference between the Nash-Berk equilibrium concept, developed in \cite{Esponda2014BerkNashEA,Fudenberg2016ActiveLW}, and the stable misinformed equilibrium.

Now we steer our discussion to the solution concepts. In HGs, the outcome is not influenced by the actual/exact game. On the other hand, in GwUs the solution concept is influenced by the elimination of beliefs (rationalizability) and does not require correct conception and information, or mutual knowledge. Here, we both agglomerate the actual/exact game and we do not eliminate any belief towards the outcome. In MPs, the direction is towards the updates that restrict the subjectivity of the agents, while in this work we do not enforce such kind of condition.

Another equilibrium concept closely related to the solution concept in this work is that of self-confirming equilibrium \cite{10.2307/2951716,DEKEL1999165}, which captures the idea of ``stable'' joint decisions in MP. Although it shares the same intuition with the stable misinformed equilibrium, which is that players' choices stabilize when the choices of their opponents confirm their view, these two concepts have two differences. Firstly, the self-confirming equilibrium is applied in extensive-form games, while the stable misinformed equilibrium is in strategic-form games. Second, the former concept requires that the beliefs of the participants are correct along the path of plays, whereas we do not make such an assumption.

The limited behaviour of players while the values of their payoffs are non-stationary, has drawn significant attention. Authors in \cite{Romanyuk2017ActiveLW} provide a characterization of asymptotic behaviour, based on a continuous time model, where the player thinks that there are two possible payoff functions, yet none of them is true. Moreover, in \cite{Heidhues2018ConvergenceIM} authors establish convergence of beliefs, and actions, in a misspecified model with endogenous actions. 
Also, in \cite{Fudenberg2016ActiveLW} a complete characterization of the limit behavior of actions in cases with misspecified Bayesian players is provided. In our case, players do not have second thoughts regarding the received information.

In a weaker sense, a posteriori lotteries contain any interaction where agents depart from correct/complete information. In this respect, another concept where players have incomplete knowledge is that of Bayesian games \cite{Harsanyi:1967:GII:3218759.3218761}. It has attracted considerable attention (\cite{Szekely,Zamir2009,Gao2013,DBLP:journals/mst/GairingMT08}, etc.), where a key assumption is that of common priors. Nevertheless, Bayesian games are defined in \cite{Harsanyi:1967:GII:3218759.3218761} as a priori lotteries. 

Furthermore, a conceptual difference with misinformation games is that, in Bayesian games, although agents are unsure as to their actual payoff, they are well aware of that, and they do their best out of the uncertainty that they have, e.g. Figure \ref{fig:Bayesian games}. On the contrary, in misinformation games, the agents play according to their subjective game definition, without considering mitigation measures, e.g., Figure \ref{fig:mG graph}. 
In other words, in Bayesian games players ``know that they do not know'', whereas in misinformation games players ``do not know that they do not know''. Thus, the scenario depicted in Figure \ref{fig:mG graph} cannot be captured by Bayesian games, as the players cannot distinguish the actual situation from the one provided to them. 

\begin{table}[htp] 
\centering
 \caption[Basic families of methodologies regarding either one-round or multiple-round interactions with either incomplete or incorrect information.]{Basic families of methodologies regarding either one-round or multiple-round interactions with either incomplete or incorrect information.}\label{<table-label>} 
    \begin{tblr}{
      colspec={cccc},
      row{even}={bg=gray!25},  
      row{1}={bg=white},
      row{9}={bg=yellow!35},
    }
      \hline
      Class & Game setting & \makecell{iterative \\ process} & outcome \\
      \hline\hline
      \makecell{HG \\ \cite{Sasaki2012},\cite{Gharesifard11}} & \makecell{incorrect information, \\ erroneously aware, \\ they don’t know that \\ they don’t know, \\ no exact/actual game} & \textcolor{forestgreen(web)}{\cmark} & \makecell{ hyper Nash \\ equilibrium} \\
      \makecell{GwU \\ \cite{DBLP:journals/mss/Schipper14}} & \makecell{lack of conception, \\ not lack of information} & \textcolor{red}{\xmark} & \makecell{notions of \\ rationalizability} \\
      \makecell{MP \\ \cite{Esponda2014BerkNashEA}} & \makecell{players know that they \\ don't know the game setting} & \textcolor{forestgreen(web)}{\cmark} & \makecell{variants \\ of \\ equilibria} \\
      \makecell{Selten \\ games \\ \cite{Harsanyi:1967:GII:3218759.3218761}} & \makecell{players know the game setting, \\ know their own attribute vector \\ but don't know the exact/actual game, \\ know that they don’t know} & \textcolor{red}{\xmark} & \makecell{equilibrium \\ selection} \\
      \makecell{Global \\ games \\ \cite{Globalgames}} & \makecell{players know the game setting, \\ don't know the exact/actual game, \\ know that they don’t know}  & \textcolor{forestgreen(web)}{\cmark} & \makecell{equilibrium \\ selection} \\
      \makecell{Bayesian \\ games}  & \makecell{players know the game setting, \\ don’t know the exact/actual game, \\  know that they don’t know } & \textcolor{forestgreen(web)}{\cmark} & \makecell{Bayes-Nash \\ equilibrium} \\
      \makecell{Misinformation \\ games \cite{VFBF},\\ \small{\cite{VFFB}}} & \makecell{incorrect information, \\ erroneously aware, \\ they don’t know that \\ they don’t know} & \textcolor{red}{\xmark} & \makecell{natural \\ misinformed \\ equilibrium} \\
      \hline 
    \makecell{Misinformation \\ games \textbf{plus} \\ Adaptation \\ Procedure \cite{PVF}} & \makecell{incorrect information, \\ erroneously aware, \\ they don’t know that \\ they don’t know \\ \textbf{the information they learn} \\ \textbf{becomes commonly known}} & \textcolor{forestgreen(web)}{\cmark} & \makecell{stable \\ misinformed \\ equilibrium} \\
      \hline
    \end{tblr} 
\end{table}

\section{Conclusions}\label{sec:conclusions}

This work explores the pervasive nature of misinformation in multi-player interactions. By employing the framework of \emph{misinformation games}, we gain insights into various real-world dynamics in scenarios where the players' information on the interaction are erroneous, but the players have no clue that their specifications are false.

In this setting, we define an iterative process, the \emph{Adaptation Procedure}, which enables players to refine their strategies through continuous interaction and feedback. This process is rooted in the idea that players will adjust their strategies based on the payoffs that are publicly announced in each iteration, which may be different from what they know, thereby progressively reducing misinformation.

The Adaptation Procedure models the evolution of players' decisions as they assimilate new information, culminating in a stable misinformed equilibrium, which is the strategy profile adopted upon process stabilization. We establish that this procedure invariably concludes, guaranteeing the existence of such an equilibrium and study (and prove) various related properties. 
 
Lastly, we devise algorithmic tools to analyze this evolution, revealing the \textbf{PPAD}-complete complexity of the stable misinformed equilibrium ($sme$), and provide some experiments demonstrating the characteristics and computational properties of the Adaptation Procedure under various randomly-generated settings.

%%%%%%%%%%%%

\bibliographystyle{plain}
\bibliography{main.bib}

%%%%%%%%%%%%

\clearpage

\appendix

\section{Proofs of results}\label{appendix:omitted proofs}

\subsection{Normal-form misinformation games}

\begin{proof}[Proof of Proposition~{\upshape\ref{prop:inflate-cons-hi}}]
By equation \eqref{eq:payoff_function}, we have that:
\[
h_{i}(\sigma) = \sum\limits_{k \in S_1}\dots \sum\limits_{j\in S_{\vert N \vert}}P_i(k, \dots, j) \cdot \sigma_{1, k} \cdot \ldots \cdot \sigma_{\vert N \vert,j},
\]
Replacing the tuple $(k, \dots, j)$ with $s = (s_i) \in S$, we can rewrite this as follows
\[
h_{i}(\sigma) = \sum\limits_{s \in S}
P_i(s) \cdot \sigma_{1, s_1} \cdot \ldots \cdot \sigma_{\vert N \vert,s_{\vert N \vert}}
\]
Working analogously for $h_i'$, we get:
\[
h_{i}'(\sigma') = \sum\limits_{s' \in S'}
P_i'(s') \cdot \sigma_{1, s_1'}' \cdot \ldots \cdot \sigma_{\vert N \vert,s_{\vert N'\vert}'}'
\]
Now, let us write $s'$ as $(\hat s, \tilde s)$, where $\hat s$ represents the strategies in $s'$ for the players in $N$ and $\tilde s$ the strategies in $s'$ for the players in $N' \setminus N$. Let us also denote with $\tilde S$ the different possible values of the tuple $\tilde s$ (note that $\tilde S = \times_{i\in N' \setminus N} S_i'$).
Under this notation, we have that, for any $s'$, $P_i'(s') = P_i(\hat s)$, from the third bullet of Definition \ref{def:inflated-nfgame}, and thus, rearranging the sums and using associativity, the latter equation on $h_i'$ can be written as:
\[h_{i}'(\sigma') = \sum\limits_{\hat s \in S} P_i(\hat s) 
\cdot 
\left(
\sigma_{1, \hat s_1}' \cdot \ldots \cdot \sigma_{\vert N \vert,\hat s_{\vert N\vert}}' \right)
\cdot 
\left(
\sum\limits_{\tilde s \in \tilde S}
\sigma_{\vert N \vert + 1, \tilde s_{\vert N \vert + 1}}' \cdot \ldots \cdot \sigma_{\vert N' \vert,\tilde s_{\vert N'\vert}}' \right)
\]
The second big parenthesis sums to $1$. Indeed, using associativity:
\[
\sum\limits_{\tilde s \in \tilde S}
\sigma_{\vert N \vert + 1, \tilde s_{\vert N \vert + 1}}' \cdot \ldots \cdot \sigma_{\vert N' \vert,\tilde s_{\vert N'\vert}}' =
\prod_{i \in N'\setminus N}
\sum_{j \in S_i'}
\sigma_{ij}' = 
1 \cdot \ldots \cdot 1 = 1
\]
Since $\sigma \comp \sigma'$, it follows that $\sigma_{ij} = \sigma_{ij}'$ for all $i \in N$, $j \in S_i$, and thus, combining the above results, we get:
\[
h_{i}'(\sigma') = \sum\limits_{\hat s \in S} P_i(\hat s) 
\cdot 
\sigma_{1, \hat s_1} \cdot \ldots \cdot \sigma_{\vert N \vert,\hat s_{\vert N\vert}}
\]
From the latter equation, and the respective equation for $h_i$, it is clear that $h_i(\sigma) = h_i'(\sigma')$, which completes the proof.
\end{proof}

\begin{proof}[Proof of Proposition~{\upshape\ref{prop:infl-same-size-same-game}}]
Apparently, if $G=G'$ then $G \infl G'$ (by the definition), so let us focus on the reverse. Assume that $G \infl G'$. Then, since $N = N'$ and $S= S'$ by the hypothesis, take any $i \in N$ and $s=(s_i)_{j\in N} \in S$, $s'=(s_i')_{j\in N'} \in S'$ such that $s_j = s_j'$ for all $j\in N$. Then $s = s'$ (by the hypothesis that $N=N'$ and $S= S'$, and $P_i(s) = P_i'(s')$ (by the hypothesis that $G \infl G'$), and thus $P = P'$, i.e., $G = G'$.   
\end{proof}

\begin{proof}[Proof of Proposition~{\upshape\ref{prop:infl-all-NE}}]
By the fourth bullet of Definition \ref{def:inflated-nfgame}, there exists some strategy profile in $G'$, say $\hat \sigma$, such that $\hat \sigma \in NE(G')$ and $\sigma \comp \hat \sigma$. 
Now consider any other strategy profile $\tilde \sigma$ in $G'$, such that $\sigma \comp \tilde \sigma$. It suffices to show that $\tilde \sigma \in NE(G')$.
We observe, by Proposition \ref{prop:inflate-cons-hi}, that $h_i'(\hat \sigma) = h_i(\sigma) = h_i'(\tilde \sigma)$. Since $h_i'(\hat \sigma) = h_i'(\tilde \sigma)$, the result follows easily from the definition of Nash equilibrium.
\end{proof}

\begin{proof}[Proof of Proposition {\upshape\ref{prop:infl-partial order}}]
Any game $G$ is trivially an inflated version of itself, so reflexivity holds.\\
For antisymmetry, we observe that $G_1 \infl G_2$ and $G_2 \infl G_1$ implies that $G_1$ and $G_2$ have the same set of players and the same set of strategies per player. Antisymmetry now follows by applying Proposition \ref{prop:infl-same-size-same-game}.\\
For transitivity, let us denote, for $i=1,2,3$, $N^i, S^i, P^i$ the set of players, the set of strategies and the payoff matrix, respectively, of $G_i$. Let us also denote by $S_j^i$ the available strategies for player $j$ in game $G_i$.
Suppose that $G_1 \infl G_2$ and $G_2 \infl G_3$. We will show that $G_1 \infl G_3$.\\
Indeed, it is clear that $N^1 \subseteq N^2 \subseteq N^3$, so the first bullet of the definition holds.\\
Similarly, for any $i \in N^1$, $S_i^1 \subseteq S_i^2$. Also $i \in N^1$ implies $i \in N^2$ and thus $S_i^2 \subseteq S_i^3$. Thus $S_i^1 \subseteq S_i^3$ for all $i \in N$, so the second bullet holds as well.\\
For the third, pick any $i \in N^1$, and $s^1=(s_i^1)_{j\in N^1} \in S^1$, $s^3=(s_i^3)_{j\in N^3} \in S^3$ such that $s_j^1 = s_j^3$ for all $j\in N^1$.
Let us pick some tuple $s^2=(s_i^2)_{j\in N^2} \in S^2$ such that $s_j^2 = s_j^3$ for all $j\in N^2$. Now, by the fact that $G_1 \infl G_2$ and $G_2 \infl G_3$ and the previous results, we get that $P_i^1(s^1) = P_i^2(s^2)$ and $P_i^2(s^2) = P_i^3(s^3)$, and thus the third bullet holds as well.\\
For the fourth, suppose that $\sigma^1 \in NE(G_1)$. Then, since $G_1 \infl G_2$, there exists some $\sigma^2 \in NE(G_2)$ such that $\sigma^1 \comp \sigma^2$.
Moreover, since $G_2 \infl G_3$, there exists some $\sigma^3 \in NE(G_3)$ such that $\sigma^2 \comp \sigma^3$. Given that $N^1 \subseteq N^2 \subseteq N^3$ and $S^1 \subseteq S^2 \subseteq S^3$, it is easy to see that $\sigma^1 \comp \sigma^2$ and $\sigma^2 \comp \sigma^3$ imply that $\sigma^1 \comp \sigma^3$. Thus, there exists some $\sigma^3 \in NE(G_3)$ such that $\sigma^1 \comp \sigma^3$, which shows the fourth bullet.\\
For the fifth bullet, we proceed analogously.
\end{proof}

\begin{proof}[Proof of Proposition {\upshape\ref{prop:infl-algo-add-player}}]
Observe that since the payoffs of the new player are zero (line 3), then this player will be indifferent in any choice of the other players. Thus does not affect the outcome of the $G'$. So, we have $G \infl G'$.
\end{proof}

\begin{proof}[Proof of Proposition {\upshape\ref{prop:infl-algo-add-strategy}}]
In Algorithm~\ref{algo:addstrategy_inflate_game} we increase the pure strategies of player $i$ by adding the $j$ pure strategy. From lines 6-7 of the algorithm the new entries of the payoff matrices will have values smaller than any other value in the payoff matrix without the $j$ strategy. Thus, this strategy is dominated and does not affect the strategic behavior of the players. So, we have $G \infl G'$.
\end{proof}

\begin{proof}[Proof of Proposition {\upshape\ref{prop:infl-algo-full}}]
We observe that Algorithm \ref{algo:inflate_game} consists of repetitive calls to Algorithms \ref{algo:addplayer_inflate_game} and \ref{algo:addstrategy_inflate_game}.
The result now follows from Propositions \ref{prop:infl-algo-add-player}, \ref{prop:infl-algo-add-strategy}, and the transitivity of the $\infl$ relation (Proposition \ref{prop:infl-partial order}).
\end{proof}

\subsection{Adaptation Procedure}

\begin{proof}[Proof of Proposition~{\upshape\ref{prop:additive}}]
By definition:
\begin{equation*}
\begin{split}
    \AD{}{M} & = \{ mG_{\vec{u}} \mid mG \in M, \vec{u} \in \chi(\sigma), \sigma \text{ is a } nme \text{ of } mG \} \\
    & = \bigcup_{mG \in M} \{ mG_{\vec{u}} \mid \vec{u} \in \chi(\sigma), \sigma \text{ is a } nme \text{ of } mG \}   = \bigcup_{mG\in M} \AD{}{\{mG\}}
\end{split}
\end{equation*}
\end{proof}

\begin{proof}[Proof of Proposition~{\upshape\ref{prop:mG eq mG'}}]
Observe that $mG_i \in \AD{}{\AD{}{\dots \AD{}{\{mG_i\} }}} = \AD{(n)}{\{mG_i\}}$ $\forall i \in [n]$. 
Suppose, for the sake of contradiction, that $mG_i \neq mG_j$ for some $i,j$, and assume, without loss of generality, that $i < j$. Then, it holds that $mG_j \in \AD{(j-i)}{\{ mG_i\} }$, i.e., $mG_j$ has resulted from $mG_i$ by updating some (at least $1$ and at most $j-i$) elements of the respective payoff matrices of $mG_i$. But then, we also have that $mG_i \in \AD{(n-j+i)}{\{ mG_i \}}$ (by the periodic pattern above), so again, $mG_i$ has resulted from $mG_j$ by updating some (at least $1$ and at most $n-j+i$) elements of the respective payoff matrices of $mG_j$. But this is an absurdity, because replacements are cumulative and cannot be ``undone'' by subsequent ones (see Definitions \ref{def:vec_def} and \ref{def:AD(M)-adaptation procedure iterative process}).
\end{proof}

\begin{proof}[Proof of Proposition~{\upshape\ref{prop:mG-in-terminal-set}}]
    Let the misinformation game $mG^{\prime}$, for which it holds $mG^\prime \in \AD{(1)}{\{mG^\prime\}}$. It is easy to see, through induction, that $mG^\prime \in \AD{(\tau)}{\{mG^\prime\}}$ for all $\tau \geq 1$. Since, $mG^{\prime} \in \AD{(t)}{\{mG^{(0)}\}}$, then $mG^\prime \in \AD{\infty}{\{\mGt{0}\}}$.
\end{proof}

\begin{proof}[Proof of Proposition~{\upshape\ref{prop:algo-terminal-set-sme}}]
The first bullet is a consequence of the Definition~\ref{def:endingcriterion}. For the second bullet we have that, since $\sigma \in SME(\mGt{0})$ applying the Definition~\ref{def:sme} there must be a misinformation game $mG$ in the terminal set, such that $\sigma$ is a $nme$ in $mG$. 
\end{proof}

\begin{proof}[Proof of Proposition~{\upshape\ref{thm:bounded mG}}]
    Observe that in any given branch of the Adaptation Procedure, a certain position can be updated at most once. Given that the number of positions is $\vert S \vert$, we conclude.
\end{proof}

\begin{proof}[Proof of Proposition~{\upshape\ref{prop:finite-mgs}}]
    We use two ways to upper bound the number of misinformation games $\vert \AD{*}{\{mG\}} \vert$. Firstly, note that $\vert \AD{*}{\{mG\}} \vert \leq 2^{\vert S \vert}$, since there cannot be generated more misinformation games, throughout the adaptation procedure, than the powerset of position vectors $S$.
    
    For the other lower bound, notice that, in the worst case the Adaptation Procedure will produce a tree with $\LAD{mG}$ height. Each node in this tree will have at most $\vert S \vert$ children. We have,
    \begin{align*}
        \vert \AD{*}{\{mG\}} \vert  & \leq \sum_{i=1}^{\LAD{mG}} \vert S \vert^i \\
        & = \frac{\vert S \vert\left( \vert S \vert^{\LAD{mG}} - 1\right)}{\vert S \vert - 1} \\
        &\leq \vert S \vert^{\LAD{mG} + 1}
    \end{align*}
    Note that in the last inequality we assumed $\vert S \vert \geq 1$, without loss of generality. Thus, we showed the desideratum.
\end{proof}

\begin{proof}[Proof of Corollary~{\upshape\ref{cor:ad-infty-bound}}]
    Since $\AD{\infty}{mG} \subseteq \AD{*}{mG}$, the desideratum holds from Proposition \ref{prop:finite-mgs}.
\end{proof}

\begin{proof}[Proof of Theorem~{\upshape\ref{theo:stable-set-characterisation}}]
Let $\mathcal{S} = \AD{\infty}{mG}$ to be the Stable Set of the Adaptation Procedure on $mG^0$. It holds $\AD{}{\mathcal{S}} = \mathcal{S}$. We prove that $\AD{*}{\mathcal{T}} \subseteq \mathcal{S}$ and $\mathcal{S} \subseteq \AD{*}{\mathcal{T}}$. The direction $\AD{*}{\mathcal{T}} \subseteq \mathcal{S}$ holds trivially. Indeed, from Proposition \ref{prop:algo-terminal-set-sme} we have $\mathcal{T} \subseteq \mathcal{S}$. On the other hand, since $\mathcal{S}$ is the Stable Set, we have $\AD{}{\mathcal{S}} = \mathcal{S}$. From Proposition \ref{prop:additive}, we take $\AD{}{\mathcal{S}} = \AD{}{\mathcal{S} \setminus \mathcal{T}} \cup \AD{}{\mathcal{T}}$. Thus, $\AD{*}{\mathcal{T}} \subseteq \AD{}{\mathcal{S}}$.

For the direction $\mathcal{S} \subseteq \AD{*}{\mathcal{T}}$, we work as follows. For the sake of contradiction, assume $mG \in \mathcal{S} \setminus \AD{*}{\mathcal{T}}$. We consider the set $C$ of all ancestors of $mG$ that belong to $\mathcal{S}$.
\begin{equation}
    \label{eq:ancestor-set}
    C = \left\{mG^\prime \in \mathcal{S} \mid mG \in \AD{k}{mG^\prime}, k \in \mathbb{N} \cup \{0\}\right\}.
\end{equation}
                
Note that $C \neq \emptyset$, since $mG \in \AD{0}{mG}$, hence $mG \in C$. We choose some $mG^{\prime\prime} \in C$, for which there is not an $mG^\prime \in C$, such that $mG^{\prime\prime} \in \AD{}{mG^\prime}$. Observe that $C \cap \mathcal{T} = \emptyset$, since $mG \notin \AD{*}{\mathcal{T}}$. Therefore, such an $mG^{\prime\prime}$ exists. Intuitively, $mG^{\prime\prime}$ is the oldest ancestor%
\footnote{%
    Formally, the $\AD{}{\cdot}$ operator defines a transitive, anti-symmetric (but not reflexive) relation $\preceq_{\mathcal{AD}}$ on the sets of all misinformation games $\AD{*}{mG^{(0)}}$ generated from $mG^{(0)}$. From Proposition \ref{prop:finite-mgs} we have that $\AD{*}{mG^{(0)}}$ is finite. Thus, there are maximal elements in $\preceq_{\mathcal{AD}}$. Moreover, we have maximal elements in $C$, since $C \subseteq \AD{*}{mG^{(0)}}$. The \emph{"oldest ancestors"}, we discuss here, are these maximal elements of $\preceq_{\mathcal{AD}}$, restricted on $C$.
}
of $mG$, that belongs to the Stable Set. This oldest ancestor exists, since non of the elements of $C$ generates itself. Moreover, $mG^{\prime\prime}$ has no ancestor in $\AD{*}{mG^{(0)}}$. Otherwise, it would also be an ancestor of $mG$, thus belonging to $C$, and $mG^{\prime\prime}$ would not be the oldest ancestor. From the above, there is no element in $\mathcal{S}$ to generate $mG^{\prime\prime}$. Hence, $mG^{\prime\prime} \in \mathcal{S} \setminus \AD{}{\mathcal{S}}$, which is a contradiction.
\end{proof}

\begin{proof}[Proof of Proposition~{\upshape\ref{prop:mG leq mG'}}]
Suppose that $mG' \in \AD{(t_0)}{\{mG\}}$ for some $t_0 \geq 0$. 
Since $mG' \in \AD{}{\{ mG' \}}$, it follows that $mG' \in \AD{(t)}{\{mG\}}$ for all $t \geq t_0$, thus, $mG' \in \AD{\infty}{\{mG\}}$. 
Since $\AD{}{\{ mG' \}} = \{mG'\}$, it is clear that for all 
$\sigma \in NME(mG')$ and for all $\vec{v} \in \chi(\sigma)$, it holds that $mG'_{\vec{v}} = mG'$. Now the result is direct from Definition \ref{def:sme}.
\end{proof}

\begin{proof}[Proof of Proposition~{\upshape\ref{prop:existence}}]
Set $S = \AD{\infty}{\{mG\}}$.
For any given $mG_1, mG_2 \in S$, we define the relation $\mapsto$, such that $mG_1 \mapsto mG_2$ iff $mG_1 \neq mG_2$ and $mG_2 \in \AD{}{\{ mG_1 \}}$.
Now let us suppose, for the sake of contradiction, that $mG$ has no $sme$.
By Proposition \ref{prop:mG leq mG'}, it follows that for any $mG' \in S$ there exists some $mG'' \in S$ such that $mG' \mapsto mG''$ (otherwise $SME(mG) \neq \emptyset$ by Proposition \ref{prop:mG leq mG'}, which contradicts our hypothesis).
Since $S$ is finite (see Theorem \ref{thm:bounded mG}), there must exist a sequence of $mG_1, \dots, mG_n \in S$, such that $mG_i \mapsto mG_{i+1}$ (for $i \in [n-1]$) and $mG_n \mapsto mG_1$. Which is an absurdity by the definition of $\mapsto$ and Proposition \ref{prop:mG eq mG'}.
\end{proof}

\begin{proof}[Proof of Proposition~{\upshape\ref{theo:sme-complexity}}]
Assume a $N$-player misinformation game $mG = \langle G^0, G^1, \dots, G^N \rangle$, where $S^1 = S^2 = \cdots = S^N$. 
Further, let $P^0$ be the payoff matrix for the actual game $G^0$, and $P$ the payoff matrices for the $G^i$, for all $i$. 
For $P$ we use the dovetail principle to assign distinct integers, to each cell, form $\sum_{k=1}^{\vert N \vert} \vert S^k \vert$ to $1$, starting from $P[\vert S^1 \vert, \vert S^2 \vert, \dots, \vert S^{\vert N \vert}\vert]$. 
On the other hand, we take $P^0 = -P$.
It is easy to see that initially each player $i$ picks the $\vert S^i \vert$-th pure strategy, so the strategy profile $(\vert S^1 \vert, \vert S^2 \vert, \dots, \vert S^{\vert N \vert} \vert)$ is a pure Nash equilibrium in dominated strategies, for each player, constituting an $nme$. After a single step of the Adaptation Procedure, the players will be informed that $P^0[i_1, i_2, \dots, i_n] = -P[i_1, i_2, \dots, i_n]$. Using the same rationale as before, the next $nme$ chosen by the players $(\vert S^1 \vert -1, \vert S^2 \vert -1 , \dots, \vert S^{\vert N \vert}\vert  -1)$, which again will be updated to opposite value. Inductively, we can prove that the Adaptation Procedure will take $\vert S \vert$ steps.
\end{proof}

\subsection{Computing the Adaptation Procedure}

\begin{proof}[Proof of Proposition~{\upshape\ref{prop:adapt-graph-dag}}]
    It suffices to show that $\Gamma^\prime$ does not contain directed circles. This is an immediate consequence of Proposition \ref{prop:mG eq mG'}.
\end{proof}

\begin{proof}[Proof of Proposition~{\upshape\ref{prop:adapt-graph-source}}]
    We first show that there can be no other source $mG^\prime \in \AD{*}{mG} \setminus \{ mG \}$. For the sake of contradiction, let the misinformation game $mG^\prime \in \AD{*}{mG}$, with $mG^\prime \neq mG$, and $d^{-}_{\Gamma^\prime}(mG^\prime) = 0$. Let $k \in \mathbb{N}$ be the smallest integer such that $mG^\prime \in \AD{k}{mG}$. Therefore, $mG^\prime \notin \AD{k-1}{mG}$. Nevertheless, from Definition \ref{def:AD(M)-adaptation procedure iterative process}, $\AD{k}{mG} = \AD{}{\AD{k-1}{mG}}$. Thus, there is some $mG^{\prime\prime} \in \AD{k-1}{mG}$, with $mG^\prime \in \AD{}{mG^{\prime\prime}}$. We reached a contradiction, since we assumed that $mG^\prime$ is a source. 
    
    It remains to show that $d^{-}_{\Gamma^\prime}(mG) = 0$. Suppose, for the sake of contradiction, that $d^{-}_{\Gamma^\prime}(mG) > 0$, so there is some $mG^\prime \in \AD{*}{mG}$, $mG^\prime \neq mG$ such that $mG \in \AD{}{mG^\prime}$. Repeating the above argument, we conclude that there is some sequence of position vectors $X$ such that $mG = mG_X$, which is impossible by Proposition \ref{prop:mG eq mG'}.
\end{proof}

\begin{proof}[Proof of Proposition~{\upshape\ref{prop:adapt-graph-sinks}}]
    Assume some sink node $mG^\prime$, i.e., $mG^\prime \in K$. Since $mG^\prime$ is a sink in $\Gamma^\prime$, any outgoing edges from $mG^\prime$ in $\Gamma$ are self loops. This shows that $\AD{}{mG^\prime} = \{mG^\prime\}$, and, thus, $mG^\prime \in \mathcal T$, which proves the first claim of the proposition.

    For the second claim, we first observe that, as shown above, $\AD{}{mG^\prime} = \{mG^\prime\}$ for all $mG^\prime \in K$, and thus $K \subseteq \AD{\infty}{mG}$. Now take some $\sigma \in NME(mG^\prime)$ for some $mG^\prime \in K \subseteq \AD{\infty}{mG}$. 
    Given that $\AD{}{mG^\prime} = \{mG^\prime\}$, it is easy to confirm that all the conditions of Definition \ref{def:sme} hold, and thus $\sigma$ is an $sme$.
\end{proof}

\begin{proof}[Proof of Proposition~{\upshape\ref{prop:adapt-graph-longest-path}}]
    An immediate consequence of Proposition \ref{thm:bounded mG}.
\end{proof}

\begin{proof}[Proof of Proposition~{\upshape\ref{prop:semantic-equivalence-complexity}}]
Since $mG^1, mG^2$ are canonical misinformation games on the same number of players $\vert N \vert$ and the same number of strategies $\vert S \vert$, they can be represented by $\vert N \vert + 1$ matrices each with %$\vert S \vert^{\vert N \vert}$ 
$\vert S \vert$ cells. Each cell of the matrices will contain a $\vert N \vert$-dimensional vector containing the payoffs for each player.
To show equality, we have to compare each component of each of those vectors from these misinformation games, which leads us to the result.
\end{proof}

\begin{proof}[Proof of Proposition~{\upshape\ref{prop:child-syntactic-equivalence-complexity}}]
Since $\theta[Y] = \theta[X]_{\vec{v}}$, in order to determine if $Y = X$, it suffices to show that $\vec{v} \in X$. Since, we keep the position vectors sets as ordered list, the membership query $\vec{v} \in X$ can be done in $O(\log \vert S \vert)$ time.
\end{proof}

\begin{proof}[Proof of Lemma~{\upshape\ref{lem:correctness}}]
It suffices to show that for each $mG \in \mathcal{T}$ there exists some $i \in \mathbb{N}$ such that $mG \in \mathcal{Q}_i$. For the sake of contradiction assume that there is no such $\mathcal{Q}_i$ for $mG$. Let $mG^\prime$ be the most recent ancestor of $mG$ that belongs to some $\mathcal{Q}_j$, $j \geq 0$. There is such $mG^\prime$, since $mG^0$ is an ancestor of $mG$ and $mG^0 \in \mathcal{Q}_0$. Observe that all the \emph{proper} descendants $\AD{}{mG^\prime} \setminus \{mG^\prime\}$ of $mG^\prime$ belong to $\mathcal{Q}_{j+1}$. This is a contradiction from the choice of $mG^\prime$; thus, the result follows.
\end{proof}

\begin{proof}[Proof of Proposition~{\upshape\ref{prop:algo-3-correctness}}]
From equation \eqref{eq:queue-rec}, if $\mathcal{Q}_i = \emptyset$, then $\mathcal{Q}_{i+1} = \emptyset$, and $\mathcal{T}_{i+1} = \emptyset$. Thus, $\cup_{k=1}^{i+1} \mathcal{T}_k = \cup_{k=1}^{\infty} \mathcal{T}_k$. From Lemma \ref{lem:correctness}, $\cup_{k=1}^{i+1} \mathcal{T}_k = \mathcal{T}$. Then, from equation \eqref{eq:stable-rec}, we have $M_{i+1} = \AD{i+1}{\cup_{k=1}^{i+1} \mathcal{T}_k}$, since $\mathcal{Q}_{i+1} = \emptyset$. Moreover, it holds $\AD{i+1}{\cup_{k=1}^{i+1} \mathcal{T}_k} = \AD{*}{\cup_{k=1}^{i+1} \mathcal{T}_k}$. Observe that after the $(i+1)$-th step, no additional misinformation games will be generated, otherwise, $Q_{i+1} \neq \emptyset$. Lastly, $\AD{*}{\cup_{k=1}^{i+1} \mathcal{T}_k} = \AD{*}{\mathcal{T}}$, from the above. Therefore $M_{i+1} = \AD{*}{\mathcal{T}}$. From Theorem \ref{theo:stable-set-characterisation}, $M_{i+1}$ is a Stable Set, and the desideratum holds.
\end{proof}

\begin{proof}[Proof of Proposition~{\upshape\ref{prop:graph-traversal-complexity}}]
As we noted in the previous subsection, for the set-based representation approximation $\widetilde{\Gamma} = (2^S, \widetilde{E})$ of the adaptation graph. For each of these nodes, we compute its $nme$s in $ \vert N \vert \mathtt{t}(\textbf{NASH})$ time. These $nme$s may result in at most $\vert S \vert$ position vectors. Lastly, in the worst case, for each of these position vectors, we will perform an equality check of $O(\vert N \vert^2 \vert S \vert)$ time steps (see Proposition \ref{prop:semantic-equivalence-complexity}).
\end{proof}

\begin{proof}[Proof of Theorem~{\upshape\ref{theo:complexity-corectness}}]
The complexity of Algorithm \ref{algo:algo-2} is established in Proposition \ref{prop:graph-traversal-complexity}. We argue about the complexity of Algorithm \ref{algo:algo-3}. We note that for each misinformation game $mG_X$, such that $X$ belongs to the set-based representation $\widetilde{\mathcal{T}}$ of the terminal set, and for each position vector resulting from an $nme$, we perform an equivalence check. The number of position vectors is at most $\vert S \vert$.
\end{proof}

\begin{proof}[Proof of Corollary~{\upshape\ref{cor:maximal-path}}]
    Since $mG^\prime$ is the last node on a maximal path from $mG$ in $\Gamma^\prime$, then it is a sink in $\Gamma^\prime$. From Proposition \ref{prop:adapt-graph-sinks} we obtain the desideratum.
\end{proof}

\begin{proof}[Proof of Theorem~{\upshape\ref{theo:complexity-1-sme-upper-bound}}]
Assume a misinformation game $mG$. The \textbf{while} loop  of steps \ref{algo-Traverse-Adaptation-Graph-start-while}--\ref{algo-Traverse-Adaptation-Graph-end-while} in Algorithm \ref{algo:algo-2} will make $\LAD{mG}$ steps. On the other hand, $\LAD{mG} = \Theta(\vert S \vert)$. In each step, we compute the $nme$s, consuming $O(\vert N \vert  \cdot t(\textbf{NASH})$ time.  A single $nme$ $\sigma$ may have at most $\vert N \vert$ position vectors, i.e. $\vert \chi(\sigma) \vert \leq \vert N \vert$. For each position vector we perform a equality check, consuming time $O(\vert N \vert^2 \cdot \vert S \vert)$.
\end{proof}

%%%%%%%

\end{document}